\documentclass[letterpaper]{article}
\usepackage{aaai2027}
\usepackage[hyphens]{url}
\usepackage{graphicx}
\usepackage{natbib}
\usepackage{caption}

\usepackage{algorithm}
\usepackage{algorithmic}

\usepackage{amsmath}
\usepackage{amssymb}
\usepackage{mathtools}
\usepackage{amsthm}

\usepackage{cleveref}

\usepackage{xcolor}

\theoremstyle{plain}
\newtheorem{theorem}{Theorem}[section]

\newtheorem{lemma}[theorem]{Lemma}

\theoremstyle{definition}
\newtheorem{definition}[theorem]{Definition}

\theoremstyle{remark}

\usepackage{newfloat}
\usepackage{listings}
\DeclareCaptionStyle{ruled}{labelfont=normalfont,labelsep=colon,strut=off}
\floatstyle{ruled}
\newfloat{listing}{tb}{lst}{}
\floatname{listing}{Listing}

\usepackage{booktabs}

\nocopyright

\title{A Theoretical Framework for Parallel Lifelong MAPF Using \\ Group Decentralized Planning}
\author{
	Alex DeWeese \textsuperscript{\rm 1}, Jiaoyang Li \textsuperscript{\rm 2}, Guannan Qu \textsuperscript{\rm 1}\\
}
\affiliations{
    \textsuperscript{\rm 1}Department of Electrical and Computer Engineering, Carnegie Mellon University, Pittsburgh, PA 15213, US\\
    \textsuperscript{\rm 2} Robotics Institute, Carnegie Mellon University, Pittsburgh, PA 15213, US\\
    \{mdeweese, gqu\}@andrew.cmu.edu, \{jiaoyangli\}@cmu.edu\\
}

\begin{document}

\maketitle

\begin{abstract}

In the Lifelong Multi-Agent Path Finding (L-MAPF) problem, agents must repeatedly move from one destination to another while avoiding obstacles and inter-agent collisions. Widely regarded as one of the highest-performing solutions to this problem is the Rolling-Horizon Collision Resolution (RHCR) framework. However, commensurate with its quality solutions, it incurs a computational cost that limits its applicability to even modest agent counts.
In this paper, leveraging theoretical methods from the Locally Interdependent Multi-Agent MDP literature \cite{deweese2024locally}, we first theoretically prove the near-optimality of RHCR in a discounted MDP formulation of the L-MAPF problem.
Then, we leverage these results to naturally motivate an extended framework called Group Decentralized RHCR (GD-RHCR) which incorporates a group decentralized structure that partitions agents based on a transitive communication scheme and plans for each partition of agents in parallel. We show that both RHCR and GD-RHCR achieve similar exponentially close to optimal guarantees, establishing a theoretical duality between the time based restrictions performed by vanilla RHCR and the additional space based partitioning performed by GD-RHCR.
Lastly, we show that across varying maps, GD-RHCR is able to attain high throughput that scales into higher agent counts while maintaining a significantly lower per plan cost.

\end{abstract}

 \begin{links}
 \end{links}

\section{Introduction}

The (one-shot) multi-agent path finding (MAPF) problem is a highly studied area which navigates a large number of agents from some start points to end points on a map while avoiding obstacles as well as collisions among each other. Applications are broad, encompassing warehouse logistics \cite{ma2017lifelong}, airport logistics \cite{li2019departure}, UAV traffic management \cite{ho2019multi}, parking navigation \cite{okoso2019multi}, video games \cite{li2020moving} etc.

When agents are continuously moving such as in automated package delivery in warehouses, this shifts to a new paradigm called \textbf{Lifelong MAPF} (L-MAPF). In this setting as agents reach a destination, they are immediately assigned a new destination. They must attain a high throughput (e.g. delivered packages) while avoiding collisions with obstacles and with other agents.

A wide range of methods have been proposed for this lifelong setting (see \cref{related_works}). Of these, one important baseline is Priority Inheritance with Backtracking (PIBT), an ultra fast and scalable method that serves as a strong baseline but can have greedy behavior. In contrast to PIBT, another popular solution is the Rolling-Horizon Collision Resolution (RHCR) framework which is seen as one of the highest performing solutions in terms of attaining high throughput.

Unfortunately, a weakness of RHCR is its high computational cost especially when compared to PIBT. In practical scenarios like robot navigation in warehouses, RHCR must repeatedly use one-shot MAPF solvers in a short amount of time. Even when using suboptimal solvers like Priority Based Search (PBS), we observe that RHCR explodes exponentially in computational cost.
In fact, within the research community, running RHCR as a baseline is seen as a computationally expensive part of running simulations and often must be reported as a partial plan before timeouts hit or run on a smaller subset of problems \cite{arita2026lifelong, chen2024traffic}.
In practice, RHCR can be used with a fallback mechanism where another algorithm (such as PIBT) is taken when some timeout is hit.

While it is tempting to speed up RHCR by developing parallelized versions of it (across agents), it remains open how to do so in a theoretically principled way without sacrificing the performance of RHCR.

\subsection{Contributions}

The goal of this work is to \emph{develop theoretically principled parallelized RHCR methods that keep (or exceeds) the high-throughput of RHCR, and plans with a fraction of the compute cost.
} Towards this goal, we build a \emph{novel theoretical foundation for RHCR}, which leads to a theoretically principled parallelization of RHCR called \emph{Group Decentralized RHCR (GD-RHCR)}, as detailed below.

\textbf{Theoretical foundation of RHCR.} We start with a theoretical underpinning of vanilla RHCR as a near-optimal method when modeling the MAPF problem as a discounted MDP using newly established techniques from the Locally Interdependent Multi-Agent MDP (LI-MDP) literature \cite{deweese2024locally, deweese2025thinking}. That is, in this discounted setting, the performance of RHCR converges to optimal exponentially fast with the increase in the planning horizon (\cref{rhcr_guarantee}). This worst-case bound is agnostic to the replan window (the frequency of reevaluation), suggesting a theoretical motivation for this previously existing concept.

\textbf{Group Decentralized RHCR.} The theoretical foundation also motivates a Group Decentralized (GD) version (see \cref{gd_def}) to the RHCR framework. We show that the new GD-RHCR framework satisfies similar theoretical guarantees as RHCR up to constant factors
(\cref{gdrhcr_guarantee}),  establishing a theoretical duality between the time based restrictions performed by RHCR and the group decentralized distance based partitioning scheme performed by GD-RHCR. These theory results mean GD-RHCR \textit{keeps the high-performance advantage} of RHCR.

In addition to the theoretical guarantee, GD-RHCR has a variety of desirable empirical characteristics that significantly \textit{reduce the computational cost} compared to RHCR.
(1) The time complexity of near-optimal one-shot MAPF algorithms (e.g., CBS and PBS) grows exponentially with the number of agents. By partitioning agents into groups and solving each group independently, GD-RHCR significantly reduces the complexity of the planning problem.
(2) GD-RHCR enables replanning for different groups to be performed in parallel and even asynchronously. In particular, only a subset of the groups needs to be replanned at each timestep, reducing the computational burden of online planning.
(3) GD-RHCR allows different one-shot MAPF planners to be assigned to different groups. For example, more optimal solvers such as PBS can be used for small groups, while more scalable planners such as PIBT can be used for large groups, providing a natural trade-off between solution quality and computational efficiency.

Because of these strengths, for small agent counts that RHCR can already handle, GD-RHCR achieves similar throughput in many cases with significantly faster planning speed; for large agent counts that only PIBT can handle, GD-RHCR can also solve them yet with consistently higher throughput.

\textbf{Experimental Validation. }
We demonstrate GD-RHCR empirically across a variety of maps and show that the average plan time of can be reduced by a factor of $24.9$x and match RHCR in performance before its collapse. Further, it performs well deeper into the large agent range and shows up to a $57.7\%$ improvement in the throughput against PIBT and collapsed RHCR.

\section{Related Works}
\label{related_works}
\textbf{One-Shot MAPF}
One-shot MAPF solvers are widely put into two categories: Theoretically guaranteed versus greedy
methods. The theoretically guaranteed methods provide optimal or near-optimal solution at a potential of high computation time. Some algorithms include Conflict Based Search (CBS) \cite{SHARON201540}, Enhanced CBS (ECBS) \cite{barer2014suboptimal}, Explicit Estimation CBS (EECBS) \cite{li2021eecbs}. However, since solving MAPF optimally is NP-Hard \cite{yu2013structure}, these methods can tend to blow up in computation time especially as the number of agents increase.

For greedy methods, there are priority based planners such as Priority Planning \cite{erdmann1987multiple} and Priority Based Search (PBS) \cite{ma2019searching}. These methods can be fast
but forgo the theoretical near-optimality. In a similar vein, there are ultra-fast but greedy methods that rely on Priority Inheritance with Backtracking (PIBT) \cite{okumura2022priority} and its variants such as LaCAM \cite{okumura2023lacam, okumura2023improving}. These methods are extremely scalable but may be suboptimal (see \cref{simulations}).

There are also learning based methods that use reinforcement learning or imitation learning directly to solve the MAPF problem \cite{sartoretti2019primal}. These methods can perform some parallelism through decentralization, yet there is no guarantee on how much performance is lost due to parallelization/decentralization.

\textbf{Lifelong MAPF}
Largely, this literature is divided into the the slow but high quality solutions provided by RHCR and the fast but greedy solutions like PIBT. Recall, RHCR enables conversion of the one-shot MAPF solvers to the lifelong setting but may incur a large computation time since one-shot MAPF is NP-hard \cite{morag2023adapting}.

On the other hand, PIBT is a method that can be used in the one-shot context or the lifelong setting and it along with its variations \cite{chen2024traffic, okumura2019winpibt} remain as a fast alternatives that are greedy and can often suffer from deadlocks when the graph is not biconnected \cite{okumura2022priority}.

Analogously we have multi-agent RL and imitation learning methods for this setting as well \cite{damani2021primal, jiang2025deploying}.

\textbf{Grouping in MAPF}
There have been a number of prior works that attempt to speed up computation by grouping agents. A static grouping method can divide the map into sub-regions and treat agents in each sub-region as a group \cite{leet2022shard}. In another context, more sophisticated methods attempt to dynamically identify independent groupings \cite{veerapaneni2025windowed, zhang2026dynamic}. This however is done in the context of speeding up an algorithm that maintains completeness guarantees rather than identifying coordination groups that may perform well together.

\textbf{Locally Interdependent Multi-Agent MDP}
Adjacent to the MAPF literature, there has been a recent breakthrough in the study of discounted multi-agent MDP environments with dynamic local dependencies.
A model was proposed called the Locally Interdependent Multi-Agent MDP (LI-MDP) which is a theoretical model for multi-agent systems with local interactions \cite{deweese2024locally, deweese2025thinking}.

Further the same works proposed the group decentralized setting which connects agents transitively based on their visibility and connected agents are allowed to coordinate. This special observability structure in between the centralized and decentralized setting permits strong theoretical guarantees not available to the decentralized setting (exponentially close to optimal with respect to visibility) but still allows for groups of agents to plan in parallel.

We show that the LI-MDP framework brings insight into the effectiveness of existing methods like RHCR to prove near optimality and can be used to propose new methods like GD-RHCR that integrates this new observability structure.

\section{Motivation: Natural Extension to RHCR}
\label{motivation}

 \begin{figure}
     \centering
     \includegraphics[width=0.85\linewidth]{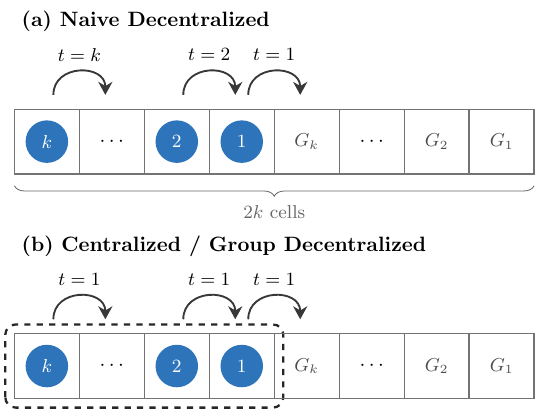}
     \caption{$k$ agents with goal locations $G_i$ on the right. Arrows depict timestep of first movement.}
     \label{cent_dec}
 \end{figure}

To motivate our method, we begin by introducing an abstracted version of RHCR. To stay consistent with traditional MDP notation, we will use $H$ to refer to the planning horizon and $\alpha$ will refer to the replanning period.

\textbf{RHCR:} Assume a one-shot solver \textit{SOLVER(sources, destinations, H)} that produces a $H$ step collision-free path for all agents and as $H \rightarrow \infty$, the path moves all agents from their sources to destinations.
The RHCR framework proposes to run \textit{SOLVER} for
some $\alpha\leq H$ timesteps before recomputing and repeating. That is, at each $t = k \alpha$, RHCR will run $SOLVER(positions[k\alpha], remaining\_destinations[k\alpha], H)$ and that path will be taken for $t' \in [k\alpha, k\alpha + \alpha)$.

In this paper, we seek to mitigate the computational requirements of repeatedly using $SOLVER$ in RHCR by running $SOLVER$ on groups of agents in parallel. Our observation is that we can view RHCR as \emph{temporally} ignoring agent interactions beyond some time horizon $H$ to reduce computation time. Therefore, if we want to achieve parallelization, a natural idea is also to ignore interactions but now based on spatial distance (e.g. with some threshold $\mathcal{V}$) -- this way, the planning can be broken up into parallel groups where inter-group interactions have been intentionally ignored. The question is then,  \emph{(*) what is the ``correct way'' to ignore interactions based on distance to achieve parallelization without sacrificing the performance of RHCR?}

\textbf{Failure of Naive Decentralization:} Naively, it appears that the conventional fully decentralized method, where all agents are computed independently in parallel using the information within visibility $\mathcal V$ of the agents (without communication), would be the answer to the question.
However, in \cref{cent_dec}, we show a counter example.
Notice that a centralized scheme could move the agents to their goals in $k$ steps but for a decentralized scheme it would take at least $2k$ timesteps because agents cannot move before the agent in front has moved (at least $k$ iterations for agent $k$ to move its first step).

Therefore, increasing $\mathcal V$ for the independent decentralized setting is clearly not analogous to increasing the time horizon $H$ in $SOLVER$, which would approach the optimal solution as $H \rightarrow \infty$.

\textbf{Group Decentralized Setting:} Fortunately, it turns out that L-MAPF problem can be modeled as a recently proposed Locally Independent Multi-Agent MDP (LI-MDP) model \cite{deweese2024locally}, which also presents a theoretically grounded method to divide the agents into parallel groups without sacrificing much performance. The main solution concept is the group decentralized setting.

\begin{definition}[group decentralization]
    Given visibility radius $\mathcal{V}$ and a distance metric, for agents $\mathcal N$ treated as vertices, create a graph there is an edge between two agents if they are within distance $\mathcal V$. We define groups as the connected components of the graph. Each group is planned independently.
    \label{gd_def}
\end{definition}
\noindent Notice this resolves the issues faced in \cref{cent_dec}. See \cref{gd_empty} for a visual of a group decentralization partitioning.

In \cref{rhcr_theory}, we will formally formulate L-MAPF as a discounted LI-MDP problem and provide a theoretical guarantee for RHCR. Then, we propose our full extension to the RHCR framework that uses this group decentralized setting (GD-RHCR) in  \cref{gdrhcr}, then formally prove this new framework non-trivially satisfies a similar guarantee to RHCR in \cref{gdrhcr_theory}. We show the performance of this method in practice in \cref{simulations}.

\begin{figure}
    \centering
    \includegraphics[width=0.4\linewidth]{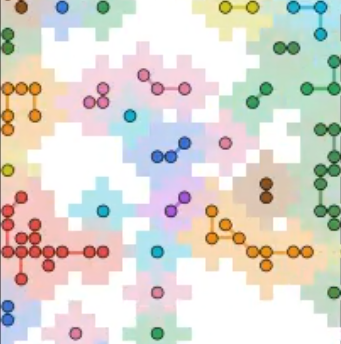}
    \caption{Group decentralized partitioning for $\mathcal V = 2$.  Groups are connected by lines and color coded (colors are reused).}
    \label{gd_empty}
\end{figure}

\section{Theory: LI-MDP Analysis of RHCR}
\label{rhcr_theory}

In this section we will set up the L-MAPF as a discounted reward LI-MDP model. RHCR when viewed in a LI-MDP environment will turn out to be a near optimal method.

\textbf{Conventional L-MAPF Model:} $n$ agents are placed in a grid environment with open spaces $\mathcal X$. At each step every agent $i$ may use movements $\mathcal D = \{UP, DOWN, LEFT, RIGHT, STAY\}$ to transition from $x_i^{prev}\in \mathcal X$ to the same or an adjacent space $x_i \in \mathcal X$ if no other agent moves to the space at the same time $x_i \neq x_j\forall j \in \mathcal N$  (vertex collision) and agents do not interchange locations $\neg(x_i^{prev} = x_j \wedge x_j^{prev} = x_i),  \forall j \in \mathcal N$ (edge collision). Each agent $i$ is assigned a list of goal locations $g_i^1, g_i^2, \ldots$ to be navigated to in sequence and the metric of evaluation is number of goals achieved per unit time referred to as the throughput.

\textbf{L-MAPF as Multi-Agent MDP:}
To model the L-MAPF problem an MDP with $n$ agents, we will assume the state space for every agent $i$ is $\mathcal S_i = \mathcal X \times D \times G$ where $\mathcal D$
is interpretted as the direction the agent came from
and $G \subset \mathcal X$ is a set of goal states. This creates the joint state $\mathcal S = \mathcal S_1 \times \mathcal S_2 \ldots\times \mathcal S_n$. The action space for every agent $i$ will be the standard MAPF grid movements $\mathcal A_i = \mathcal D$ with joint action space $\mathcal A = \mathcal A_1 \times \mathcal A_2, \ldots \mathcal A_n$. The transition function
will transition an agent $i$'s state $(x_i, d_i, g_i)$ deterministically
with action $a_i$ by updating $x_i$ based on the action taken and adjacent positions in $\mathcal X$, update $d_i$ with the direction of their previous locations, and set a new goal whenever $x_i = g_i$ (keeping the same $g_i$ otherwise).
If the spaces are blocked by obstacles, they will remain in their current positions.

Lastly, the reward function will be the sum of a positive reward $r_{goal} > 0$ for every agent $i$ at their goal ($g_i = x_i$) and negative reward (penalty) $p_{collide} < 0$ for every agent $i$ involved in a vertex or edge collision. That is $r((x, d, g), a) = \sum_i \big( I[x_i = g_i]r_{goal} + \sum_j I[collision(x_i, x_j)]p_{collide}\big)$. Here $collision$ refers to any vertex or edge collisions where $x_j^{prev}$ is inferred from the current positions $x$ and the directions they came from $d$. We will assume $\lvert p_{collide}\rvert$ is much larger than $\lvert r_{goal}\rvert$. Without loss of generality, we will also assume $\lvert r\rvert \leq 1$ for bounded reward.

Our metric for evaluation will be $V^\pi(s) = \mathbb E_{\tau \sim \pi\lvert_s} [\sum_t \gamma^t r(s(t), a(t))]$ where $\gamma \in (0,1)$ and the trajectory $\tau \sim \pi\lvert_s$ is the trajectory $(s(t), a(t))$ starting at $s(0) = s$ taking the policy $\pi: \mathcal S \rightarrow \Delta(\mathcal A)$ at each timestep. This discounted model is helpful for the lifelong setting as the sum of $r_{goal}$ terms may diverge as $t \rightarrow \infty$ but the discounted rewards will remain finite. For analysis purposes, for a finite horizon policy $\pi_{finite}(s) = \{\pi_0, \pi_1, \ldots, \pi_{H - 1}\}$ we will define the finite horizon value function $V_h^{\pi_{finite}}(s) = \mathbb {E}_{\tau \sim \pi_{finite}\lvert_s} [\sum_{t = h}^{H - 1} \gamma^{t - h} r(s(t), a(t))].$

Denote $\pi^*$ as the optimal stationary infinite horizon policy and the optimal discounted finite horizon policy for horizon $H$ as $\pi_{finite}^* = \{\pi^*_0, \pi^*_1, \ldots, \pi^*_H\}$. We will denote $V^*(s) = V^{\pi^*}(s)$ and $V_h^*(s) = V^{\pi^*_{finite}}_h (s)$.

\textbf{L-MAPF as LI-MDP:}
It turns out that these assumptions made on the multi-agent MDP fit into the LI-MDP model \cite{deweese2024locally} except for a minor discrepancy with edge collisions. Agents move independently in a space with a distance metric (shortest path) and move at most 1 space at a time with local rewards and penalties for interactions. See \cref{limdp_fit} for a formal introduction to the LI-MDP and how this model is a near special case of LI-MDP.

Modeling the MAPF problem as a LI-MDP has two consequences (1) leveraging techniques from the LI-MDP theory, we prove a near optimality guarantee for RHCR  below (2) LI-MDP theory also suggests a near optimal group decentralized policy structure (see \Cref{gd_def}), which motivates our GD-RHCR framework.

\textbf{Near Optimality of RHCR:}
Under this model, if we assume that $SOLVER$ outputs a trajectory $\tau^{SOLVER}_s$ is $\epsilon$-near optimal $V^*_0 - \mathbb E_{\tau^{SOLVER}_s}[\sum_t \gamma^t r(s(t), a(t))] \leq \epsilon$ then,
\begin{theorem}
    Let $\tau^{RHCR}_s$ denote the trajectory generated by RHCR with planning horizon $H$ and replan window $\alpha \in [1,H]$ starting at the state $s$. Let $V^{RHCR}(s)$ the corresponding sum of discounted rewards for $\tau^{RHCR}_s$. Then, the following result holds:

    \hspace{7ex}$V^*(s) - V^{RHCR}(s) \leq \frac{2\gamma^H}{(1 - \gamma)^2} + \frac{\epsilon}{1 - \gamma}$
    \label{rhcr_guarantee}
\end{theorem}

Aside from the error from the $\epsilon$-optimality of $SOLVER$, we can see that the suboptimality decreases exponentially fast with the increase in $H$. Also, perhaps counterintuitively, we see that the guarantee is independent of the replan window $\alpha$ establishing a theoretical justification for using a replan window to improve on computation time from \cite{li2021lifelong}. Note however, that a smaller $\alpha$ can be still be useful in practice to reduce $V^*(s) - V^{RHCR}(s)$ within $[0, \frac{2\gamma^H}{(1 - \gamma)^2} + \frac{\epsilon}{1 - \gamma}]$.

\textbf{Why group decentralized?}
It is shown in \citet{deweese2024locally} that group decentralized policies can also satisfy these type of near optimality bounds, so we expect a group decentralized version of RHCR (\cref{algorithm}) to satisfy a similar guarantee (see \cref{gdrhcr_theory}). The primary intuition will be that agents within different groups theoretically cannot collide within $\lfloor \frac{\mathcal V}{2}\rfloor$ steps (see \cref{dtl}). This means that agents in different groups can safely be ``ignored'' in the computation when considering a finite horizon. This will also be the explanation to why the group decentralized scheme is the answer to question (*) posed in \cref{motivation}.

After introducing our algorithm in the following section, we show in \cref{gdrhcr_theory} the theoretical results for our group decentralized method.

\renewcommand{\algorithmiccomment}[1]{\hfill // #1}
\begin{algorithm}[t]
\caption{GD-RHCR}
\label{algorithm}
\begin{algorithmic}[1]
\STATE $\lambda_i \gets \varnothing$ for each agent $i \in \mathcal{N}$
\FOR{$t = 0, 1, 2, \dots$}
  \STATE $\mathcal{G} \gets \textsc{Connected Components}(\mathcal{N}, \mathcal{V})$\label{bfs_line}
  \STATE $\bar{\lambda} \gets \lambda:=(\lambda_i)_{i\in\mathcal{N}}$\label{line_copy}
         \COMMENT{snapshot end-of-previous-step plans}
  \FORALL{groups $g \in \mathcal{G}$ \textbf{in parallel}}
    \IF{$|g| \ge K_{\text{th}}$}
      \STATE $P \gets \textsc{Solver2}(g)$\label{line_early_fallback}
             \COMMENT{early termination}
      \STATE $\lambda_i \gets P$ for all $i \in g$
    \ELSIF{$t = 0$ \OR $\exists\, i,j \in g:\ \lambda_i, \lambda_j$ not from same solve
           \OR $(\lambda_i)_{i\in g}$ exhausted  }
      \STATE $C \gets \textsc{SoftConstraints}(\{\bar{\lambda}_k : k \notin g\},\ c_{\text{soft}})$\label{line_soft}
      \STATE $P \gets \textsc{Solver}(g,\ H,\ C,\ \tau)$\label{line_solver}
             \COMMENT{lazy recompute}
      \IF{$P = NULL$}
        \STATE $P \gets \textsc{SOLVER2}(g)$\label{line_timeout_fallback}
               \COMMENT{solver timeout}
      \ENDIF
      \STATE $\lambda_i \gets P$ for all $i \in g$
    \ELSE
      \STATE \textbf{continue}
             \COMMENT{all agents agree: reuse cached plan}\label{line_continue}
    \ENDIF
  \ENDFOR
  \STATE advance each $i \in \mathcal{N}$ one step along $\lambda_i$
\ENDFOR
\end{algorithmic}
\end{algorithm}
\section{Our Method: Group Decentralized RHCR}
\label{gdrhcr}

Our algorithm will consist of four components and is summarized in \cref{algorithm}. Firstly, we will identify the groups from the group decentralized scheme described earlier. Secondly, our algorithm will have a ``as needed'' lazy evaluation scheme which will dynamically adjust the reevaluation window $\alpha$ for groups of agents.  Next, to improve the accuracy of our parallel computation across groups, we will add soft constraints for out of view agents in other groups in our lower level planner. Lastly, the algorithm will include a fast secondary solver to rapidly solve large congestion groups.

\textbf{Group Decentralized:}
We begin by identifying the partitions in the group decentralized setting by identifying the connected components as in definition \ref{gd_def}.
This is performed at each timestep shown in line \ref{bfs_line} of the algorithm.

In practice, we use the shortest path distance and the union find data structure to hold the connected components. The shortest path will become longer near obstacles, allowing us to take advantage of map topology (such as warehouse-10-20-2-1 in \cref{simulations}). Agents remain sparser and the grouping scheme is more likely to break apart into smaller groups.

\newlength{\plotwidth}
\newlength{\plotheight}
\setlength{\plotwidth}{0.12\linewidth}
\setlength{\plotheight}{2.2cm}

\newlength{\mapwidth}
\newlength{\mapheight}
\setlength{\mapwidth}{0.4\linewidth}
\setlength{\mapheight}{2.2cm}

\newcommand{\mapentry}[2]{
  \begin{minipage}[t]{\plotwidth}\centering
    \includegraphics[width=\mapwidth,height=\mapheight]{#1}\\[-0.1em]
    {\tiny #2}
  \end{minipage}
}
\newcommand{\plotentry}[1]{
  \begin{minipage}[t]{\plotwidth}\centering
    \includegraphics[width=\linewidth,height=\plotheight]{#1}
  \end{minipage}
}
\newcommand{\ylabeltp}{\rotatebox{90}{\hspace{2ex}\small   package throughput }}
\newcommand{\ylabelrt}{\rotatebox{90}{\qquad\small  plan runtime [s]}}
\newcommand{\ylabelwt}{\rotatebox{90}{\small wall time [s]}}

\begin{figure*}[!t]
  \centering
  \setlength{\tabcolsep}{1pt}
\setlength{\plotheight}{3.1cm}
  \setlength{\plotwidth}{0.20\linewidth}
  \begin{tabular}{@{}c@{}cccc@{}}
    &
    \setlength{\mapwidth}{0.21\linewidth}
    \setlength{\mapheight}{2cm}
    \mapentry{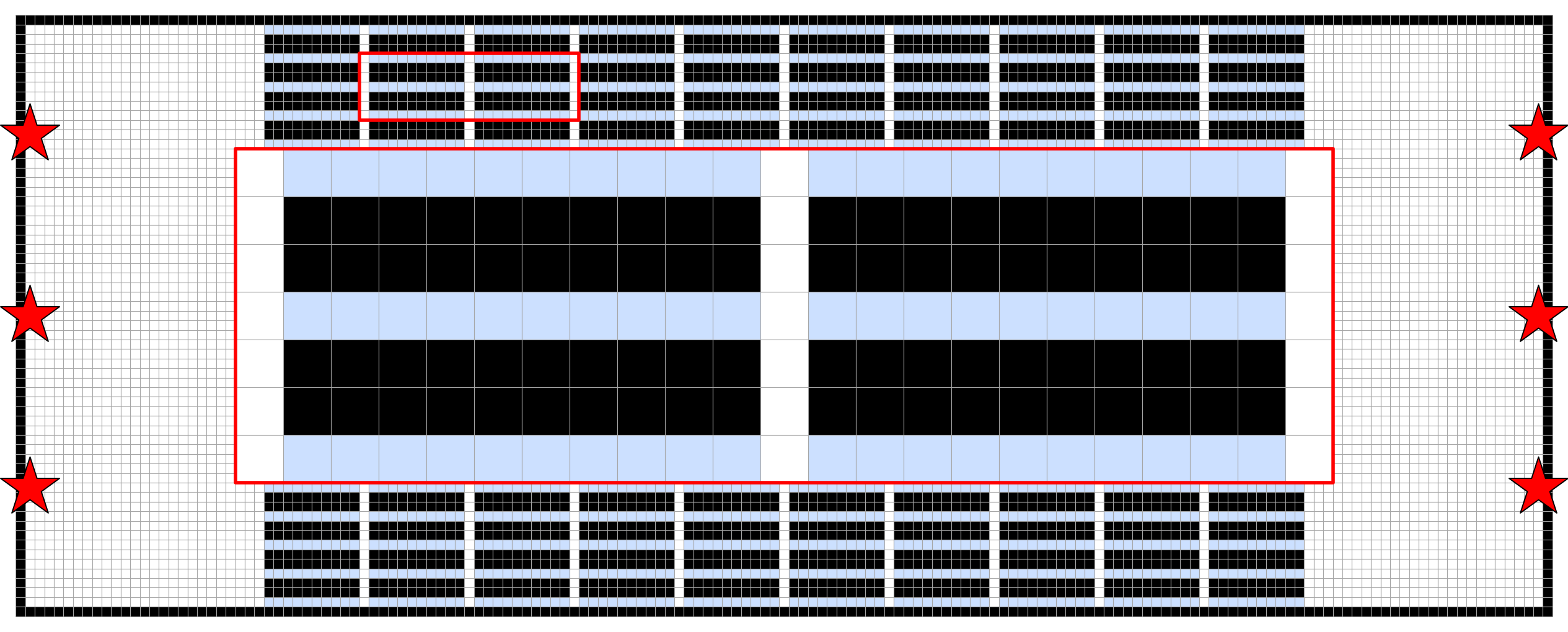}{warehouse-10-20-10-2-1} &
    \setlength{\mapwidth}{0.21\linewidth}
    \setlength{\mapheight}{2cm}
    \mapentry{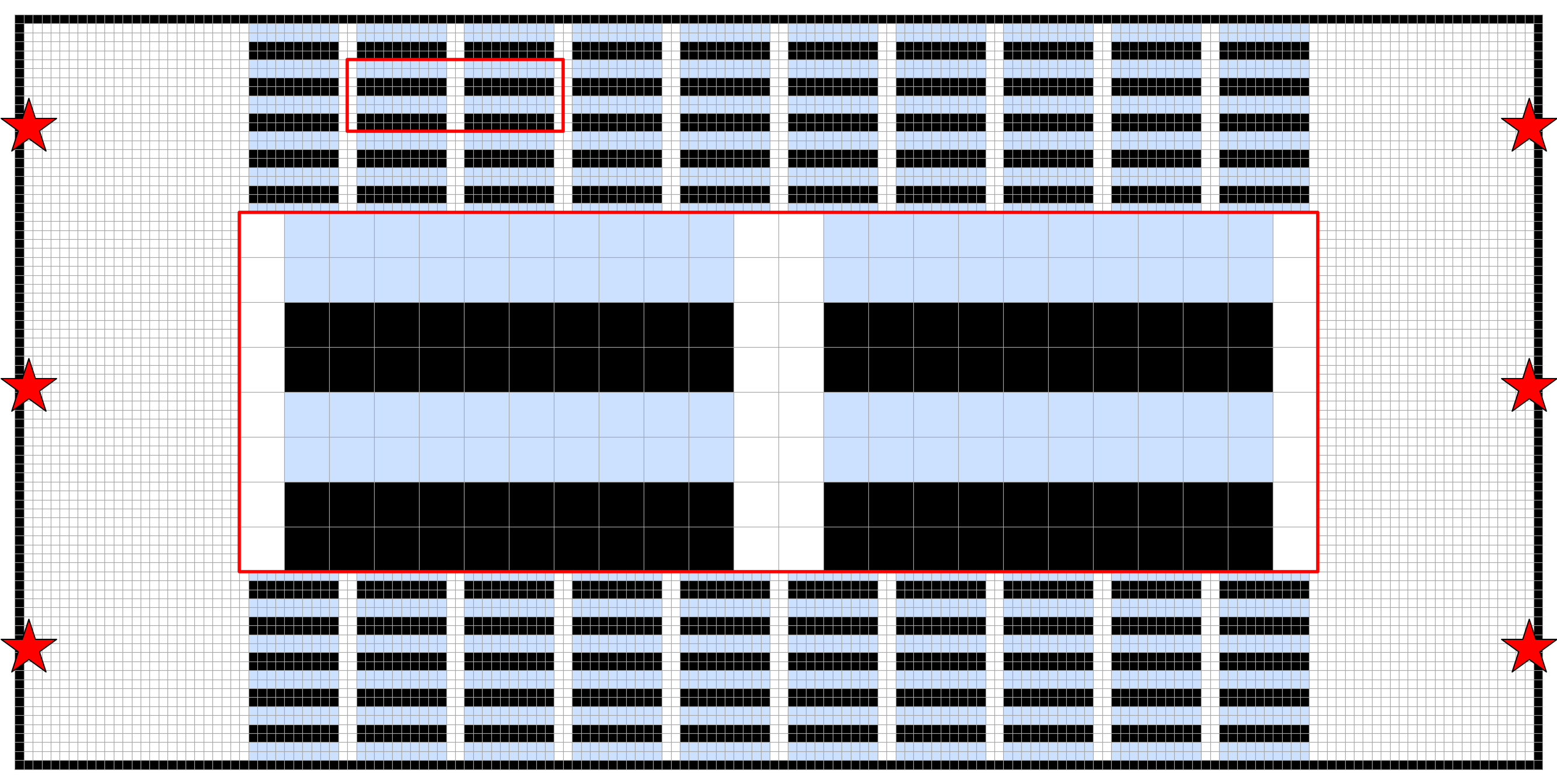}{warehouse-10-20-10-2-2} &
    \setlength{\mapheight}{2cm}
    \setlength{\mapwidth}{0.19\linewidth}
    \mapentry{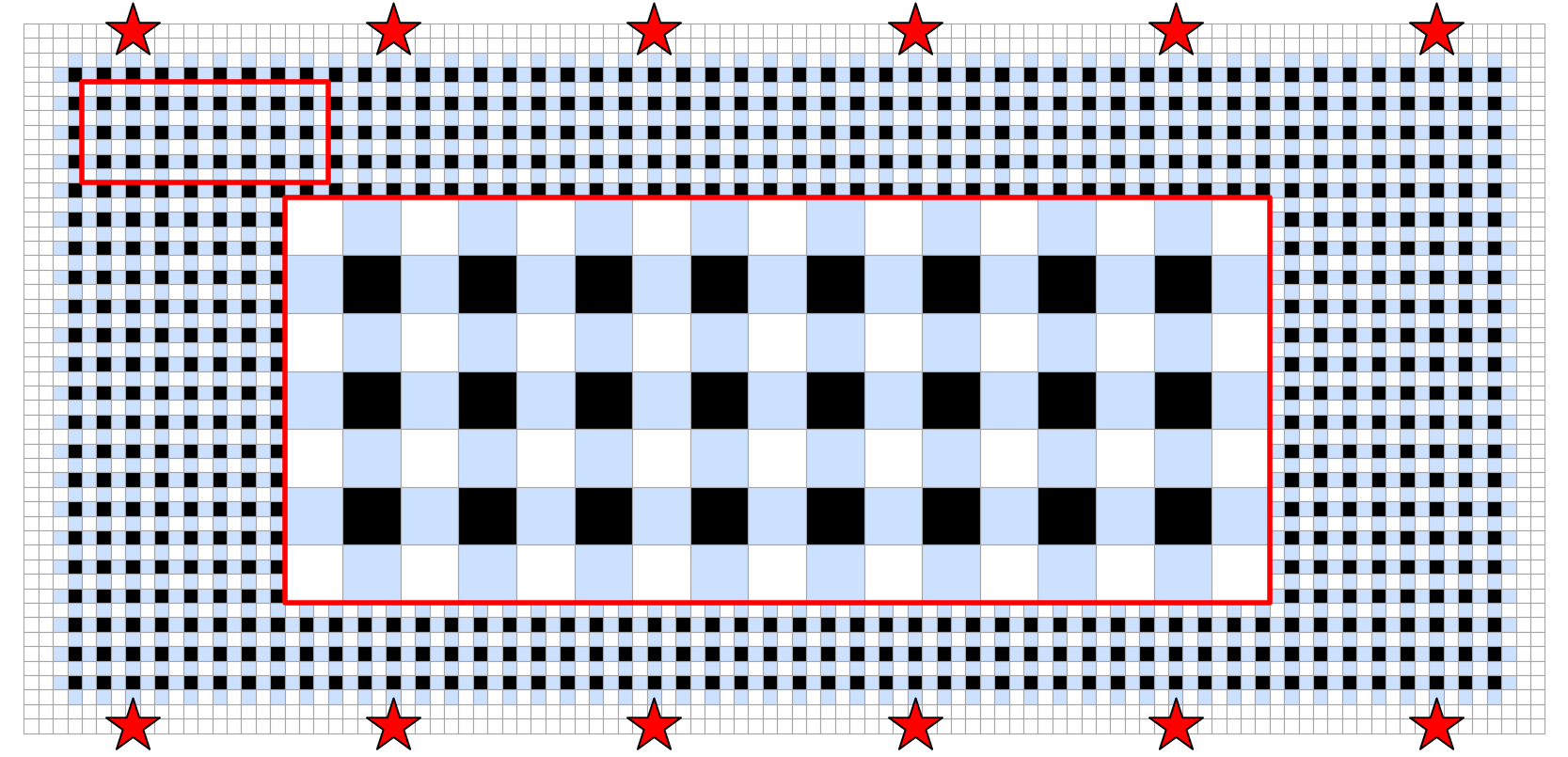}{sortation-1}&
    \setlength{\mapwidth}{0.21\linewidth}
    \setlength{\mapheight}{2cm}
    \mapentry{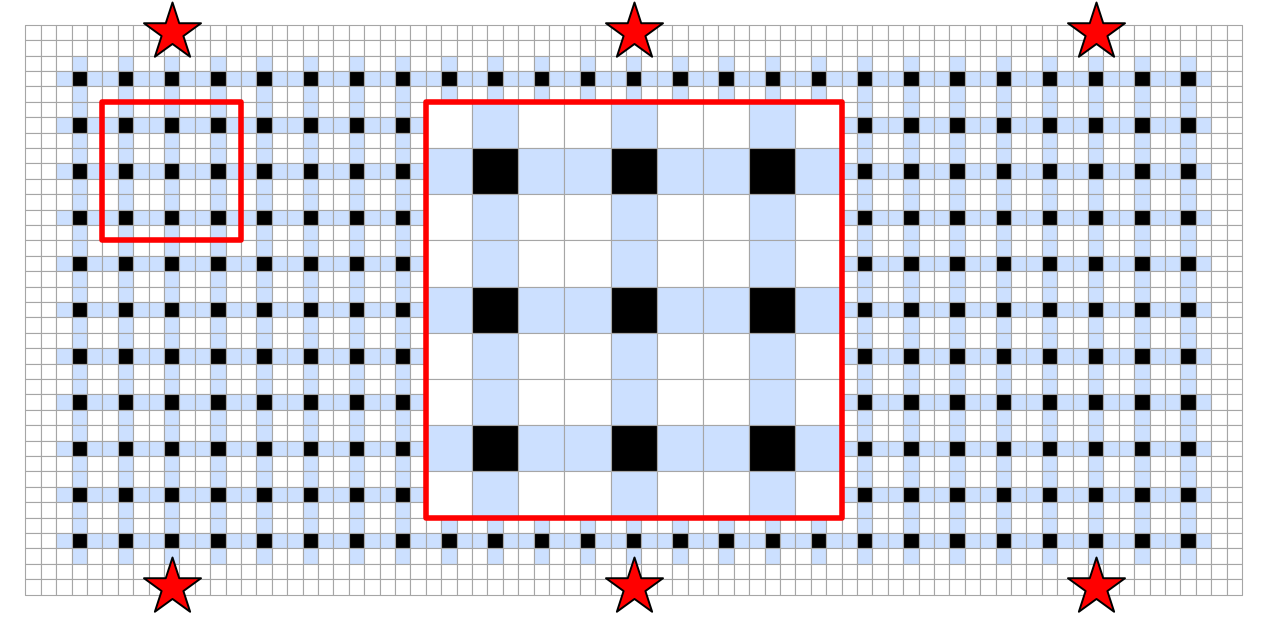}{sortation-2}
    \\[0.8em]
    \ylabeltp &
    \plotentry{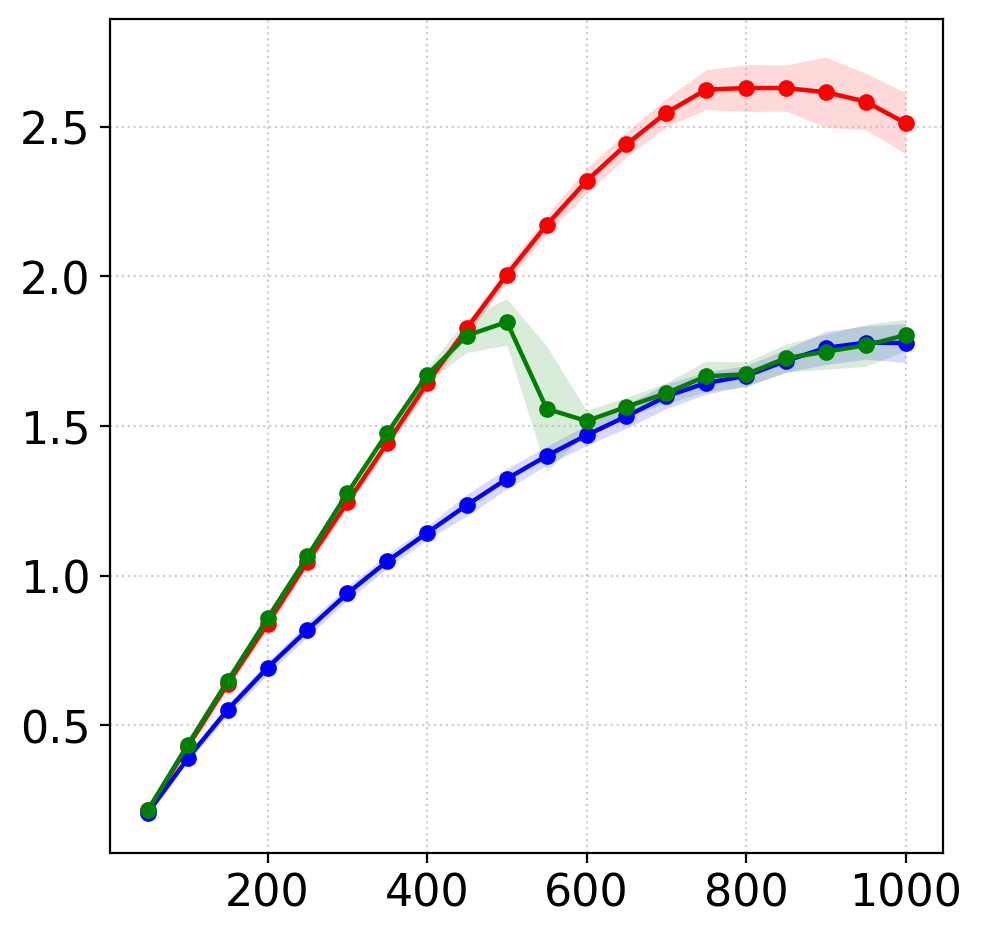} &
    \plotentry{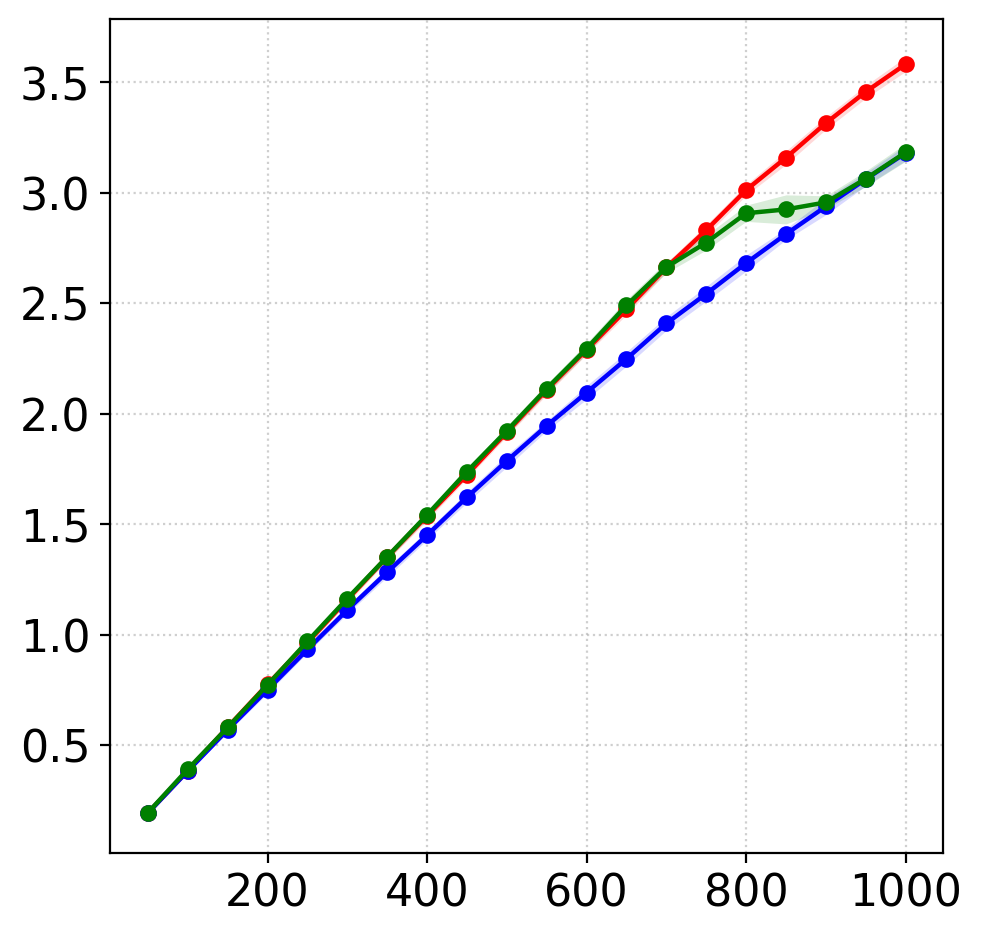} &
    \plotentry{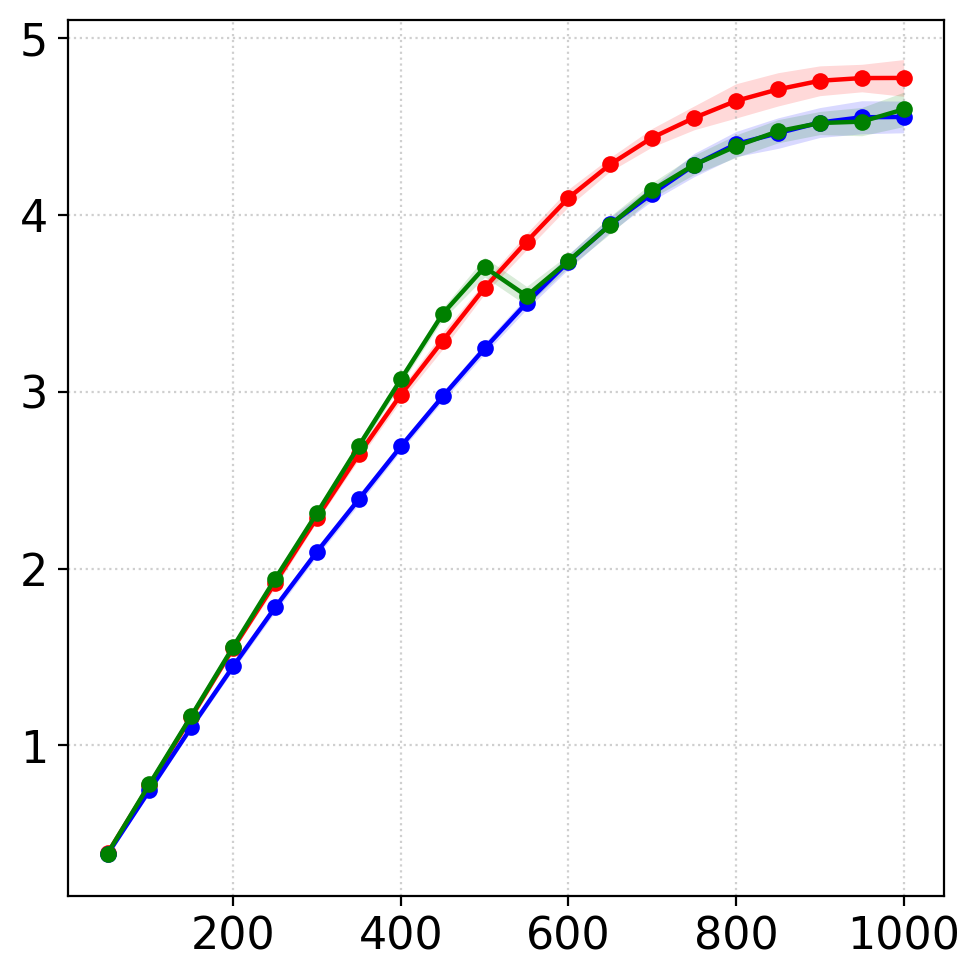} &
    \plotentry{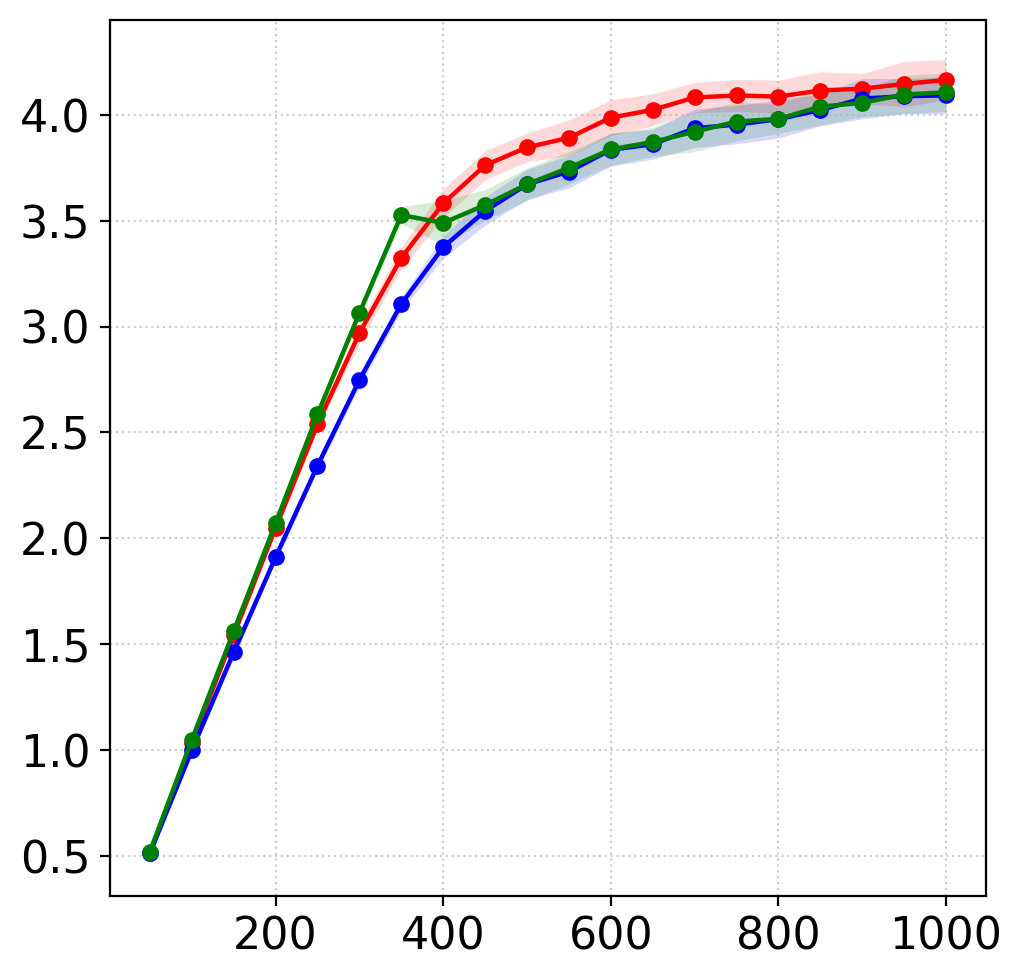}
    \\[0.3em]
    \ylabelrt &
    \plotentry{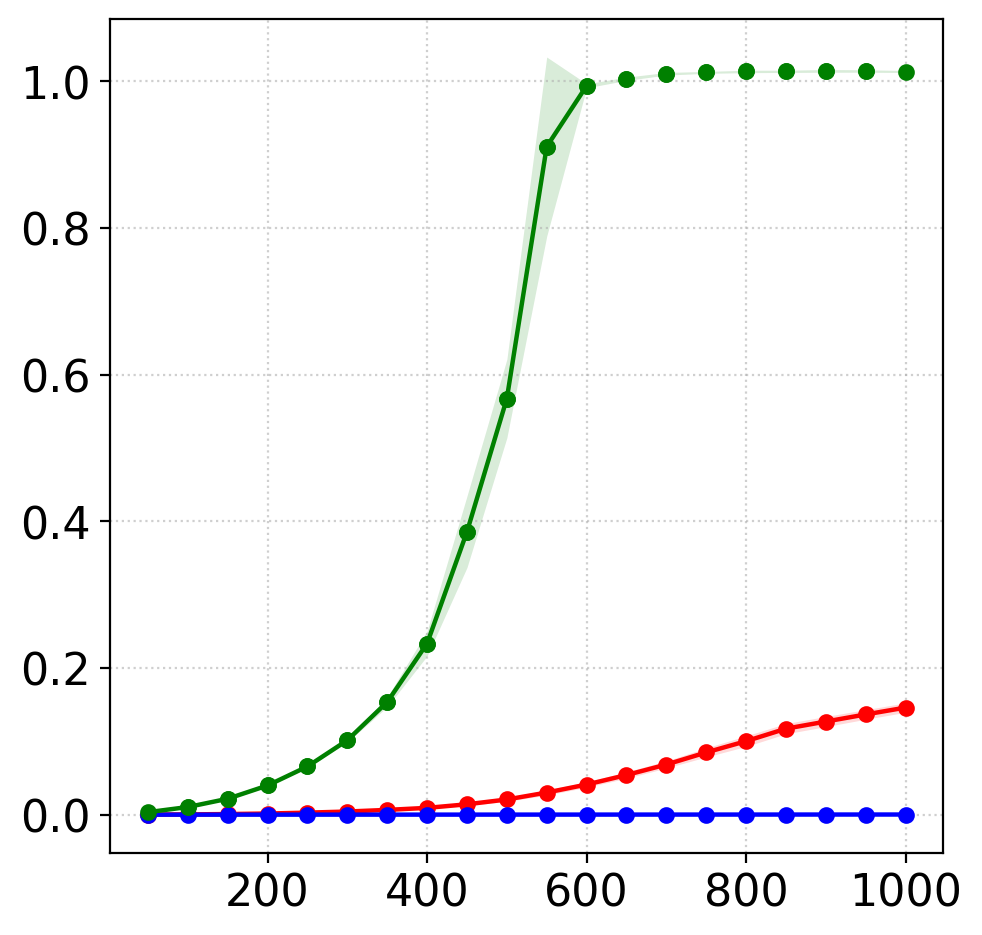} &
    \plotentry{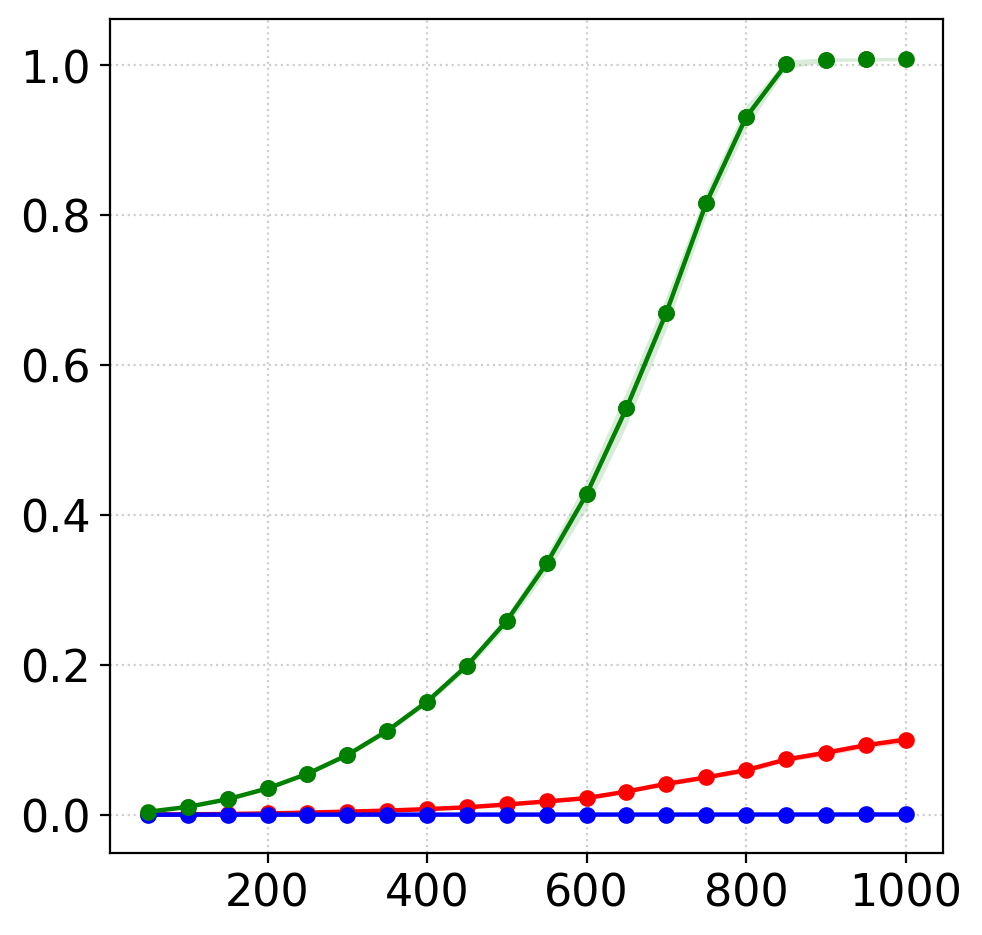} &
    \plotentry{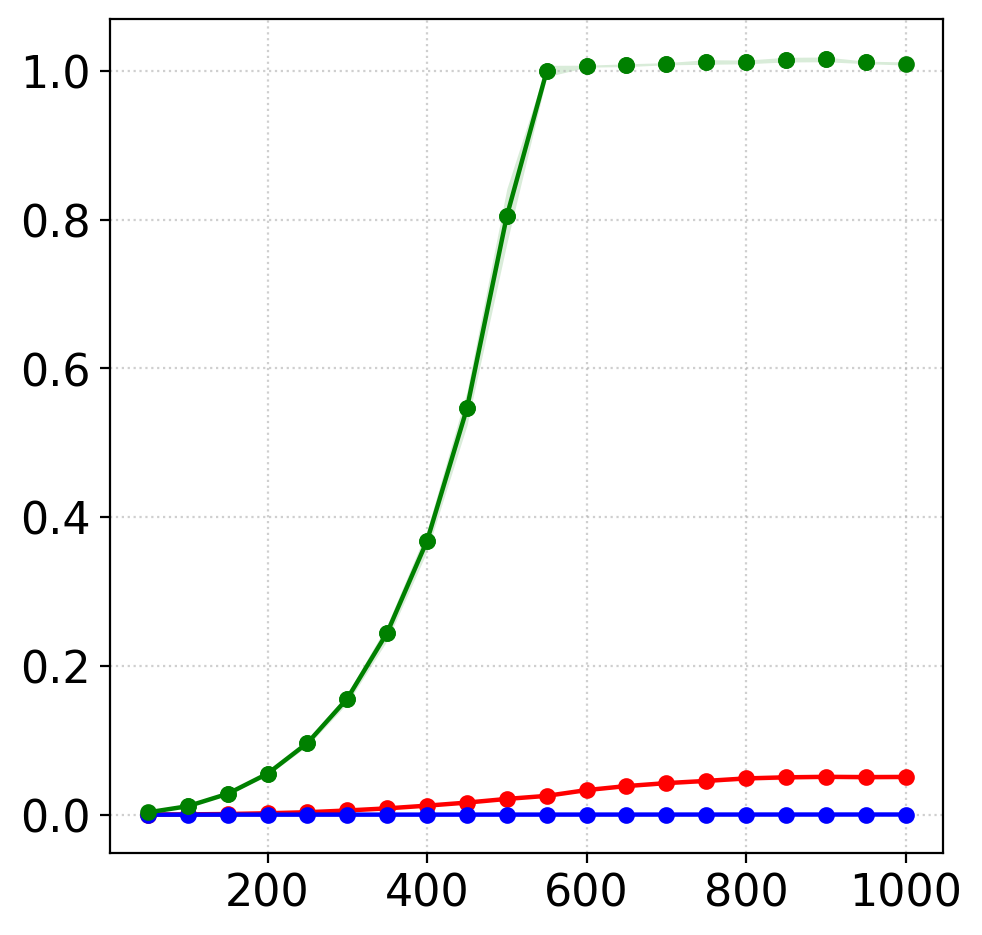} &
    \plotentry{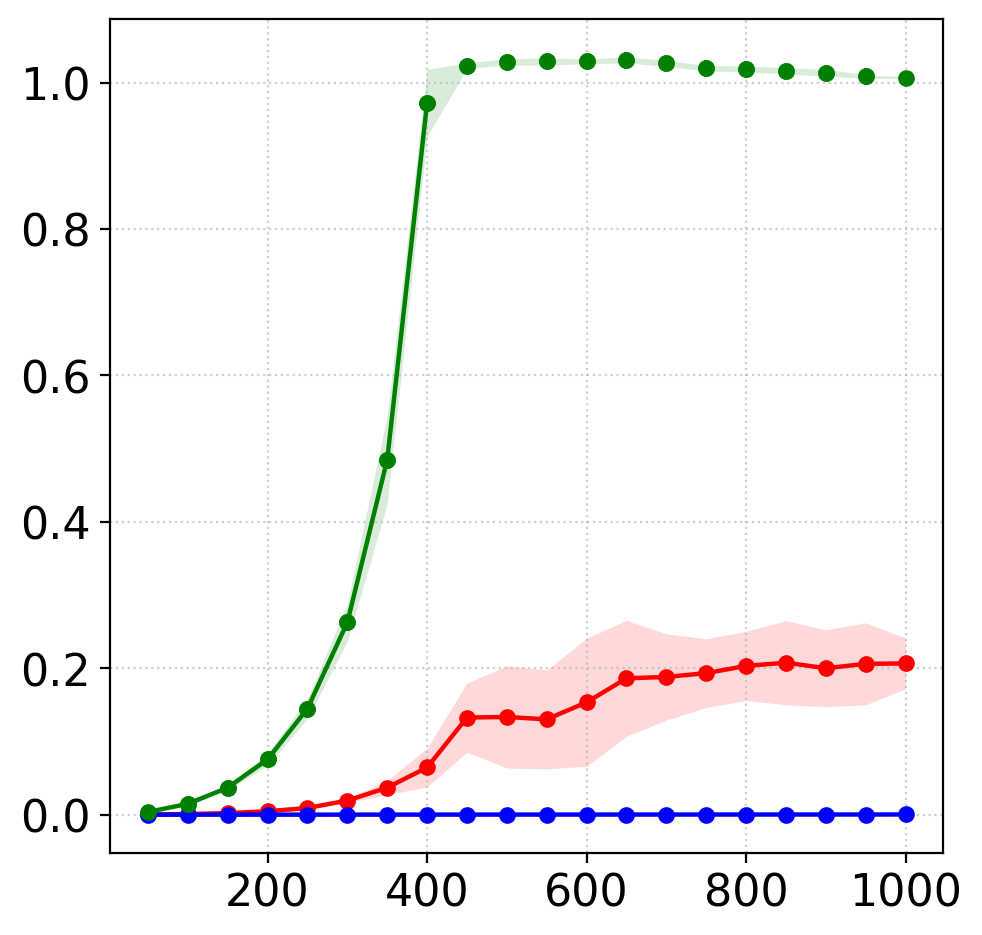}
    \\[0.3em]
    & \multicolumn{4}{c}{agents}
  \end{tabular}
  \caption{MAPD simulations. Agents move between red stars and blue squares. Color coding is \textcolor{red}{GD-RHCR}, \textcolor{green!50!black}{RHCR}, and \textcolor{blue}{PIBT}}
  \label{mapd_graphs}
\end{figure*}

\textbf{Lazy Evaluation:}
Now that the groups have been identified, we will now decide which groups will require computation. At the first timestep of the algorithm, we will run $SOLVER$ with some horizon $H$ on every identified group in parallel (line \ref{line_solver}). Agents are given a reference to the finite horizon ``plan'' $(\lambda_i)_{i \in g}$ computed for their group and move according to their plan. Then, at every subsequent timestep, if all of the agents in a group have references to the same solution, then the agents will act according to the referenced ``plan'' (line \ref{line_continue}). On the other hand, if the agents in the group disagree about the plan, the agents will recompute a new plan which all the agents will now reference (again line \ref{line_solver}). Intuitively, this triggers when a new agent unaccounted for in the plan has just joined the group.

Notice this lazy evaluation scheme can be interpereted as dynamically deciding on a replan window $\alpha$ for different subsets of agents. This can improve the the longevity of our often cheap parallel computations made across groups and staggers the evaluations which significantly reduces the \textit{per iteration} computational load compared to RHCR.

\textbf{Soft Constraints:}
When planning for a group, we will take the plans of all of the agents not within our group \textit{at the end of the previous timestep} (line \ref{line_copy}) and any collision with those plans caused during our compute in $SOLVER$ will introduce a cost $c_{soft}$, applied to the lower level search (line \ref{line_soft}). For example, our simulations we will use PBS with an efficient space-time A* implementation, SIPP, as the lower level planner, and a soft penalty $c_{soft}$ is applied to the lower level SIPP search every time the constraint is violated.

Plans from the previous timestep are used because the parallel computation across groups must begin computation before plans for other groups are settled. However, the staleness of those trajectories is mitigated with lazy evaluation, since not all groups are planning at the same time. Therefore, only the plans of a subset of agents which are planning are stale at any given time.

It is important that $c_{soft}$ remains relatively small as a large cost can increase the computation time of the lower level planner (e.g. SIPP) that is called frequently (see our ablation simulation in appendix \ref{abalations}). In fact, if we use hard constraints ($c_{soft} = \infty$), we may have no solution as this would effectively give all agents in other groups higher priority than agents we are planning for.

\textbf{Secondary Solver:}
Since the large computation time for RHCR can often be caused by a small number of congested regions, we may use our group decentralized mechanism to use different algorithms in different locations. For our algorithm, if the group size reaches some threshold $K_{threshold}$, we will immediately use a fast secondary solver $SOLVER2$ (in our case PIBT) to avoid those expensive computations (line \ref{line_early_fallback}). This allows for the use of our primary $SOLVER$ for routing other groups while a faster secondary solver like PIBT is used to decongest large groups.

This will also have the counterintuitive but highly desirable property that when $K_{threshold}$ is hit more and uses $SOLVER2$ (e.g. highly congested problems), the algorithm will \textit{speed up}. This is in contrast to RHCR which will run $SOLVER$ all the way to timeout every time it fails to find a solution quickly enough before relying on a secondary solver. Therefore GD-RHCR leverages the ``fallback'' as a part of the algorithm rather than a ``failure'' of the method as in RHCR.

\textbf{Flexible choice of solvers.} Our framework allows flexible choice of potentially heterogeneous solvers for different groups, exploiting heterogeneous structures of the groups.
Although for this work, our specific implementation uses a specific ($SOLVER$, $SOLVER2$) pair that mirrors a fallback mechanism similar to RHCR across different groups. In general, this group decentralization sets up a modular framework where different solvers can easily be used for different groups of agents. We believe this will open the door to many variations and new research directions.

\section{Theory: Near Optimality of GD-RHCR}
\label{gdrhcr_theory}

For the theoretical results, we will again consider the L-MAPF problem modeled as a discounted MDP as described in \cref{rhcr_theory}. We will also consider a simplified version of the algorithm where $K_{threshold} = \infty$ so the secondary solver is not used and the soft constraint cost $c_{soft} = 0$ to simplify the analysis. In practice $c_{soft}$ will be small and incorporating estimations of out of view often does not change the overall theoretical guarantee \cite{deweese2025thinking}.

Again assuming that $SOLVER$ outputs an $\epsilon$-optimal finite horizon policy, we have the following theorem.
\begin{theorem}
    Let $\tau^{GD}_s$ denote the trajectory generated by GD-RHCR with $c_{soft} = 0$, $K_{threshold} = \infty$, and horizon $H$ starting at the state $s$. Let $V^{GD}(s)$ be the corresponding sum of discounted rewards (taking expectation over $\tau^{GD}_s$).  Then, the following result holds:

    \hspace{7ex}$V^*(s) - V^{GD}(s) \leq \frac{6\gamma^{\min(\lfloor\mathcal V /2\rfloor + 1, H)}}{(1 - \gamma)^2} + \frac{\epsilon}{1 - \gamma}$
    \label{gdrhcr_guarantee}
\end{theorem}

Notice that the exponent $\min(\lfloor\mathcal V /2\rfloor + 1, H)$ captures the duality between the space restriction and time restriction. With the group decentralized scheme (and $c_{soft} = 0$), we will be ignoring agents beyond distance $\mathcal V$ from any agent in the group. Notice that by the Dependence Time Lemma in lemma \ref{dtl}, agents in different groups are not able to collide within $\lfloor\mathcal V /2\rfloor$ timesteps, so if $H = \lfloor\mathcal V /2\rfloor + 1$ computing according to parallel groups will have no consequence on the solution and improves the computation time for free. In practice, $\mathcal V$ and $H$ may not be aligned so the theoretical guarantee depends on the minimum of the two.

Overall, in the same way that the replan window $\alpha$ in RHCR can be motivated through theoretical justification (see \cref{rhcr_theory}), the theory also motivates this group decentralized parallelization scheme as well as the lazy evaluation scheme.

\section{Simulations}
\label{simulations}

 \Cref{algorithm} is simulate across in different maps and settings. We will  plot the behavior between PIBT, RHCR with PIBT fallback, and GD-RHCR with PIBT as $SOLVER2$.

\textbf{Average Plan Time Metric:} There will be important nuance with the computational runtime graphs, where we will  compare \textit{the average planning time} and \textit{assume perfect parallelism}. That is, for GD-RHCR we will plot the average maximum time to complete the parallel groups across that timestep (only a subset of groups will be computing at each timestep).
For RHCR and PIBT, the \textit{average planning time} is just the average of the planning times whenever it plans (RHCR only plans every $1/\alpha$ iterations).

For RHCR, this is a useful metric when there is a uniform computation limit across all timesteps such as in applications with physical robots and move times. However, this metric may not capture this reduced number of overall computations (reduced by $1/\alpha$).
In \cref{appendix_wall}, total wall times are provided which incorporate the imperfect parallelism as well as the reduced frequency of RHCR evaluations. Our method still demonstrates a improvement in this case as well.

\textbf{Maps:} We will consider the Multi-Agent Pickup Delivery (MAPD) setting with 4 maps in \cref{mapd_graphs} where agents move from random pickup and dropoff locations. The sortation-1 and sortation-2 maps are inspired by \cite{li2021lifelong}.

We also consider the random navigation  L-MAPF setting where agents are assigned goals to random open spaces with 6 maps in \cref{random_graphs}. We introduce three more custom maps room-32-32-var1, room-64-64-var1, room-64-64-var2 which are variations on the traditional room maps (room-32-32-4 and room-64-64-8) with uniform openings between the rooms (see \cref{random_graphs}). This is because the specifically placed openings between the rooms in the traditional room maps make it more ``maze like'' and suitable for traditional MAPF rather than testing the topology of the room structure in L-MAPF.

\textbf{Parameters:} All RHCR simulations are run with PBS as the main solver with PIBT fallback and $\alpha = 5$,$H = 20$. Simulations for GD-RHCR will primarily use $c_{soft} = 1$, $H = 20$, $\mathcal V = 2$ and a threshold that depends on the agent count $K_{thresh} = \min(\lfloor 0.1k\rfloor, 20)$ (sortation-2 uses $K_{thresh} = \min(\lfloor 0.1k\rfloor, 50)$). We use 30 seeds with 3 in parallel for 500 timesteps on a 13th Gen Intel(R) Core(TM) i9-13900HX. Plotted is the average with 1-SD bands.

\subsection{Results}

The outcomes of our MAPD simulations are shown in \cref{mapd_graphs} and random navigation simulations in \cref{random_graphs}. See appendix \ref{abalations} and \ref{group_stats} for abalations and group statistics respectively.

\setlength{\mapwidth}{0.36\linewidth}
\setlength{\mapheight}{2.5cm}

\renewcommand{\ylabeltp}{\rotatebox{90}{\hspace{5ex}\small  throughput }}
\begin{figure*}[!ht]
  \centering
  \setlength{\tabcolsep}{1pt}
  \setlength{\plotwidth}{0.16\linewidth}
  \setlength{\plotheight}{2.8cm}
  \begin{tabular}{@{}c@{}cccccc@{}}
    &
    \mapentry{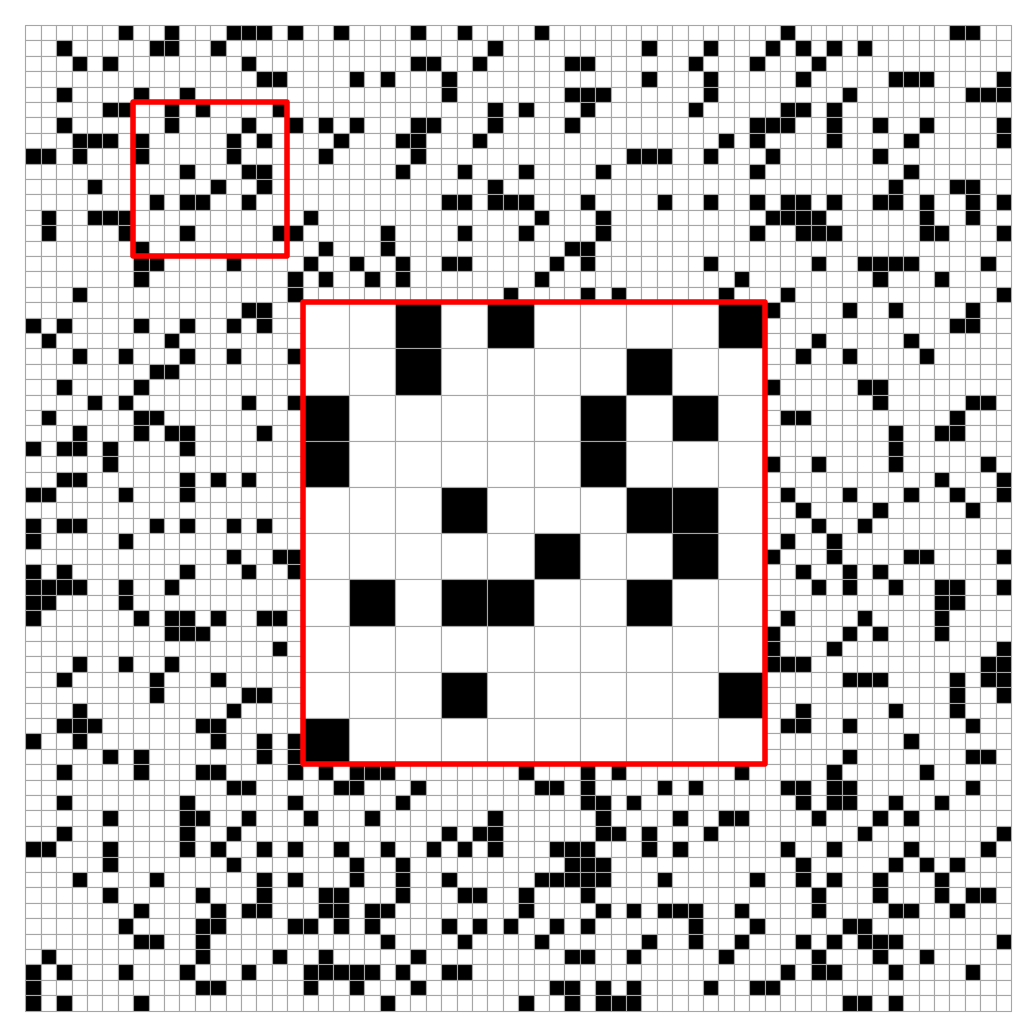}{random-64-64-20} &
    \mapentry{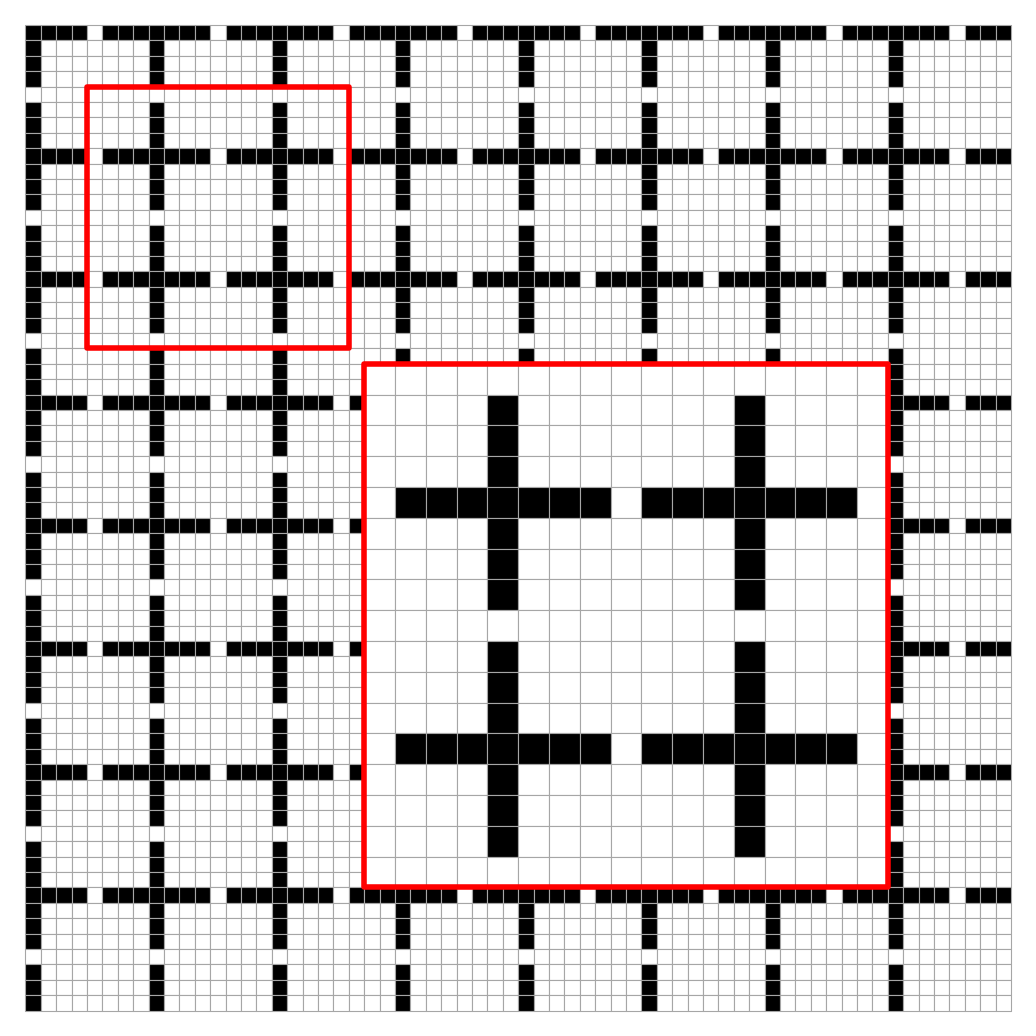}{room-64-64-var1} &
    \mapentry{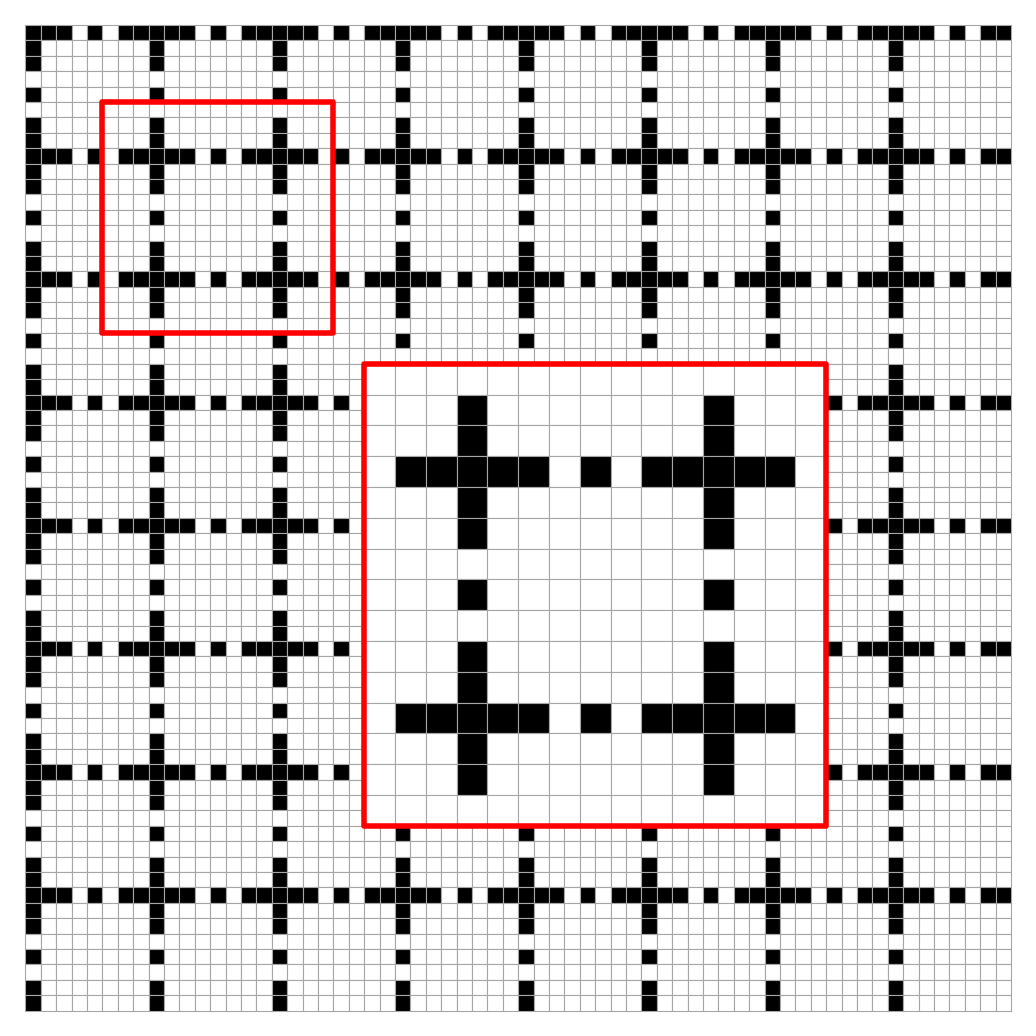}{room-64-64-var2}&
    \mapentry{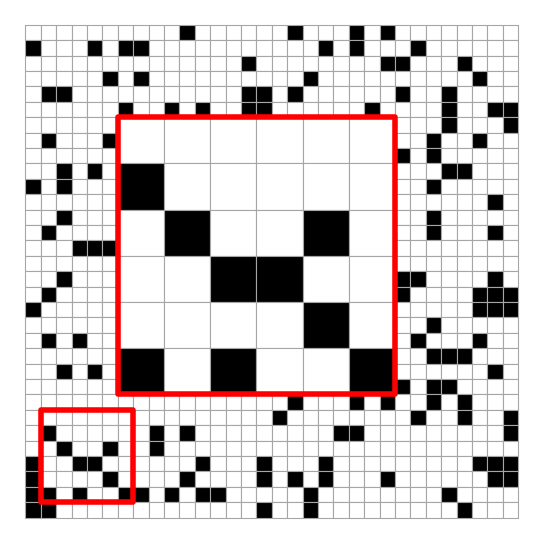}{random-32-32-20} &
    \mapentry{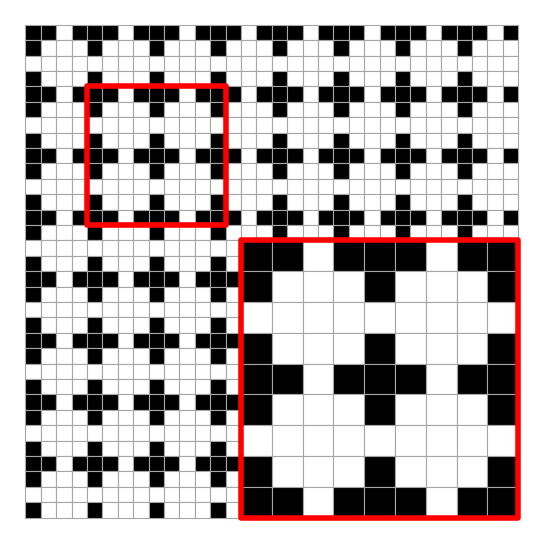}{room-32-32-var1} &
    \mapentry{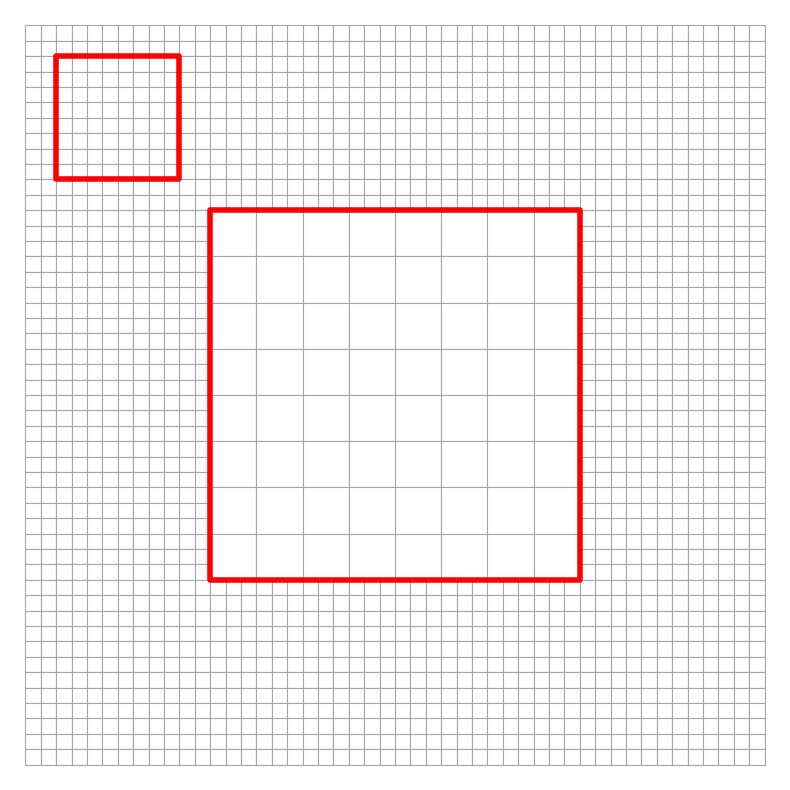}{empty-48-48}
    \\[0.8em]
    \ylabeltp &
    \plotentry{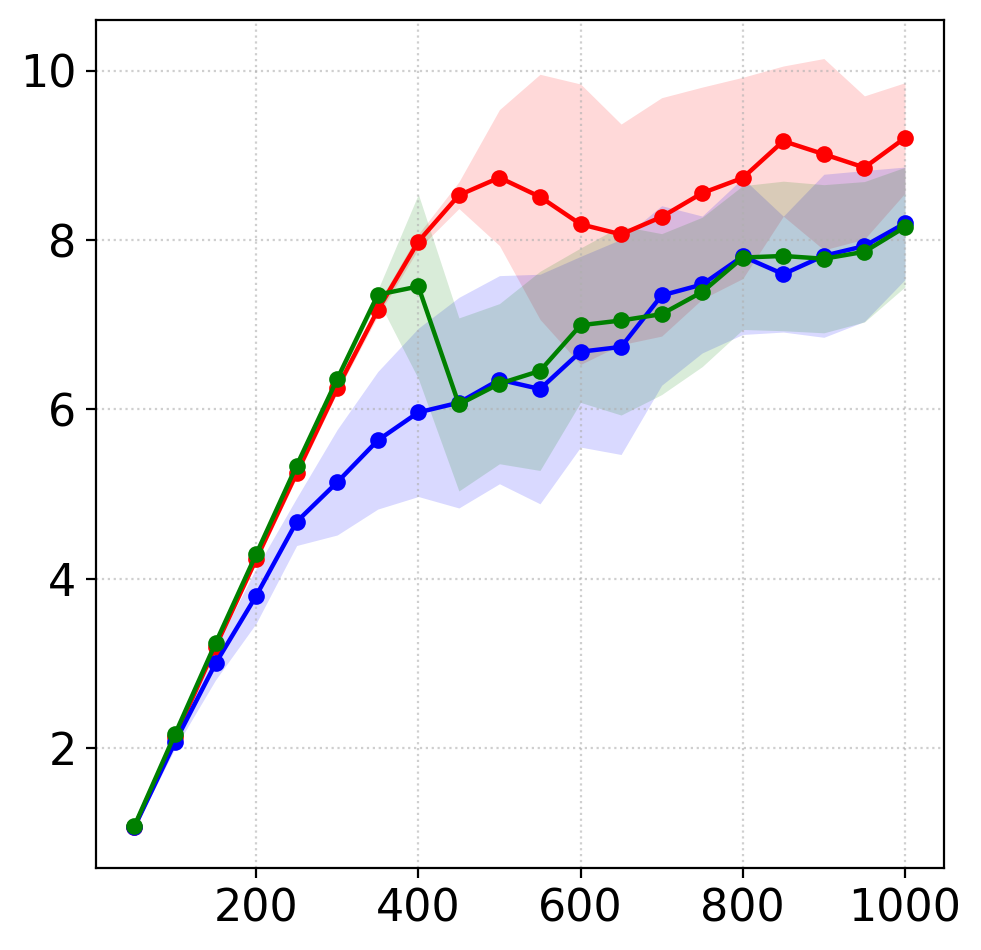} &
    \plotentry{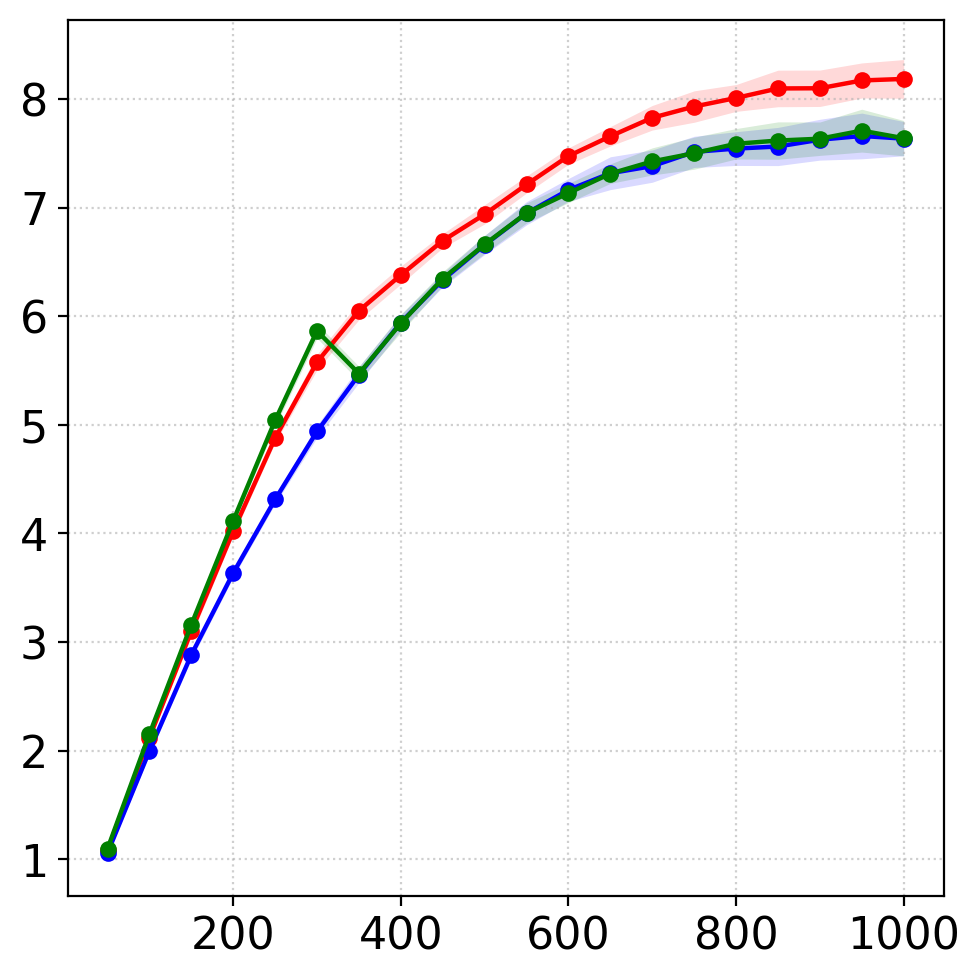}&
    \plotentry{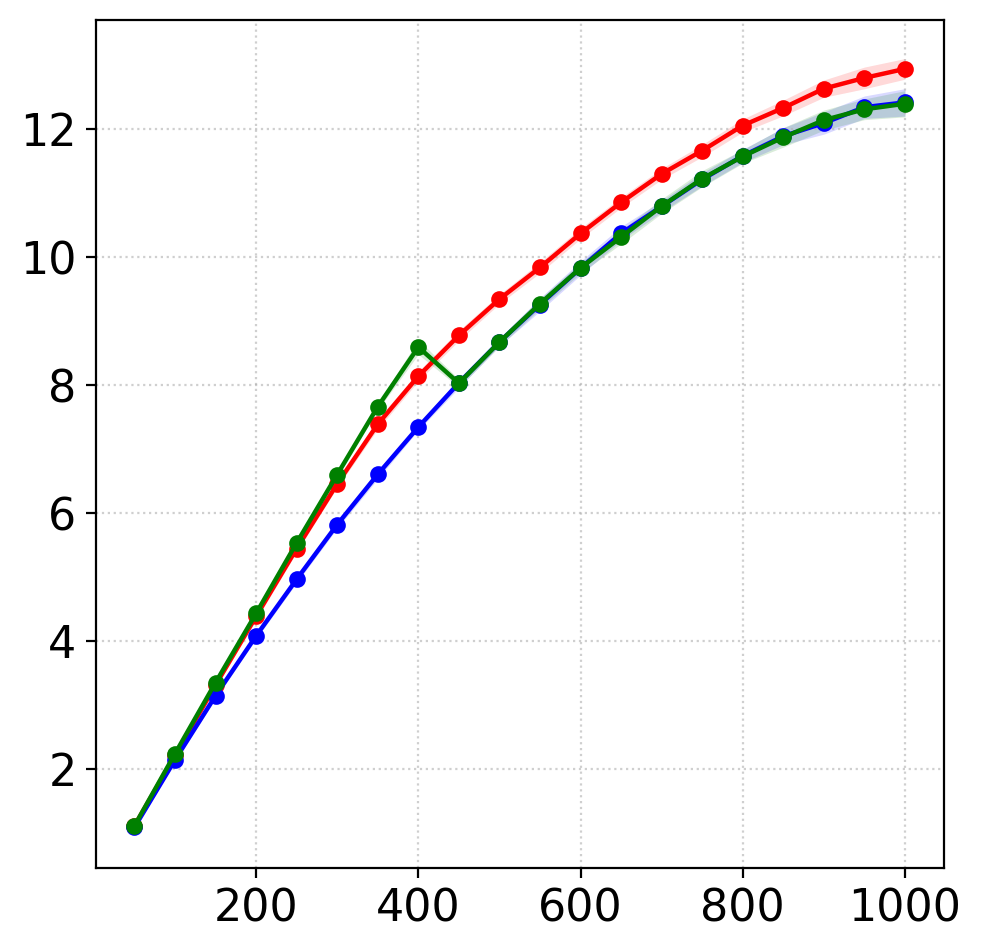}&
    \plotentry{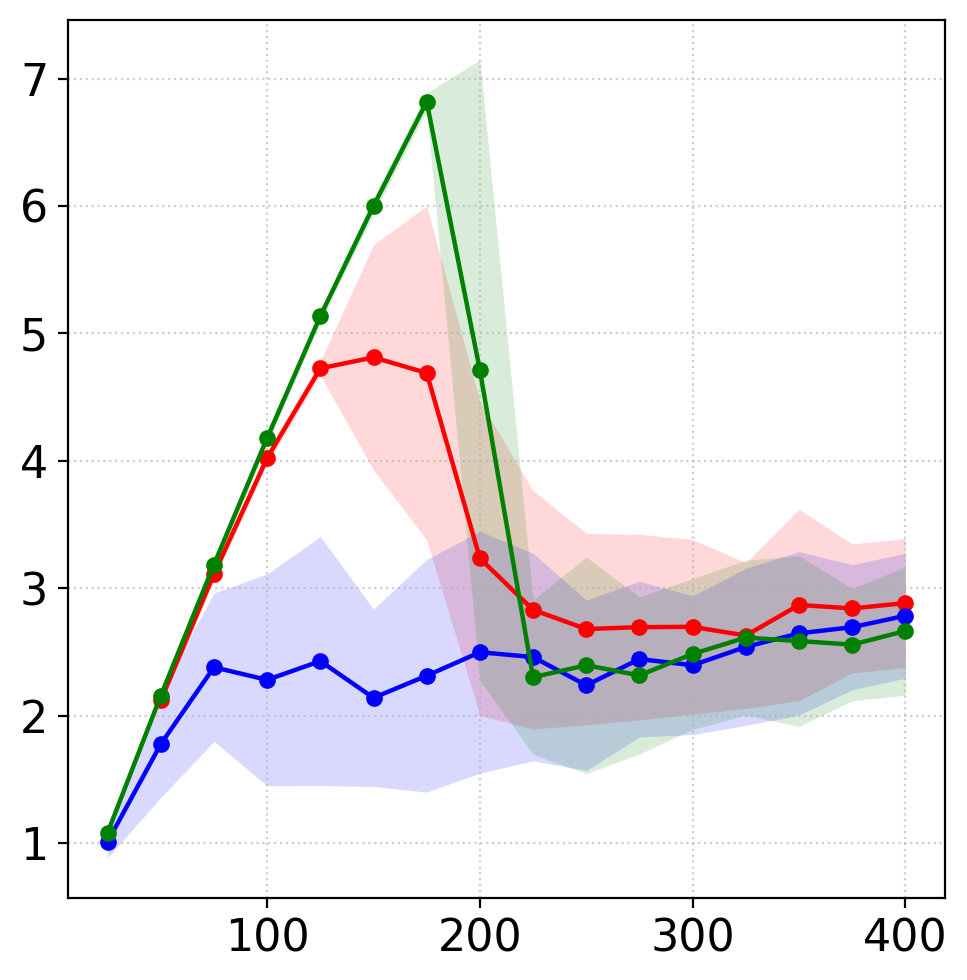} &
    \plotentry{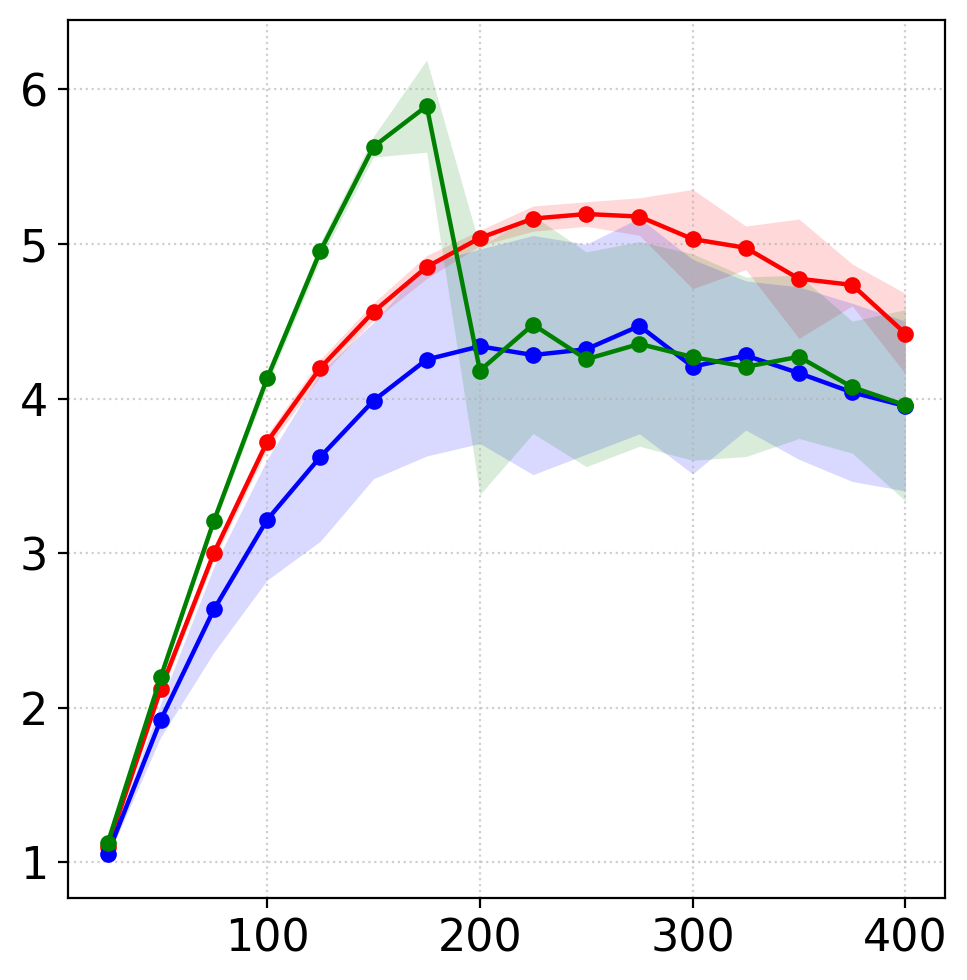} &
    \plotentry{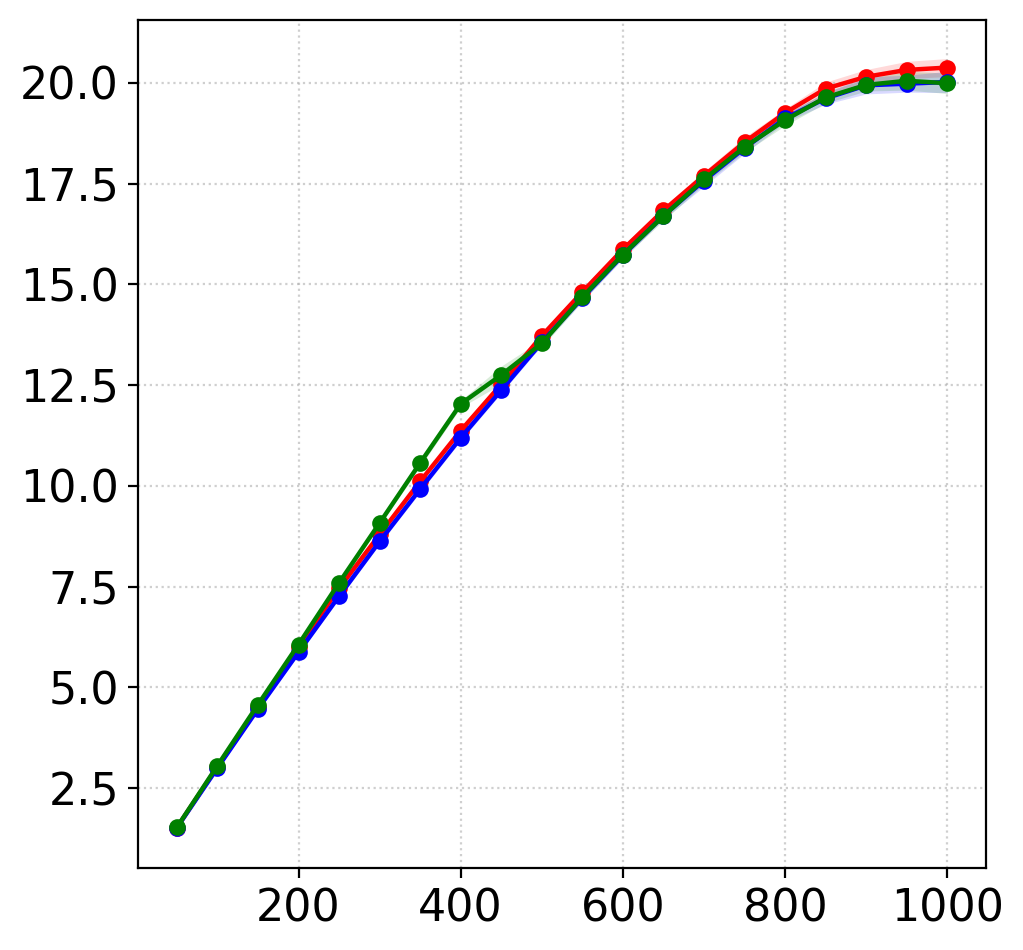}
    \\[0.3em]
    \ylabelrt &
    \plotentry{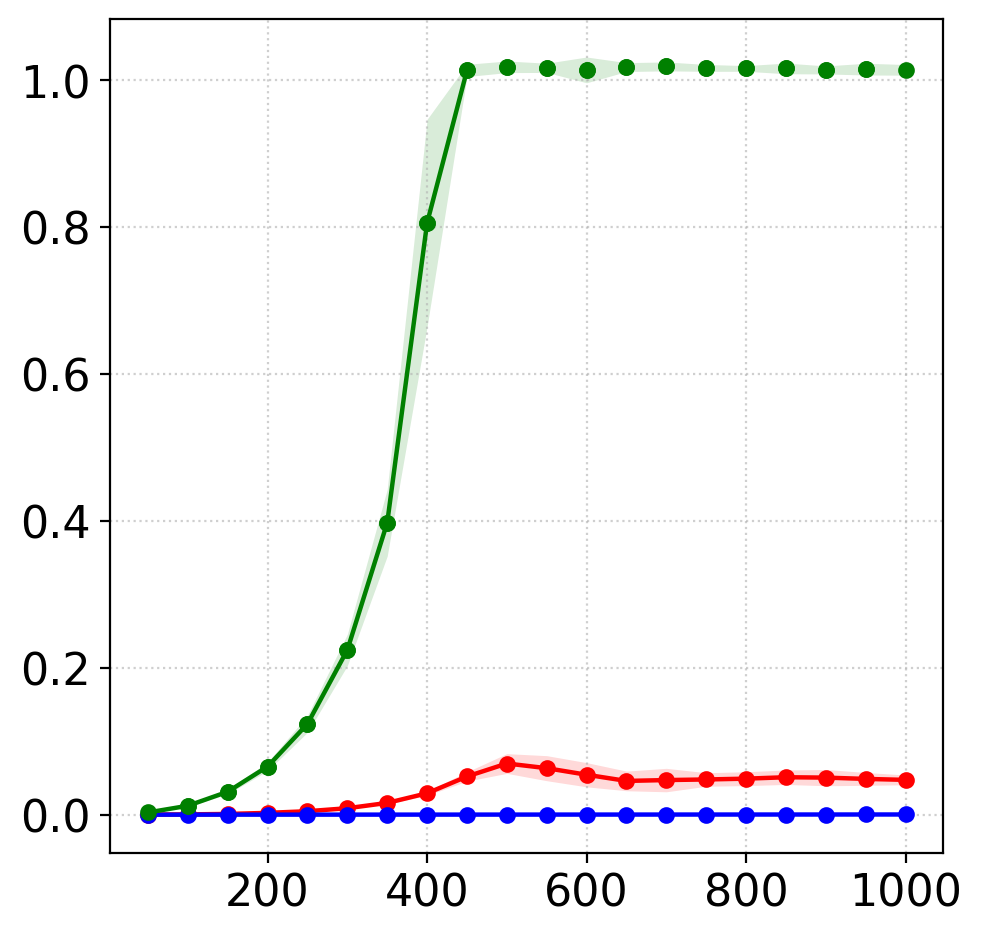} &
    \plotentry{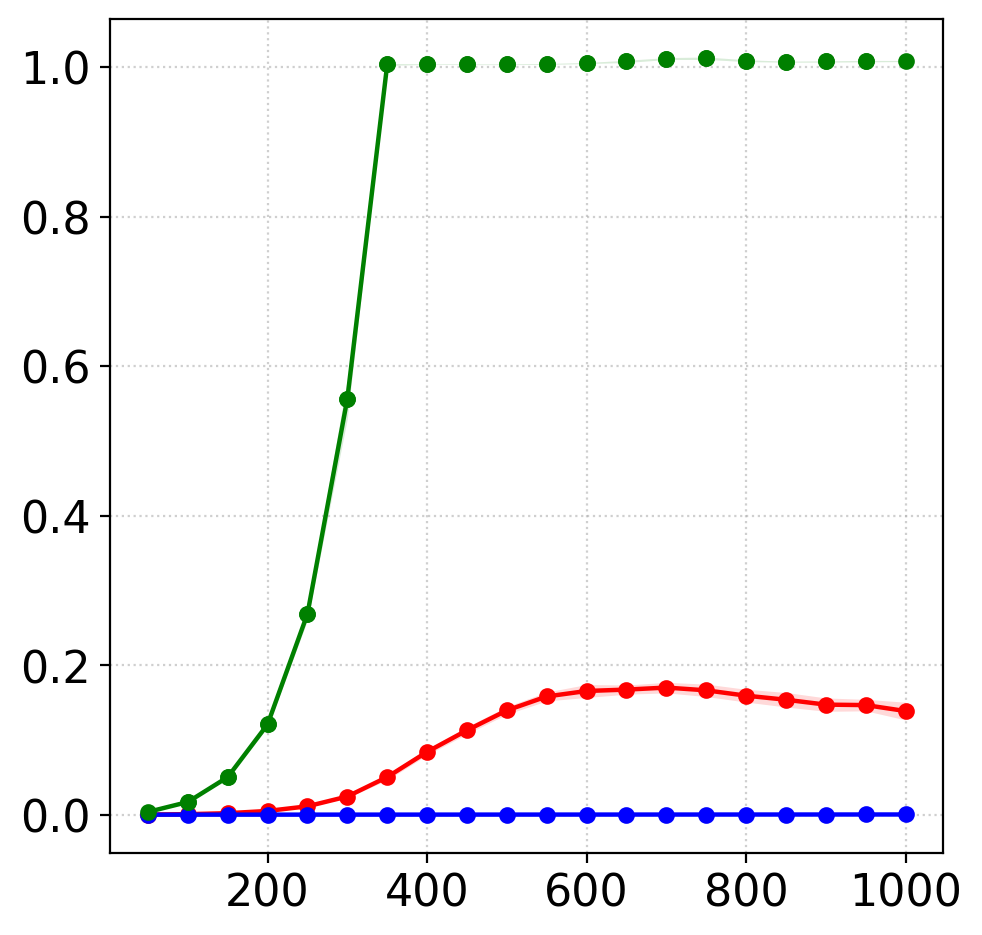}&
    \plotentry{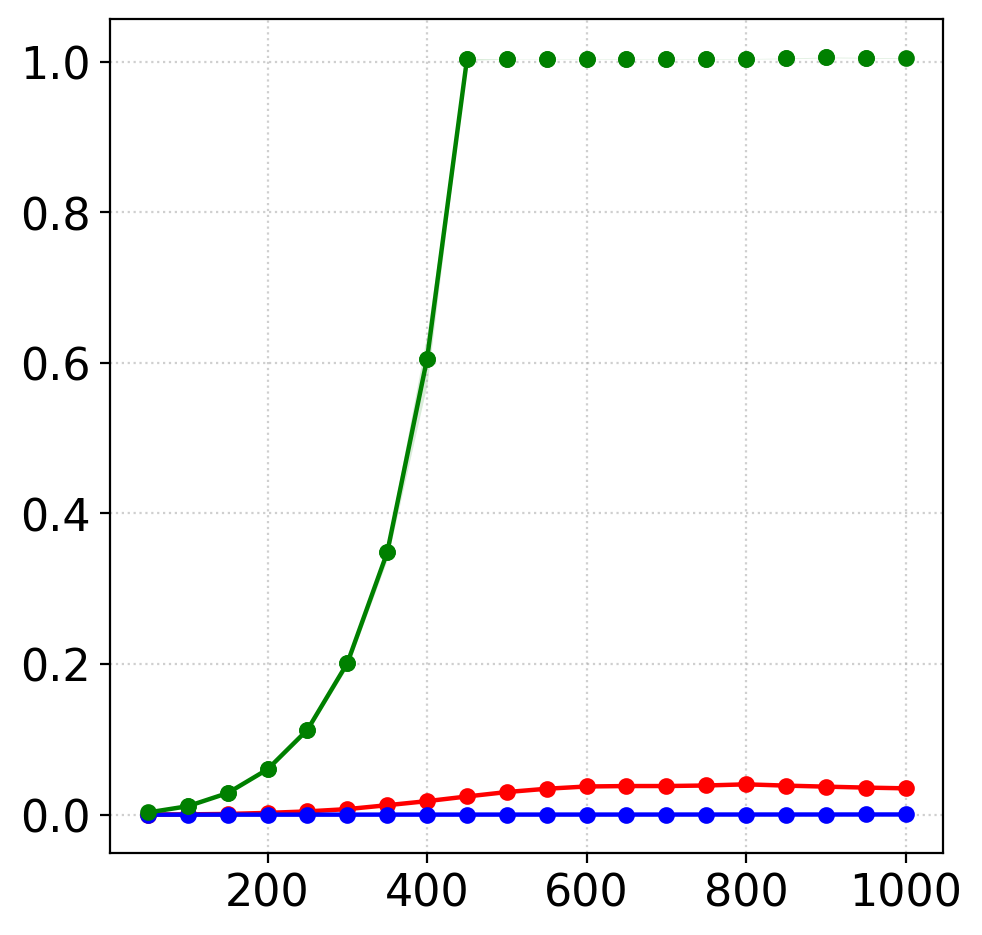}&
    \plotentry{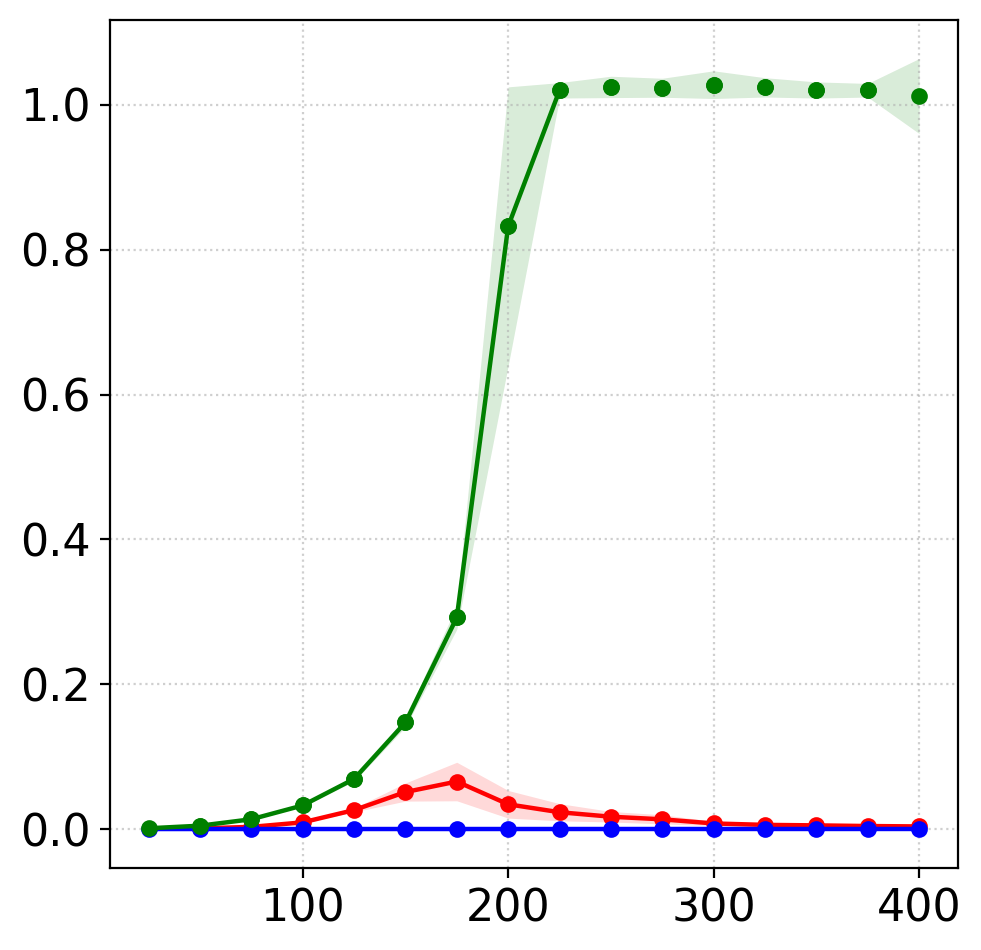} &
    \plotentry{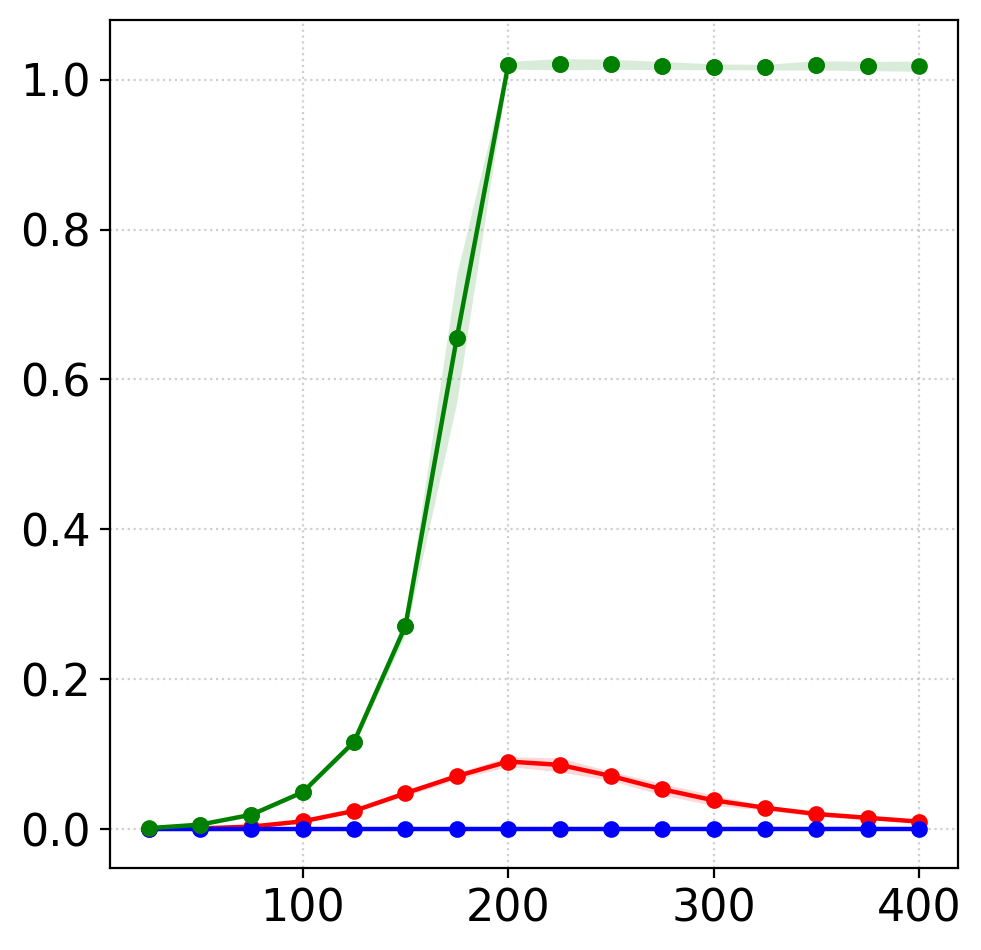} &
    \plotentry{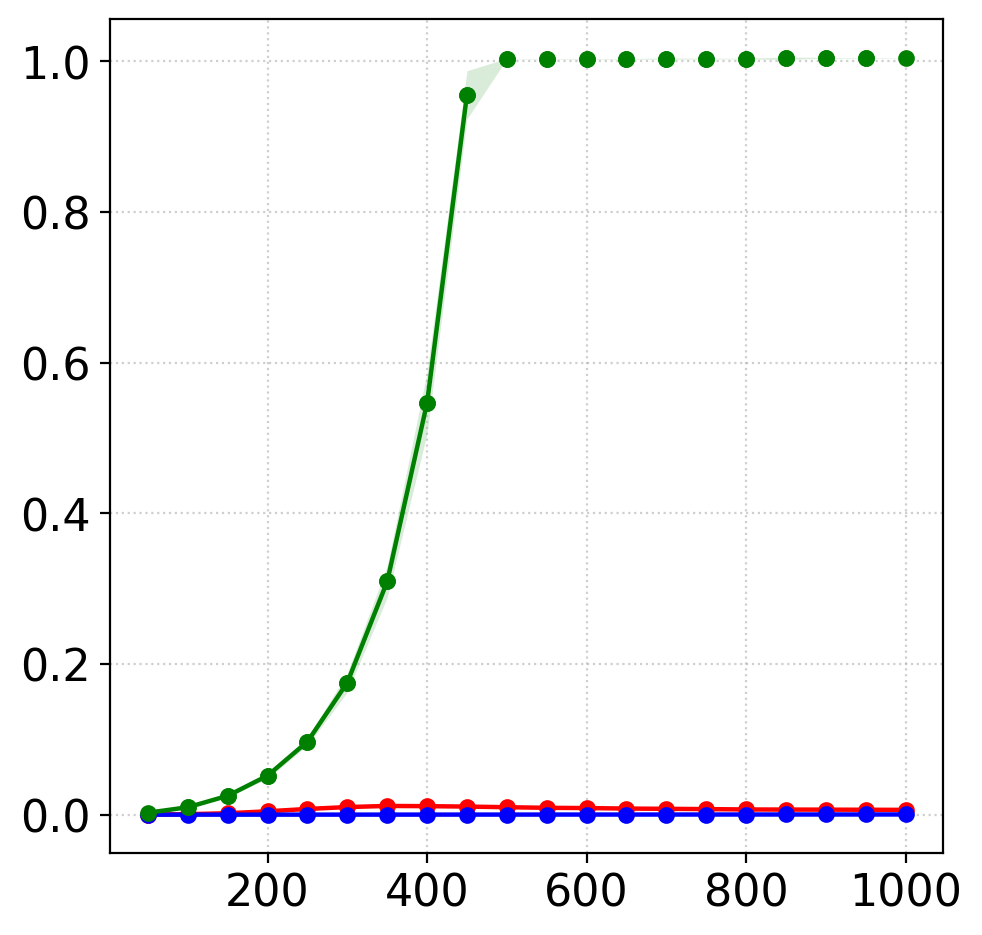}
    \\[0.3em]
    & \multicolumn{6}{c}{agents}
  \end{tabular}
  \caption{Random navigation simulations. The color coding is \textcolor{red}{GD-RHCR}, \textcolor{green!50!black}{RHCR}, and \textcolor{blue}{PIBT}}
  \label{random_graphs}
\end{figure*}
\textbf{Near RCHR Throughput of GD-RHCR with Reduced Compute Cost. }
For a wide variety of maps (warehouse-10-20-10-2-1, warehouse-10-20-10-2-2, sortation-1, random-64-64-20, room-64-64-var1, room-64-64-var2), GD-RHCR is able to maintain near RHCR performance before RHCR collapses. For example in warehouse-10-20-10-2-1, the throughput of GD-RHCR stays within $\approx 2\%$ with an average of $24.9$x improvement in plan time over RHCR before collapse. Although the visibility remains $\mathcal V = 2$, the soft constraints contributes significantly to incorporate global agent information across groups (see \cref{abalations}).

Further, we see that this is done at a significantly lower per-plan cost over RHCR. The staggered lazy evaluations of GD-RHCR significantly reduces the plan time required at each step because not all agents are planning simultaneously. Further, the paralellism substantially reduces the planning time for the agents that do evaluate at the same time.

\textbf{GD-RHCR Extends the Range of RHCR}
Because GD-RHCR restricts to planning for small groups (based on $K_{threshold}$), it is able to retain the planning leverage much further than RHCR. Even if challenging large congestion regions form, GD-RHCR will immediately use PIBT for those groups and continue to plan for the other smaller groups. This means that the throughput can remain higher and results in a much more graceful return to the PIBT as a fallback planner.

We see the most dramatic improvement in throughput on the warehouse-10-20-10-2-1 map. At its peak at $k = 800$, GD-RHCR sees a +57.7\% improvement over collapsed RHCR / PIBT while retaining its low per plan cost. Even at $k = 1000$, there are an average of approximately 250 groups that remain separated (see appendix \ref{group_stats} for full graphs). For warehouse-10-20-10-2-2 and sortation-1, we see an improvement of $12.6\%$ and $9.66\%$ respectively at $k = 1000$.

\textbf{Topology Influences Performance}
Overall, the topology of the map can substantially influence the performance of the algorithms. In maps such as warehouse-10-20-10-2-1 where the obstacles are dense and the shortest path (which wraps around obstacles) and Manhattan distance (which goes through obstacles) differ significantly, the topology induces smaller groups that allow for a significant advantage through the parallelism and group decentralized planning.

However, the advantage of GD-RHCR over the two methods is not as large when the maps are small (random-32-32-20 and room-32-32-var1) or less obstacle dense maps (sortation-2 and empty-48-48). In both cases, the topology induces large groups which must be planned all at once or drops through to the PIBT fallback when $K_{threshold}$ is reached.

Further, random-32-32-20, and room-32-32-var1 are not biconnected and can cause deadlocking in PIBT. Therefore for GD-RHCR that begins to rely on PIBT before RHCR can cause a reduced performance. However, we do see that computationally, GD-RHCR hits $K_{threshold}$ early and actually \textit{speeds up} with more agents. We also see in empty-48-48 that trivially all methods have similar throughput because RHCR itself does not have a significant advantage. However, the computation time of GD-RHCR is near PIBT because the large groups immediately transition to PIBT.

\section{Conclusion and Future Work}
In this work, we provided a theoretical grounding for RHCR using methods from the LI-MDP literature. We then used those analytical techniques to motivate our new GD-RHCR framework that satisfies similar theoretical guarantees to RHCR and performs well empirically across many maps into larger agent counts over RHCR.

For future works, we believe that group decentralization is a relatively new concept to the MAPF community and can be integrated in various ways. For example with traditional learning based methods or performing parallelism for other methods such as PIBT. Further, this work creates a simple setup to use different algoriths in different locations / settings which opens up many research directions.

\section{Acknowledgements}
Guannan Qu is supported by NSF Grants 2339112, 2512805, Jane Street, and Pennsylvania Infrastructure Technology Alliance. Jiaoyang Li is supported by NSF Grants 2328671 and 2441629. In addition, Alex DeWeese is supported by Leo Finzi Memorial Fellowship in Electrical \& Computer Engineering, the David H. Barakat and LaVerne Owen-Barakat CIT Dean's Fellowship, and the Fritsch Family Fellowship.

\bibliography{refs}

\appendix
\section{MAPF as LI-MDP}
\label{limdp_fit}
Below, we give the formal definition of the Locally Interdependent Multi-Agent MDP (LI-MDP) and demonstrate how our MAPF MDP model from \cref{rhcr_theory} fits into this model.

Assume we have agent $\mathcal N$ with some metric space $\mathcal X$ with a corresponding distance metric $d$. Each agent $i$ may each also have an internal state $\mathcal Y_i$ (e.g. battery power) resulting in a state space of $\mathcal S_i = (\mathcal X, \mathcal Y_i)$ with an arbitrary action space of $\mathcal A_i$. Each agent $i$ will transition independently in the environment according to a individual transition function $P_i(s'_i\lvert s_i,a_i)$ with $s_i,s_i'\in \mathcal S_i$, $a_i \in \mathcal A_i$ however, we assert that this probability is 0 when $d(s_i,s_i') > 1$ (agents cannot move more than 1 space at a time in the environment). Lastly, for any two agents $i,j \in \mathcal N$ we will assume a local reward function $\overline r_{i,j}(s_i, a_i, s_j, a_j)$ which is 0 when $d(s_i, s_j) > \mathcal R$ according to a dependence radius $\mathcal R \geq 0$.
Here, $i,j$ can be equal in which case $\overline r_{i,i}(s_i, a_i, s_i, a_i)$ represents the ``single agent reward'' which depends only on a single agent state / action (the agent will always be within distance $\mathcal R$ of itself so the restriction does not apply). Therefore the reward function decomposes into $\sum_{i, j\in g} \overline r_{i,j}(s_i, a_i, s_j, a_j) = \sum_{i \in \mathcal N} \overline r_{i}(s_i, a_i) + \sum_{i\neq j} \overline r_{i,j}(s_i, a_i, s_j, a_j)$ where $\overline r_{i}(s_i, a_i) = \overline r_{i,i}(s_i, a_i, s_i, a_i)$.

To summarize, we have the following:
\begin{definition}
Assume the definition of the components mentioned above. The LI-MDP is defined as:
\begin{itemize}
\item $\mathcal S := \times_{i \in \mathcal N}\mathcal S_i$ where $\mathcal S_i = (\mathcal X, \mathcal Y_i)$
\item $\mathcal A := \times_{i \in \mathcal N}\mathcal A_i$
\item $ P(s' \lvert s, a) = \prod_{i\in \mathcal N} P_i(s'_i \lvert s_i,a_i)$
\item $ r(s, a) = \sum_{i \in \mathcal N} \overline r_{i}(s_i, a_i) + \sum_{i\neq j} \overline r_{i,j}(s_i, a_i, s_j, a_j)$
\end{itemize}

\end{definition}

Notice the similarities with the MAPF MDP formulation described in \cref{rhcr_theory}. Each agent moves in an environment with a metric space at most one space at a time individually. The agents have a single agent goal incentive captured by $\overline r_{i}(s_i, a_i)$ and a interdependent collision penalty when agents have vertex collisions captured by $\overline r_{i,j}(s_i, a_i, s_j, a_j)$ with $\mathcal R = 0$. The only discrepancy is the edge collisions which allow for collision penalties with agents 1 step away (which we want to handle with $\mathcal R = 0$). It turns out that this does not impact the theory as a critical lemma called the Dependence Time Lemma (lemma \ref{dtl}) still holds.

\section{Proofs}
\label{appendix_theory}
\subsection{RHCR}
\label{rhcr_proof}
If $\pi_{finite}(s) = \{\pi_0, \pi_1, \ldots, \pi_{H - 1}\}$ is a finite horizon policy, let $V_h^{\pi_{finite}}(s) = \mathbb {E}_{\tau \sim \pi_{finite}\lvert_s} [\sum_{t = h}^{H - 1} \gamma^{t - h} r(s_t, a_t)]$. Assume $\pi^*_{finite}(s) = \{\pi_0^*, \pi_1^*, \ldots, \pi_{H - 1}^*\}$ is the finite horizon optimal policy calculated by RHCR at each replan window $\alpha$.

Recall, RHCR repeatedly takes the finite horizon policy only up to a certain number of timesteps $\alpha$. So, to express the discounted sum of rewards (the value function) of RHCR, we can repeatedly take the discounted rewards in the value function $V^{\pi^*_{finite}}_0(s_t)$ up to timesteps $\alpha$ and subtract the discounted value function beyond timestep $\alpha$.
This takes the form as follows $V^{RHCR}(s) = \mathbb{E}_{\tau \sim \pi\lvert_{s}}\bigg[V^{\pi^*_{finite}}_0(s) - \sum_{\ell = 1}^\infty \gamma^{\alpha\ell} \Delta_{\ell}^\tau\bigg]$ where $\Delta_{\ell}^\tau = V^{\pi^*_{finite}}_\alpha(s(\alpha\ell)) -  V^{\pi^*_{finite}}_0(s(\alpha\ell))$.

Now, we would like to bound $\Delta_\ell^\tau$. However, notice the value function terms within $\Delta_\ell^\tau$ contain discounted rewards that consider different time intervals. $V^{\pi^*_{finite}}_\alpha(s(\ell \alpha)) = \mathbb {E}_{\tau \sim \pi_{finite}^*\lvert_s} [\sum_{t = \alpha}^{H - 1} \gamma^{t - \alpha} r(s(t), a(t))]$ only considers $H - \alpha$ timesteps and $V^{\pi^*_{finite}}_0(s_{\alpha\ell}) = \mathbb {E}_{\tau \sim \pi_{finite}^*\lvert_s} [\sum_{t = 0}^{H - 1} \gamma^t r(s(t), a(t))]$ considers $H$ timesteps.

We will take an arbitrary finite horizon policy $\pi_{finite}$ and complete the trajectory within $V^{\pi^*_{finite}}_\alpha(s(\alpha(\ell - 1)))$. The extended value function $V^{\ell, extended}$ can be expressed as follows.
{\scriptsize
\begin{align*}
&V^{\ell, extended}= \mathbb {E}_{\tau \sim \pi_{finite}^*\lvert_s} \bigg[\sum_{t = \alpha}^{H - 1} \gamma^{t - \alpha} r(s_t, a_t)\bigg] \\
&\hspace{25ex}+ \mathbb {E}_{\tau \sim \pi_{finite}\lvert_s} \bigg[\sum_{t = H}^{H + \alpha - 1} \gamma^{t - \alpha} r(s(t), a(t))\bigg]\\
& = V^{\pi^*_{finite}}_\alpha(s(\alpha(\ell - 1))) + \mathbb {E}_{\tau \sim \pi_{finite}\lvert_s} \bigg[\sum_{t = H}^{H + \alpha - 1} \gamma^{t - \alpha} r(s_t, a_t)\bigg] \\
&\geq V^{\pi^*_{finite}}_\alpha(s(\alpha(\ell - 1))) - \gamma^{H-\alpha} (\frac{1 - \gamma^\alpha}{1 - \gamma})\\
&\geq V^{\pi^*_{finite}}_\alpha(s(\alpha(\ell - 1))) -  \frac{\gamma^{H-\alpha}}{1 - \gamma}
\end{align*}
}
Now, since $V^{\pi^*_{finite}}_0(s(\alpha\ell))$ is the optimal finite horizon reward, we may bound $\Delta_i^\tau$ as follows:
\begin{align*}
    \Delta^\tau_\ell &= V^{\pi^*_{finite}}_\alpha(s(\alpha\ell)) -  V^{\pi^*_{finite}}_0(s(\alpha\ell))\\
    &\leq V^{\ell, extended} -  V_0^{\pi^*_{finite}}(s(\alpha\ell)) + \frac{\gamma^{H - \alpha}}{1 - \gamma}\\
    &\leq \epsilon + \frac{\gamma^{H-\alpha}}{1 - \gamma}
\end{align*}

Therefore,
\begin{align*}
    V^{RHCR}(s) &= \mathbb{E}_{\tau \sim \pi\lvert_{s}}\bigg[V^{\pi^*_{finite}}_0(s) - \sum_{\ell = 1}^\infty \gamma^{\alpha\ell} \Delta_{\ell}^\tau\bigg]\\
    &\geq V^{\pi^*_{finite}}_0(s) - \frac{\gamma^{H}}{(1 - \gamma)^2} - \frac{\epsilon}{1 - \gamma}
\end{align*}

To complete our analysis, we may compare the optimal stationary policy $\pi^*$ and the RHCR trajectory as follows:
\begin{align*}
V^*(s) - V^{RHCR}(s) &\leq V^*(s) - V_0^{\pi^*_{finite}}(s) \\
&\hspace{10ex}+ \frac{\gamma^{H}}{(1 - \gamma)^2}+ \frac{\epsilon}{1 - \gamma}\\
 &\leq \frac{\gamma^H}{1 - \gamma} + \frac{\gamma^{H}}{(1 - \gamma)^2}  + \frac{\epsilon}{1 - \gamma}\\
 &\leq \frac{2\gamma^H}{(1 - \gamma)^2} + \frac{\epsilon}{1 - \gamma}
\end{align*}

\subsection{Group Decentralized RHCR}

In this section, we will prove \cref{gdrhcr_guarantee}. Recall, we will be using $c_{soft} = 0$ and $K_{threshold} = \infty$.

Let $Z(s) = \{z_1, z_2, \ldots, z_\ell\}$ where $z_i \subset \mathcal N$ be the partition formed by the transitive ``grouped decentralized'' connections.

Similar to \cite{deweese2024locally}, our proofs will rely on a fundamental geometric property that agents in separate groups cannot interact with each other within a certain number of steps. In the MAPF case, this means agents will not be able to collide within $c = \lfloor\frac{\mathcal V}{2}\rfloor$ timesteps. This is a special case of the Dependence Time Lemma from \cite{deweese2024locally} with $\mathcal R = 0$ which shows up in $c = \lfloor\frac{\mathcal V - \mathcal R}{2}\rfloor$ with the exception of edge collisions.

\begin{lemma}[Dependence Time Lemma for MAPF]
    For some state $s \in \mathcal S$, let $i,j \in \mathcal N$ be any two agents in different partitions in $Z(s)$. Then, agents $i,j$ cannot have vertex or edge collisions within $\lfloor\frac{\mathcal V}{2}\rfloor$ timesteps.
    \label{dtl}
\end{lemma}
\proof{
See proof for the Dependence Time Lemma in \cite{deweese2024locally} with $\mathcal R = 0$. Edge collisions do not make a difference in the proof because showing that the closest two agents in different groups can get is distance 1 in $\lfloor\frac{\mathcal V}{2}\rfloor$ steps rules out both vertex and edge collisions.
}

Notice by our reward model in \cref{rhcr_theory}, $r(s, a) =  \sum_i \big(I[x_i = g_i]r_{goal} + \sum_j I[collision(x_i, x_j)]p_{collide}\big)$, for any policy $\pi$, we may decompose the value function as follows for any partition $P$:
\begin{align*}
V^\pi_h(s) = \sum_{p \in P} \Big( [V_{h}^{\pi}]_p(s_p)+ \sum_{i \in p}[V_h^{\pi}]_{i\rightarrow \overline p}(s)\Big)
\end{align*}
where  we define $\overline p := \mathcal N \setminus p$ and for every $p' \subset \mathcal N$ both
\begin{align*}
&[V^\pi_h]_{p'}(s) = \mathbb E\bigg[\sum_{t = 0}^\infty\gamma^t \sum_{i \in p'} I[x_i = g_i]r_{goal} \\
&\hspace{20ex}+ \sum_{j \in p'} I[collision(x_i, x_j)]p_{collide}\bigg]
\end{align*}
and
 \begin{align*}
    [V_h^{\pi}]_{i\rightarrow p'}(s) = \mathbb E\bigg[\sum_{t = 0}^\infty\gamma^t \sum_{j \in p'}I[collision(x_i, x_j)]p_{collide}\bigg].
 \end{align*}  Symbolically, $i\rightarrow p'$ represents the interactions of agent $i$ (or collisions) with agents in $p'$.

Using these new notations, we may express the trajectory value obtained by GD-RHCR as
\begin{align*}
V^{GD}(s) \geq \mathbb{E}_{\tau \sim \pi\lvert_{s}}\bigg[\sum_{z \in Z(s)}[V^{\pi^*_{finite}}_0]_z(s) - \sum_{t = 1}^\infty \gamma^t \Delta_{t}^\tau\bigg] - \frac{\epsilon}{1 - \gamma}
\end{align*}
where
{\scriptsize
\begin{align*}
&\Delta_{t}^\tau = \sum_{z \in Z(s)}I[lazy(z, t)]\bigg(\sum_{i \in z}\Big([V^{\pi^*_{finite}}_{\delta_{i}^t}]_z(s_{z_{t,i}}(t)) \\
&\hspace{10ex}+ [V^{\pi^*_{finite}}_{\delta_{i}^t}]_{i \rightarrow z_{t,i}\setminus z}(s_{z_{t,i}}(t))\Big) -  [V^{\pi^*_{finite}}_0]_z(s(t))\bigg)
\end{align*}
}
Here $I(lazy(z,t))$ is an indicator whether lazy evaluation was triggered at that timestep and $z_{t,i}$ is the group associated with the previous "plan".

Intuitively, we repeatedly remove the remainder of the trajectory when lazy evaluation is triggered and bring in the new trajectory plan for each of the groups. This accounts for the collisions across groups because in order for agents to collide, they must enter each others groups first (triggering a lazy evaluation).

Next, as in \cref{rhcr_proof}, we will extend the trajectory of $[V^{\pi^*_{finite}}_{\delta_{i}^t}]_i(s_{z_{t,i}}(t))$ with some arbitrary policy $\pi_{i, finite}$. Notice here that in our extension $V^{i,t, ext}$, we will only consider collisions between agents that were initially in the same $z_{t,i}$ group. Defining $r_{i, max}$ as the maximum reward possible for a single agent (we assume $\sum_{i \in \mathcal N}\lvert r_{i, max}\rvert\leq 1$), we have

{\scriptsize
\begin{align*}
&\sum_{i\in\mathcal N} I[lazy(z, t)]V^{i,t,ext} = \sum_{i\in\mathcal N} I[lazy(z, t)]\bigg(V^{\pi^*_{finite}}_{\delta_i^t}(s_{z_{t,i}}(t)) \\
&\hspace{20ex}+ \mathbb {E}_{\tau \sim \pi_{i, finite}\lvert_s} \bigg[\sum_{t' = H}^{H + \delta_i^t - 1} \gamma^{t' - \delta_i^t} \bigg(I[x_i = g_i]r_{goal} \\
&\hspace{40ex}+ \sum_{j \in z_{t,i}}I[collision(x_i, x_j)]p_{collide}\bigg)\bigg] \\
&\geq \sum_{i\in\mathcal N}I[lazy(z, t)]\Big( V^{\pi^*_{finite}}_{\delta_i^t}(s_{z_{t,i}}(t)) - \gamma^{H-\delta_i^t} (\frac{1 - \gamma^{\delta_i^t}}{1 - \gamma})r_{i,max}\Big)\\
&\geq \sum_{i\in\mathcal N}I[lazy(z, t)]\Big( V^{\pi^*_{finite}}_{\delta_i^t}(s_{z_{t,i}}(t)) -  \frac{\gamma^{H-\delta_i^t}}{1 - \gamma}r_{i,max}\Big)\\
\end{align*}
}

Now we may bound $\Delta_t^\tau$ as follows.
{\scriptsize
\begin{align}
&\gamma^t\Delta_{t}^\tau = \gamma^t\sum_{z \in Z(s)}I[lazy(z, t)]\bigg(\sum_{i \in z}\Big([V^{\pi^*_{finite}}_{\delta_{i}^t}]_z(s_{z_{t,i}}(t)) \nonumber\\
&\hspace{20ex}+ [V^{\pi^*_{finite}}_{\delta_{i}^t}]_{i \rightarrow z_{t,i}\setminus z}(s_{z_{t,i}}(t))\Big) -  [V^{\pi^*_{finite}}_0]_z(s(t))\bigg)\\
    &\leq \gamma^t\sum_{z \in Z(s(t))}I[lazy(z, t)]\bigg( \sum_{i \in z} V^{i,t,ext} -  [V_0^{\pi^*_{finite}}]_z(s(t)) \nonumber\\
    &\hspace{25ex}\sum_{i \in z}[V^{\pi^*_{finite}}_{\delta_{i}^t}]_{i \rightarrow z_{t,i}\setminus z}(s_{z_{t,i}}(t))+ \sum_{i \in z}\frac{\gamma^{H - \delta_i^t}r_{i,max}}{1 - \gamma} \bigg)\label{eq_sub}\\
    &\leq \sum_{z \in Z(s)}I[lazy(z, t)]\Big(\gamma^t(\sum_{i \in z} V^{i,t,ext} -  [V_0^{\pi^*_{finite}}]_z(s(t))) \\
    &\hspace{30ex}+ \sum_{i \in z} \frac{\gamma^{\lfloor \frac{\mathcal V}{2}\rfloor + t - \delta_i^t + 1}r_{i,max}}{1 - \gamma}+ \frac{\gamma^{H + t - \delta_i^t}r_{i,max}}{1 - \gamma} \Big)\label{eq_dtl}\\
    &\leq \sum_{z \in Z(s)}I[lazy(z, t)]\Big(\sum_{i \in z}\gamma^{I[t > \lfloor \frac{\mathcal V}{2}\rfloor]\cdot t}r_{i,max} \\
    &\hspace{30ex}+ \sum_{i \in z} \frac{\gamma^{\lfloor \frac{\mathcal V}{2}\rfloor
    + t - \delta_i^t + 1}r_{i,max}}{1 - \gamma}
    + \frac{\gamma^{H + t - \delta_i^t}r_{i,max}}{1 - \gamma} \Big)\label{eq_delta}
\end{align}
}
In line \ref{eq_sub}, we substitute our extension. In line \ref{eq_dtl}, we use the Dependence Time Lemma (lemma \ref{dtl}) where cross group interactions are beyond $\lfloor \frac{\mathcal V}{2}\rfloor$ iterations (with a head start of $\delta_i^t$ iterations). In line \ref{eq_delta}, recall that in the definition of $V^{i,t, ext}$ only the collisions between agents in the same $z_{t,i}$ groups were added. For cross group interactions, these may leave residual terms. However, these terms must occur beyond $\lfloor\frac{\mathcal V}{2}\rfloor$ iterations by the Dependence Time Lemma. Therefore we introduce the quantity with the indicator in the exponent $I[t > \lfloor \frac{\mathcal V}{2}\rfloor]\cdot t$.

Now notice that $t - \delta_i^t$ is some quantity unique for every t (when the indicator $I[lazy(z,t)]$ is triggered) and therefore the $t - \delta_i^t$ terms can be summed up as indices of $1, \ldots, \infty$.
Therefore we have

{\small
\begin{align*}
\mathbb{E}_{\tau \sim \pi\lvert_{s}}\bigg[\sum_{t = 1}^\infty \gamma^{t} \Delta_{t}^\tau\bigg]\leq \frac{\gamma^{H}}{(1 - \gamma)^2} + \frac{2\gamma^{\lfloor \mathcal V / 2\rfloor + 1}}{(1 - \gamma)^2}\\
\end{align*}
}
And substituting into our $V^{GD}$ expression,
{\small
\begin{align*}
V^{GD}(s) &\geq \mathbb{E}_{\tau \sim \pi\lvert_{s}}\bigg[\sum_{z \in Z(s)}[V^{\pi^*_{finite}}_0]_z(s) - \sum_{t = 1}^\infty \gamma^{t} \Delta_{t}^\tau\bigg] - \frac{\epsilon}{1 - \gamma}\\
&\geq \sum_{z \in Z(s)}[V^{\pi^*_{finite}}_0]_z(s) - \frac{\gamma^{H}}{(1 - \gamma)^2} - \frac{2\gamma^{\lfloor \mathcal V / 2\rfloor + 1}}{(1 - \gamma)^2} - \frac{\epsilon}{1 - \gamma}\\
&\geq \sum_{z \in Z(s)}[V^{\pi^*_{finite}}_0]_z(s) - \frac{4\gamma^{\min(H, \lfloor \mathcal V / 2\rfloor + 1)}}{(1 - \gamma)^2} - \frac{\epsilon}{1 - \gamma}
\end{align*}
}
Finally, we complete our analysis by comparing to the optimal value function $V^*(s)$.
{\small
\begin{align*}
V^*(s) - V^{GD}(s) &\leq V^*(s) - \sum_{z \in Z(s)}[V^{\pi^*_{finite}}_0]_z(s) \\
&\hspace{10ex}+ \frac{4\gamma^{\min(H, \lfloor \mathcal V / 2\rfloor + 1)}}{(1 - \gamma)^2} + \frac{\epsilon}{1 - \gamma}\\
 &\leq \frac{2\gamma^{\min(H, \lfloor \mathcal V / 2\rfloor + 1)}}{1 - \gamma} + \frac{4\gamma^{\min(H, \lfloor \mathcal V / 2\rfloor + 1)}}{(1 - \gamma)^2} + \frac{\epsilon}{1 - \gamma}\\
 &\leq \frac{6\gamma^{\min(H, \lfloor \mathcal V / 2\rfloor + 1)}}{(1 - \gamma)^2} + \frac{\epsilon}{1 - \gamma}
\end{align*}
}

Here, the penultimate line uses again the Dependence Time Lemma (lemma \ref{dtl}).

\clearpage
\begin{figure*}[t!]
  \centering
  \setlength{\tabcolsep}{1pt}
  \setlength{\plotwidth}{0.20\linewidth}
  \setlength{\plotheight}{2.8cm}
  \begin{tabular}{@{}c@{}ccc@{}}
    &
    Visibility &
    Soft Constraint&
    Fallback
    \\[0.8em]
    \ylabeltp &
    \plotentry{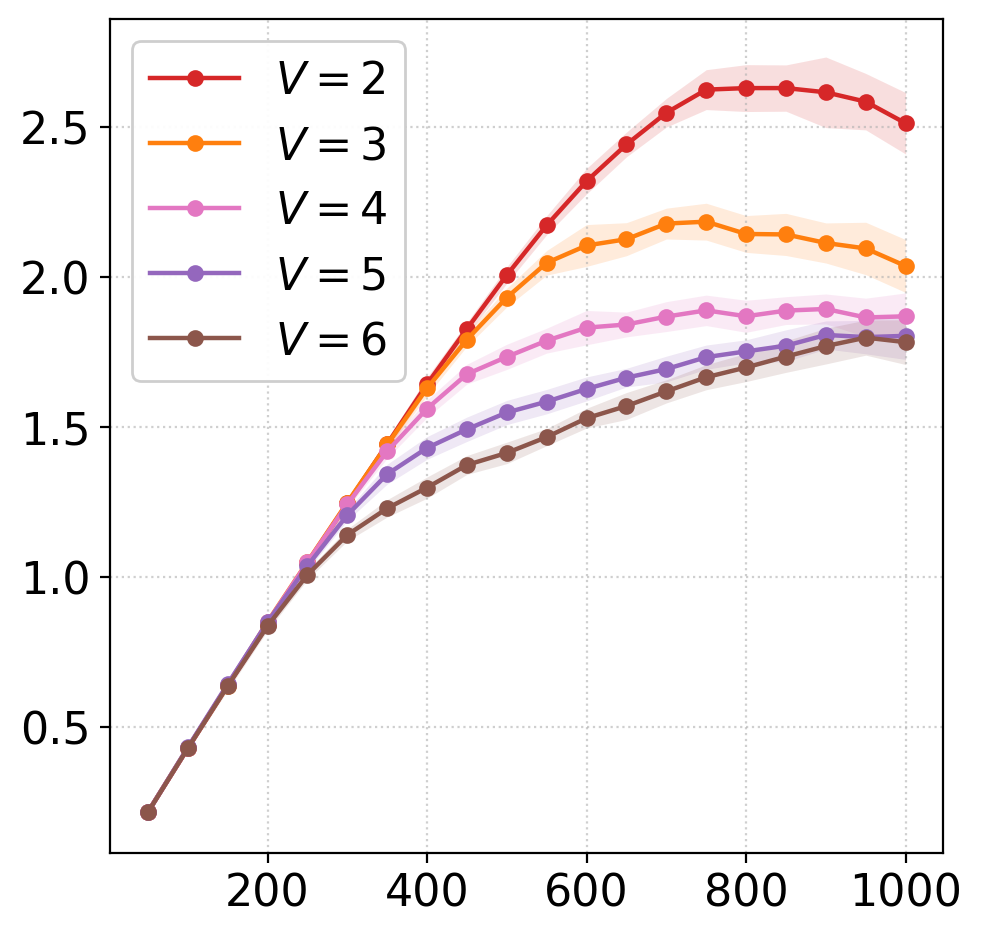} &
    \plotentry{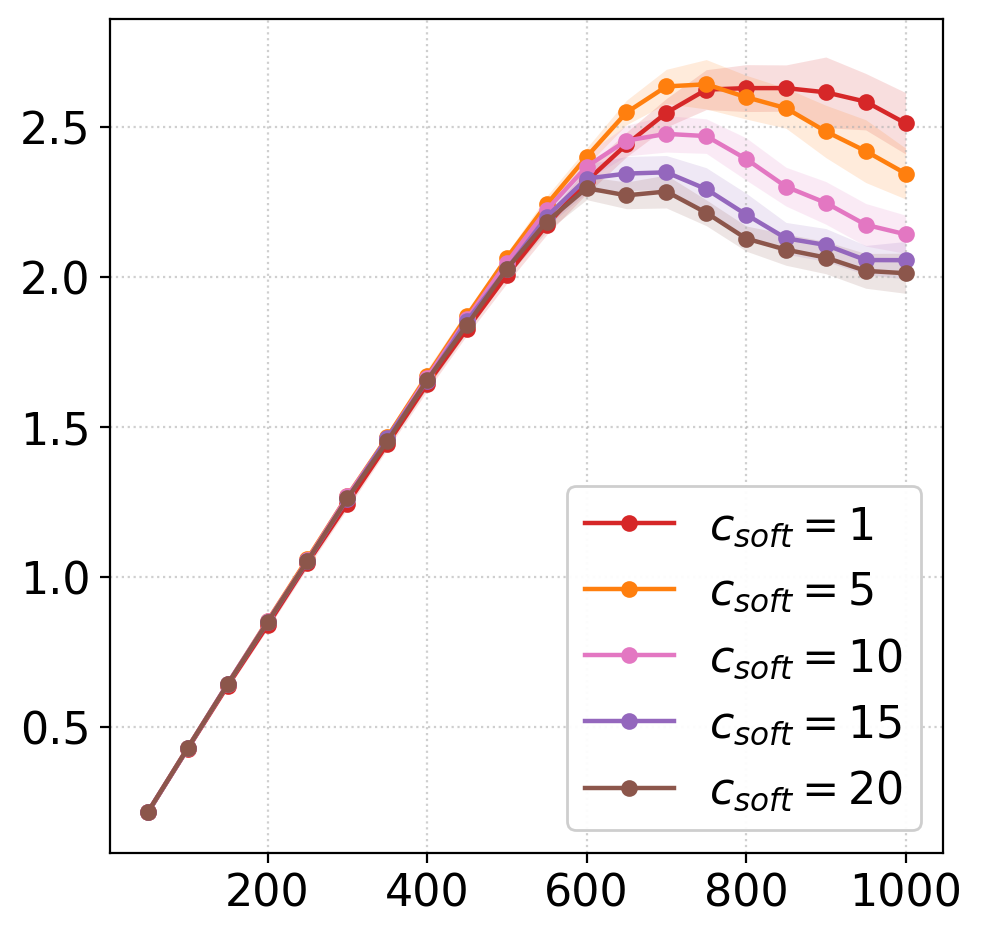}&
    \plotentry{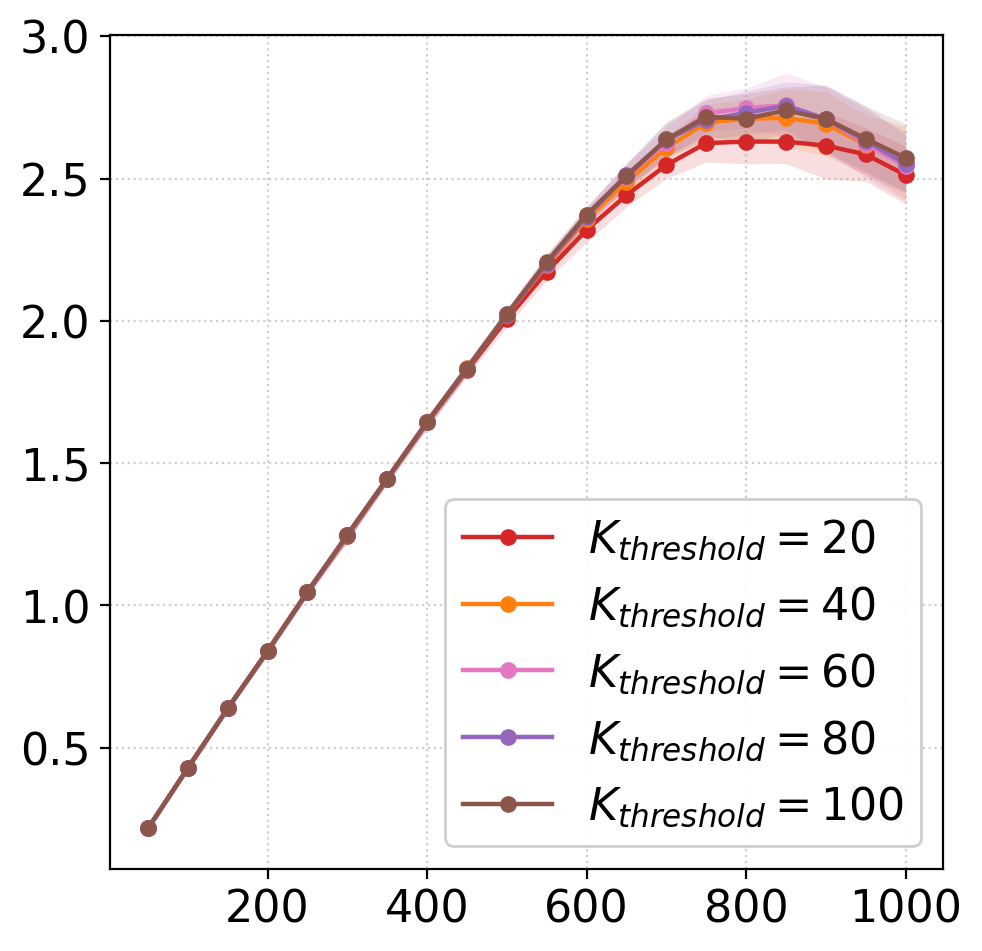}
    \\[0.3em]
    \ylabelrt &
    \plotentry{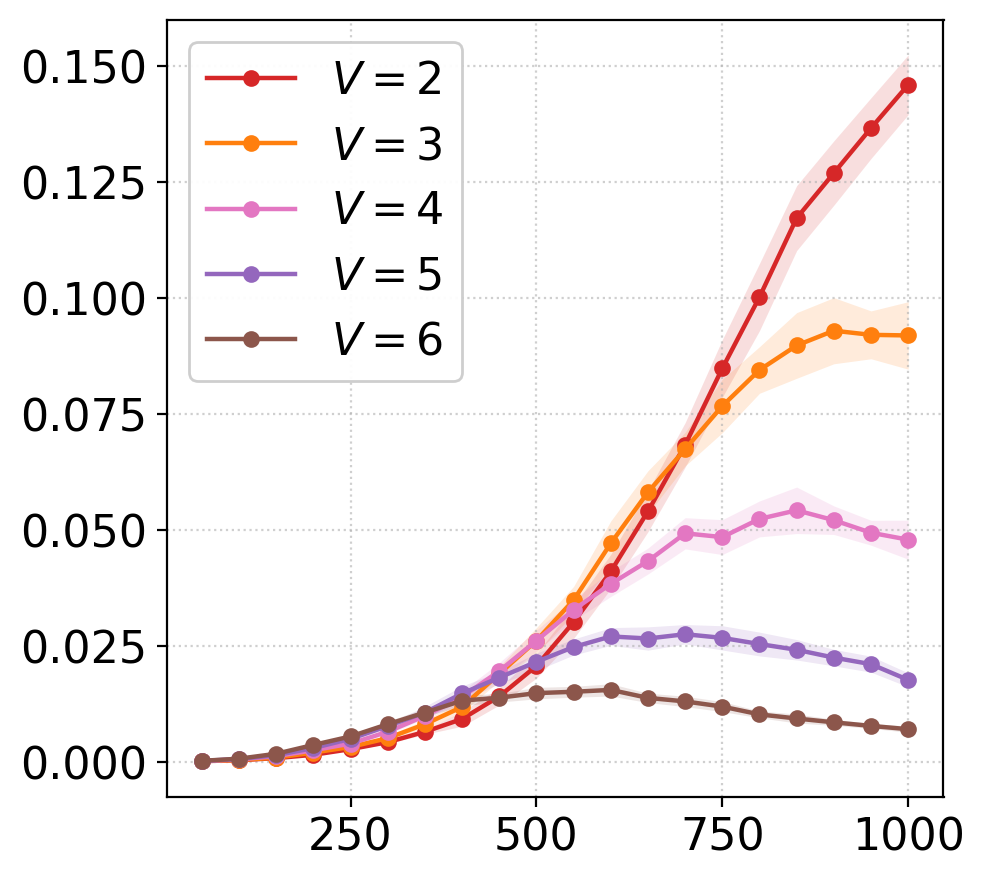} &
    \plotentry{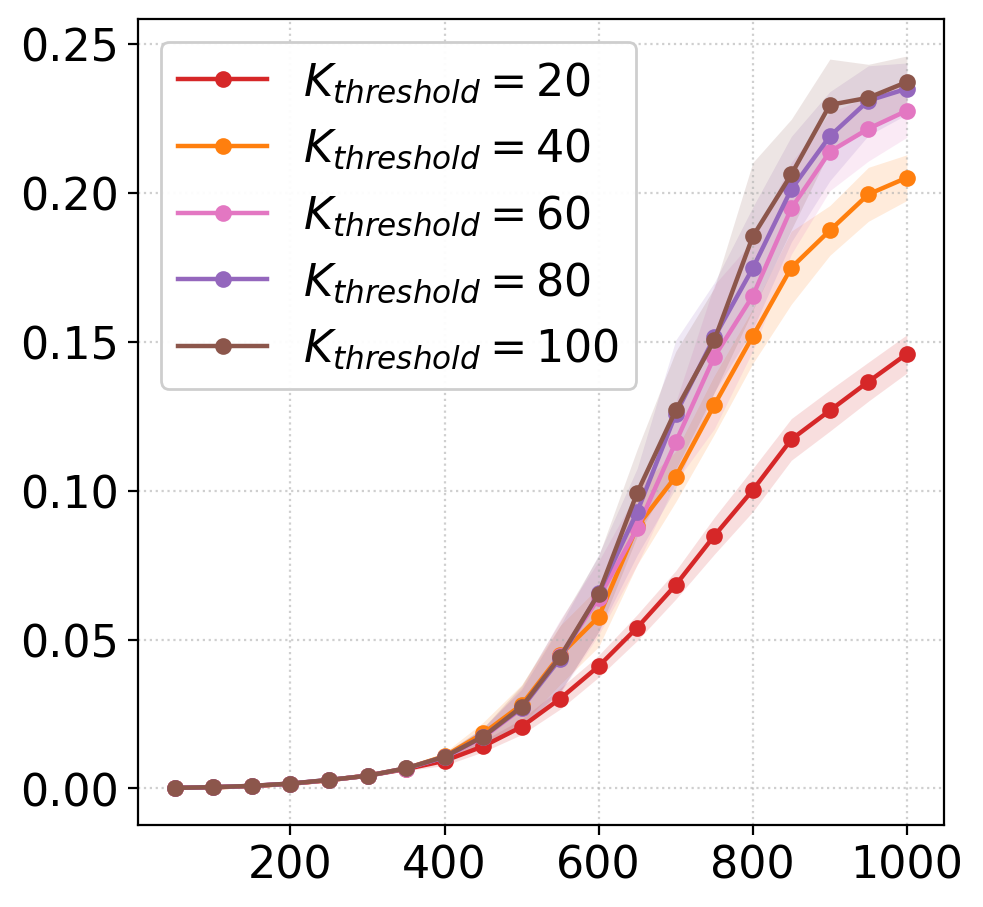}&
    \plotentry{images/warehouse1_fallback_sweep_mean_maxpt.png}
    \\[0.3em]
    & \multicolumn{3}{c}{\small agents}
  \end{tabular}
  \caption{Warehouse-10-20-10-2-1 abalation simulations: \textcolor{red}{red}, \textcolor{orange}{orange}, \textcolor{pink!70!red}{pink}, \textcolor{purple}{purple}, \textcolor{brown}{brown} is least to most of value (see legend)}
  \label{abalation_graphs}
\end{figure*}
\subsection{Abalations}
\label{abalations}
The outcome of our ablation simulations on warehouse-10-20-10-2-1 are shown in \cref{abalation_graphs}.

\subsubsection{Advantage of Smaller Visibilities}
When using different visibilities with a fixed $K_{threshold} = 20$, overall we observe that the algorithm hits this threshold sooner resulting a degradation of performance \cref{abalation_graphs}. Without this threshold, if we view RHCR as GD-RHCR with $V = \infty$ (ignoring a minor discrepancy with the replan window), we observed in previous simulations that in the region before RHCR collapse that the throughput remains close for most maps (aside from the smaller 32x32 maps) indicating the effectiveness of using smaller visibilities with soft constraints.

\subsubsection{Advantage of Smaller Soft Constraints}
When increasing $c_{soft}$, we observe a relatively fast degradation in the computation time which translates to a lower throughput when groups begin to time out.

\subsubsection{Advantage of Lower Fallback}
The simulation shows that the throughput is relatively agnostic to the increase in $K$ of $K_{threshold} = \min(0.1k, K)$. This is because for this map, the topology breaks agents into smaller groups which GD-RHCR continues to route into the large agent range and the PIBT fallback is effectively decongesting the larger groups.

\section{Simulation Times}
\label{appendix_wall}
In this section, we will show the scaled plan times and wall times for each of the following simulations shown in \cref{mapd_wall} and \cref{random_wall}. It is important to note for the following that the timeout cap placed on RHCR will cap the wall time but if this wall time is hit, this means RHCR has failed and is simply running the PIBT fallback.

Below, we have plotted both the scaled planning time and total wall times. The scaled planning time is simply the average plan time plots taken from \cref{mapd_graphs} and \cref{random_graphs} and scaled down the RHCR times by 1/5 representing the evaluation of RHCR happening 1/5th of the timesteps due to the simulation replan window $\alpha= 5$ that were used. We hope that by showing both these plots, we can get better insight into the contribution of the due to the fewer evaluations by RHCR and the error of the core contention and scheduling issues.

\textbf{Wall Times Are Faster}
Aside from some anomolous maps, we see that overall the wall times are consistently achieve at least a 2x improvement in computation time in the range before RHCR collapse. If we include the region of RHCR after collapse, this improvement ratio can often be made exponentially large as we increase the timeout time for RHCR. Also beyond the RHCR collapse region, RHCR is simply taking the fallback PIBT and not accomplishing any meaningful work.

The anomalous maps appear to be the smaller 32x32 random and room maps which is expected by our main discussion in \cref{simulations}, however, we see that warehouse-10-20-10-2-2 has an oddly large wall time that matches RHCR. Looking at the discrepancy between the scaled plan time and wall time gives a much clearer picture that this is mostly due to scheduling issues and core contention.

This primarily shows in the warehouse-10-20-10-2-2 map because the number of groups is large (the largest of all the simulations - see appendix \ref{group_stats}).

\textbf{Fallback Means Faster Wall Times}
The maps where GD-RHCR begins to rely on $SOLVER2$ (PIBT) more and cannot perform well in, we see that the computation time speeds up, sometimes significantly as in random-32-32-20 and room-32-32-var1. For larger agent counts, the groups join and become large and immediately enters PIBT, dropping the wall time. This graceful retreat to the fallback is compared to RHCR which has a computation time that is often exponential with the number of agents only to hit a timeout and falling back anyways.

\section{Group Statistics}
\label{group_stats}
A summary of group stastics for each simulation is provided in \cref{mapd_group} and \cref{random_groups}. We can see that for most maps the average number of groups increases then trends down with the number of agents. This makes sense as the density of agents increases on the map, they tend to form larger groups indicated by the larger average group size. We see clearly for the examples where GD-RHCR had throughput near PIBT (sortation-2, random-32-32-20, room-32-32-var1) the average number of groups is low towards the end range and has a high average max group size.

\newcommand{\ylabelscaled}{\rotatebox{90}{\hspace{2ex}\small scaled plan time[s]}}

\renewcommand{\ylabelwt}{\rotatebox{90}{\hspace{4ex}\small wall time [s]}}

\begin{figure*}[!t]
  \centering
  \setlength{\tabcolsep}{1pt}
\setlength{\plotheight}{3.1cm}
  \setlength{\plotwidth}{0.20\linewidth}
  \begin{tabular}{@{}c@{}cccc@{}}
    &
    \setlength{\mapwidth}{0.21\linewidth}
    \setlength{\mapheight}{2cm}
    \mapentry{maps/warehouse-10-20-10-2-1_placements2.png}{warehouse-10-20-10-2-1} &
    \setlength{\mapwidth}{0.21\linewidth}
    \setlength{\mapheight}{2cm}
    \mapentry{maps/warehouse-10-20-10-2-2_placements2.png}{warehouse-10-20-10-2-2} &
    \setlength{\mapheight}{2cm}
    \setlength{\mapwidth}{0.19\linewidth}
    \mapentry{maps/sortation_w50_h22_gx1_gy1_bx3_by3_skip8_lead2_placements.png}{sortation-1}&
    \setlength{\mapwidth}{0.21\linewidth}
    \setlength{\mapheight}{2cm}
    \mapentry{maps/sortation_w25_h11_gx2_gy2_bx3_by3_skip9_lead2_placements.png}{sortation-2}
    \\[0.8em]
    \\[0.3em]

    \ylabelscaled &
    \plotentry{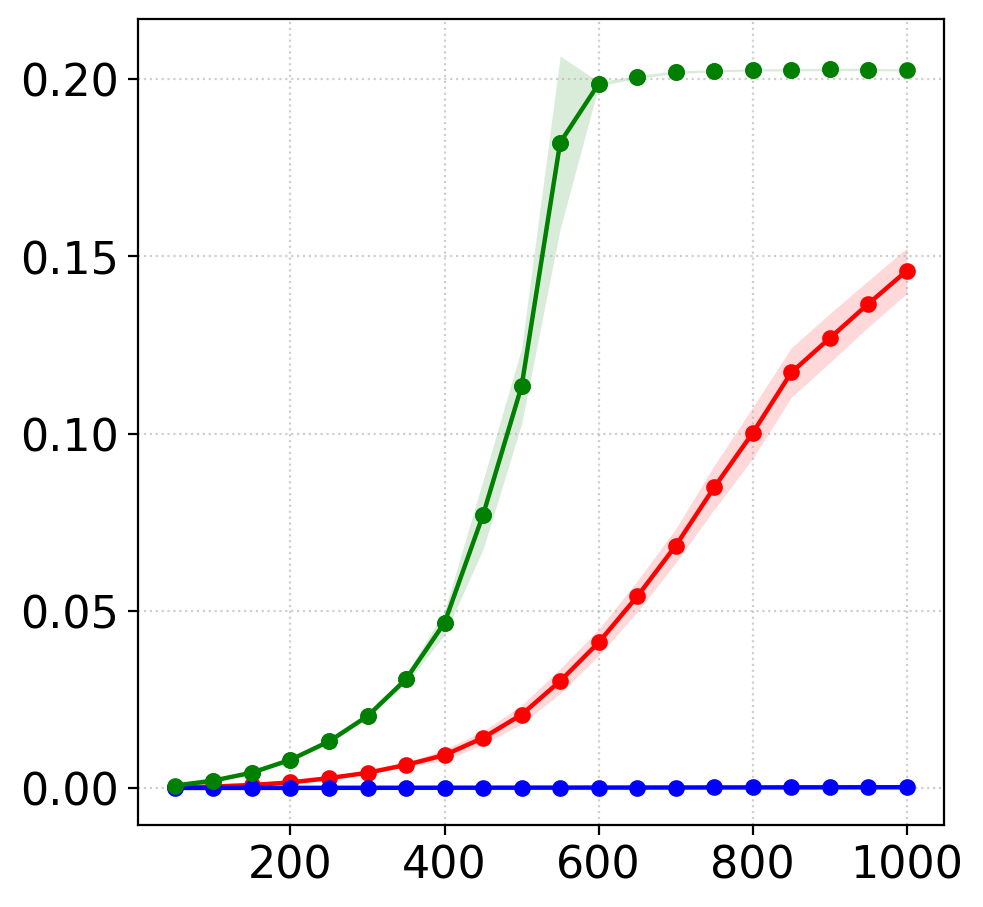} &
    \plotentry{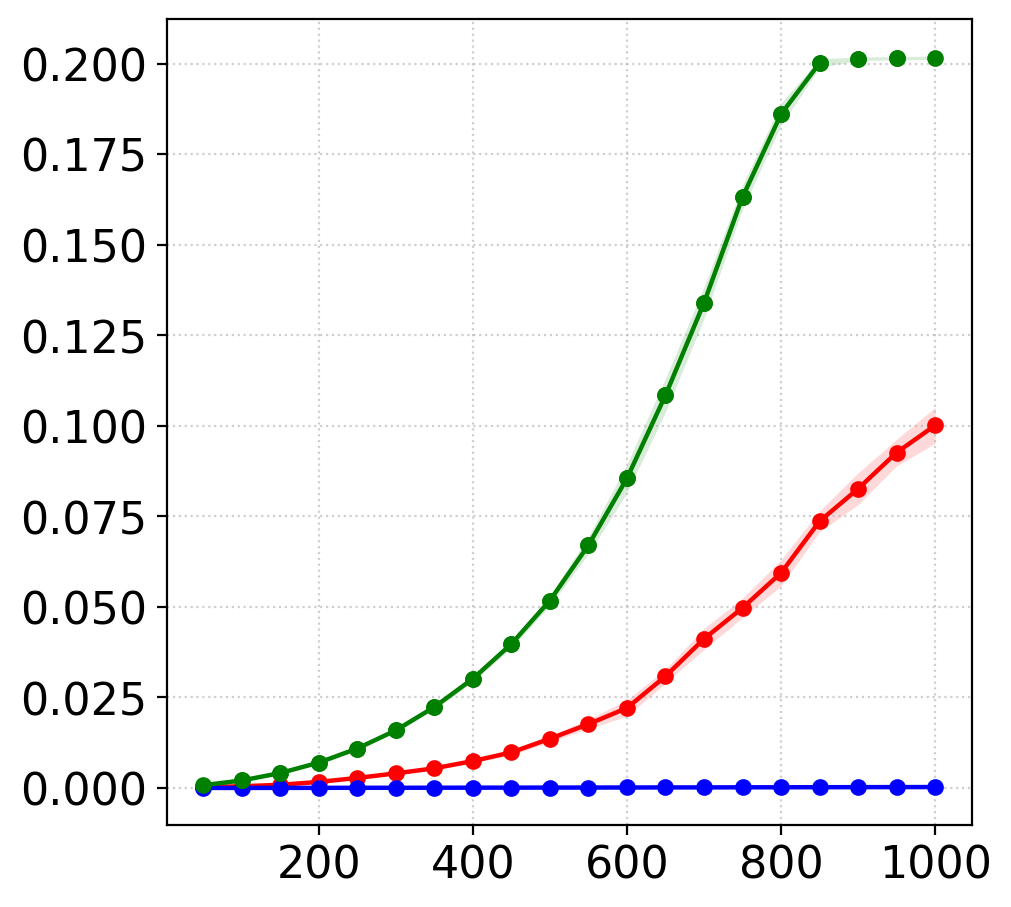} &
    \plotentry{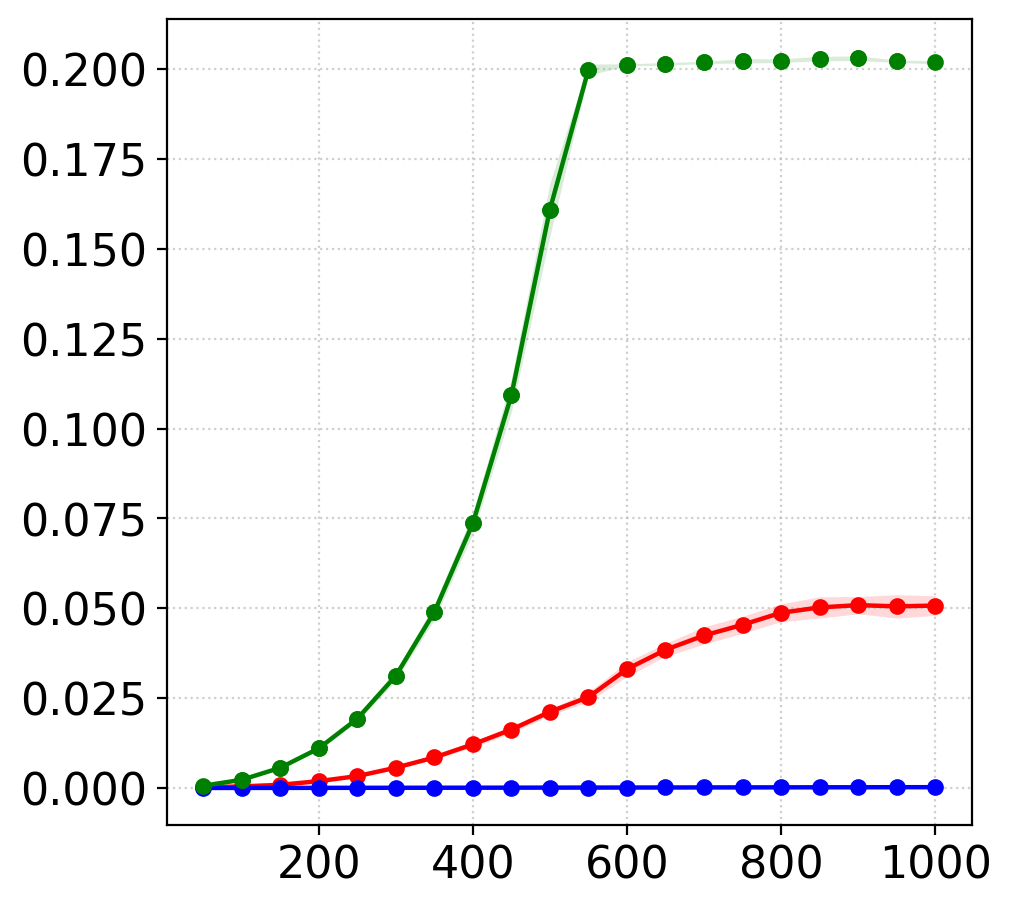} &
    \plotentry{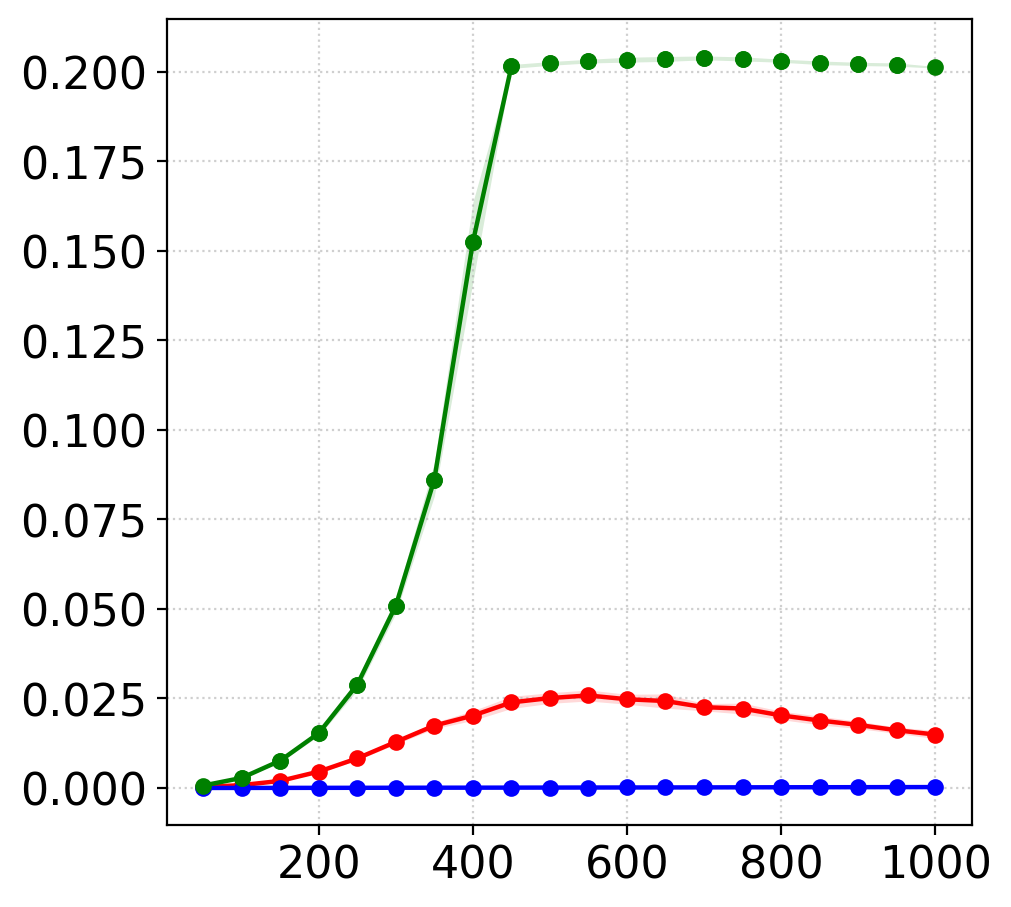}
    \\
    \\[0.2em]
    \ylabelwt &
    \plotentry{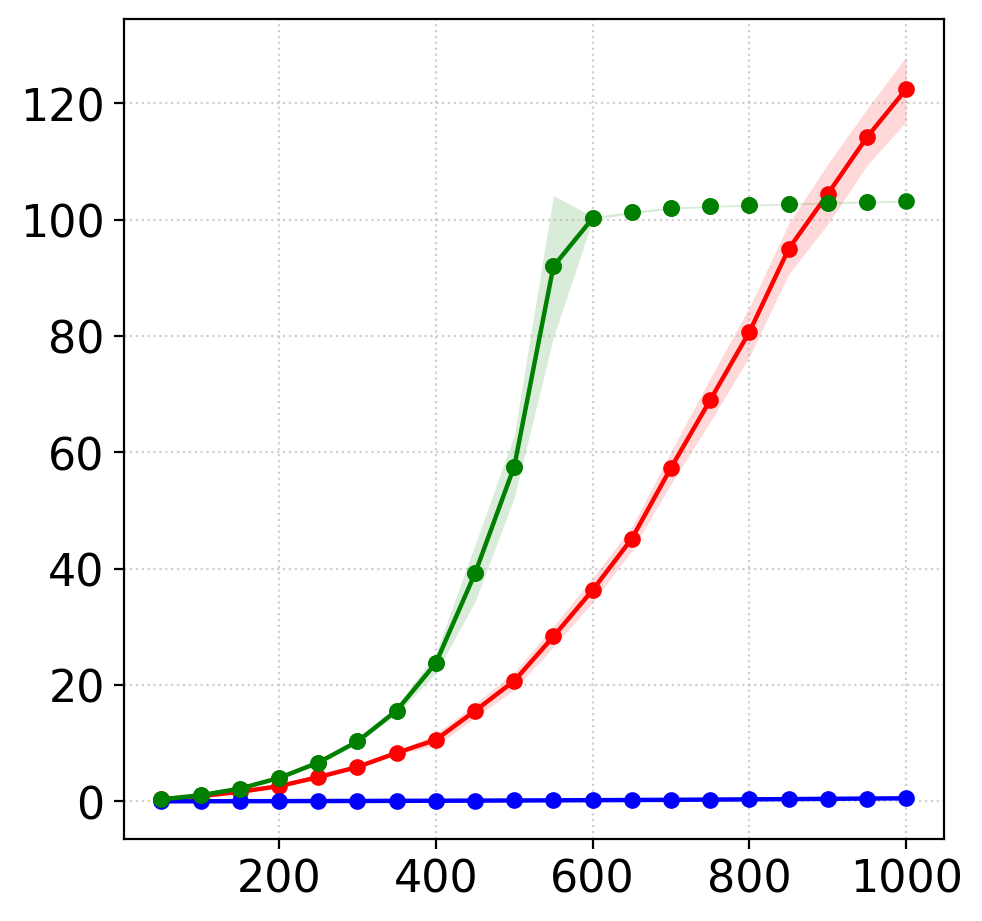} &
    \plotentry{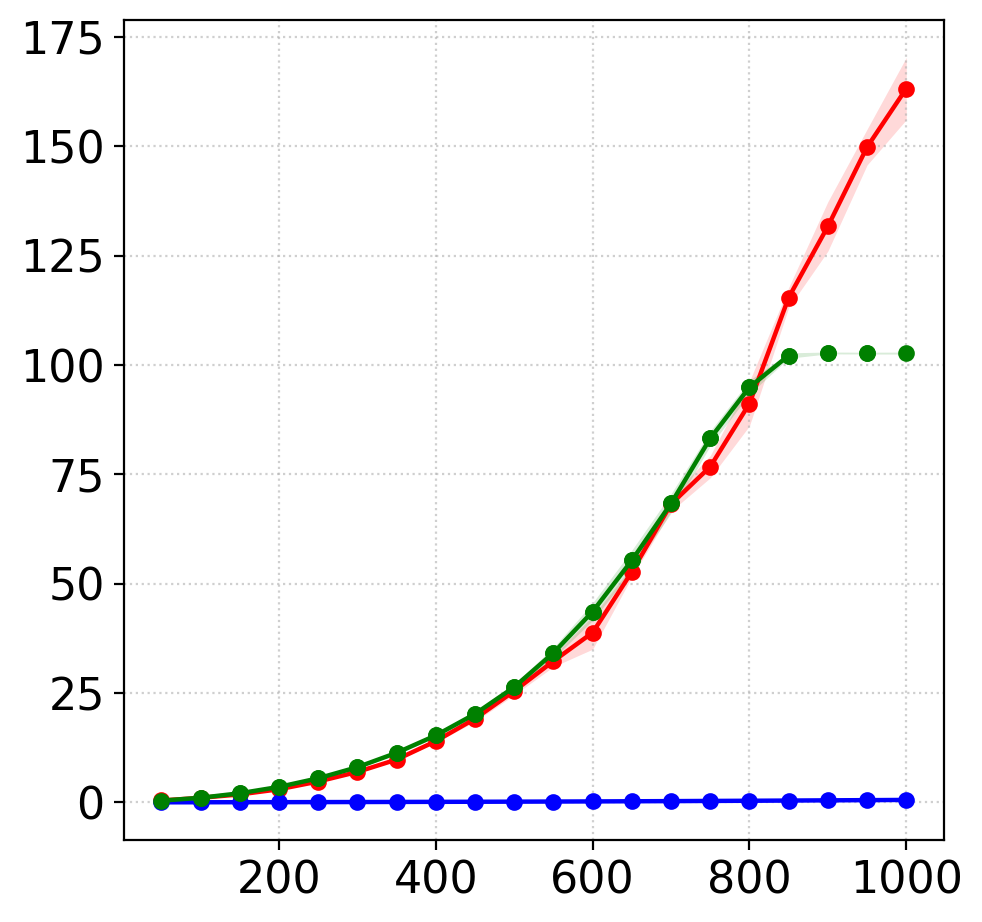} &
    \plotentry{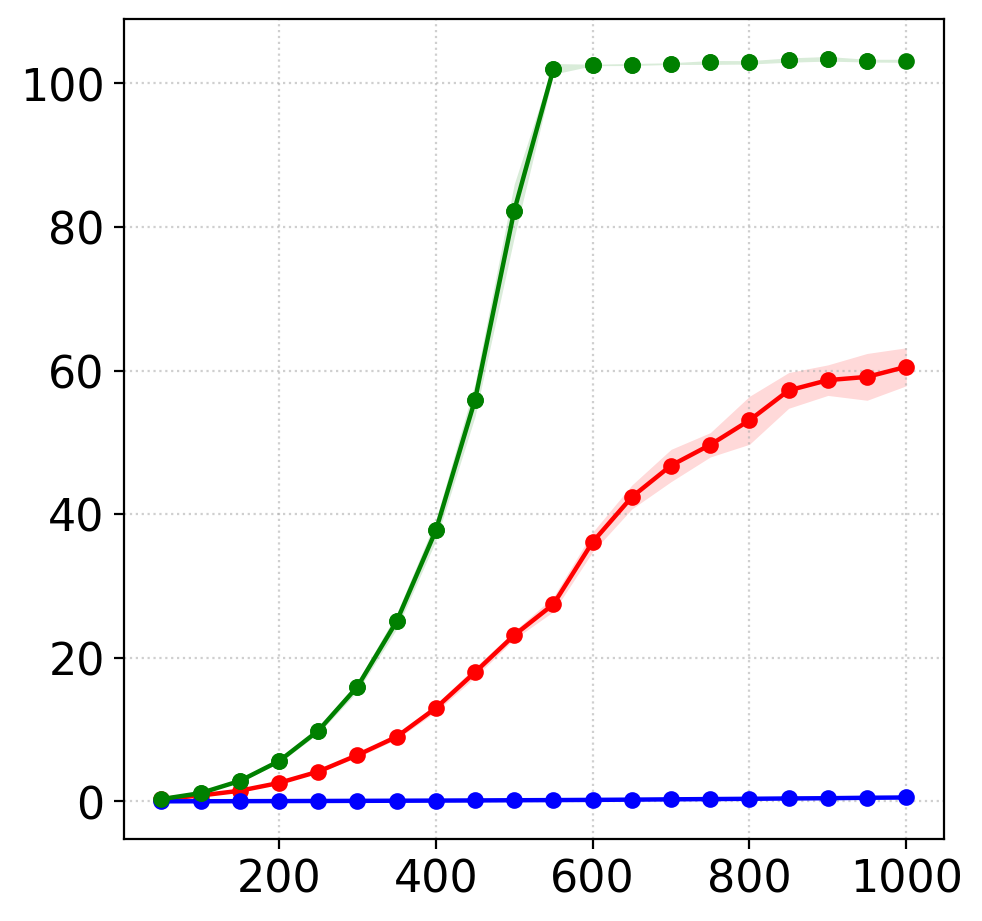} &
    \plotentry{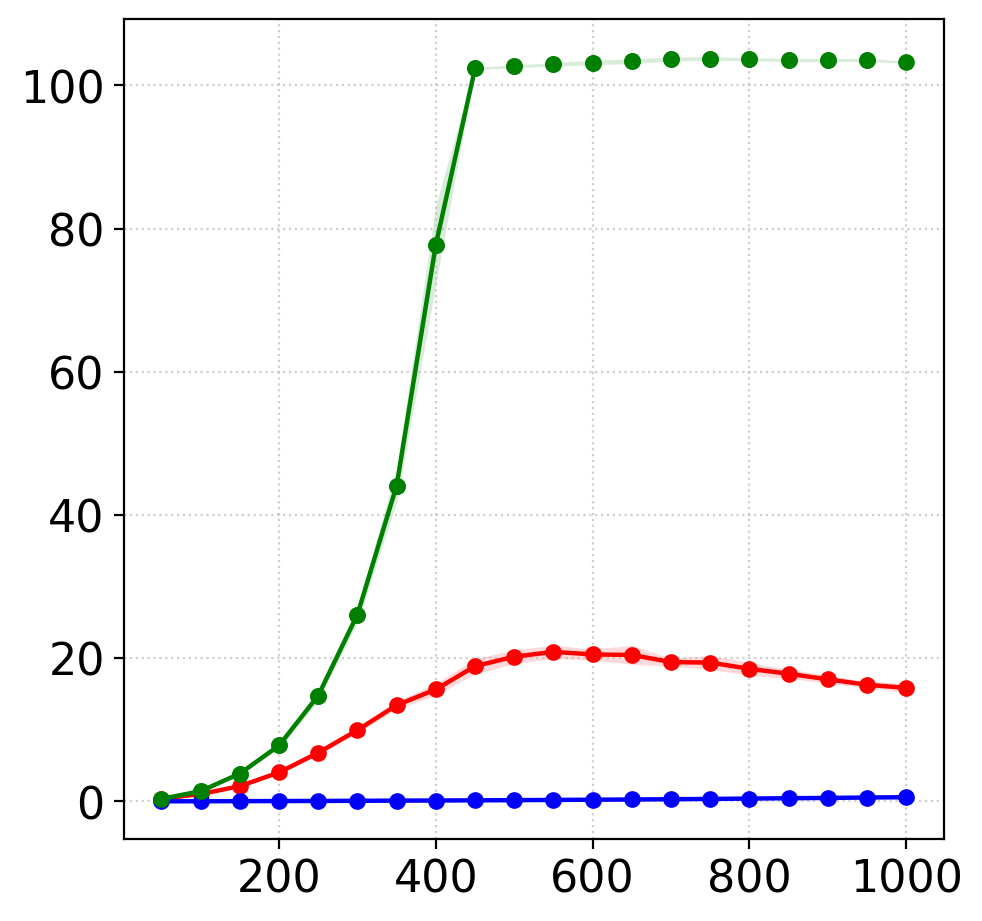}
    \\
    & \multicolumn{4}{c}{\small agents}
    \\[0.2em]
  \end{tabular}
  \caption{Wall times for MAPD. The color coding is \textcolor{red}{GD-RHCR}, \textcolor{green!50!black}{RHCR}, and \textcolor{blue}{PIBT}}
  \label{mapd_wall}
\end{figure*}

\begin{figure*}[!ht]
  \centering
  \setlength{\tabcolsep}{1pt}
  \setlength{\plotwidth}{0.16\linewidth}
  \setlength{\plotheight}{2.8cm}
  \begin{tabular}{@{}c@{}cccccc@{}}
    \mapentry{maps/random-64-64-20_hotspots_full.png}{random-64-64-20} &
    \mapentry{maps/room-64-64-uniform_1_hotspots_full.png}{room-64-64-var1} &
    \mapentry{maps/room-64-64-uniform_2_hotspots_full.png}{room-64-64-var2}&
    \mapentry{maps/random-32-32-20_hotspots_full}{random-32-32-20} &
    \mapentry{maps/room-32-32-uniform_1_hotspots_full.png}{room-32-32-var1} &
    \mapentry{maps/empty-48-48_hotspots_full.png}{empty-48-48} &
    \\[0.8em]
    \ylabelscaled
    \plotentry{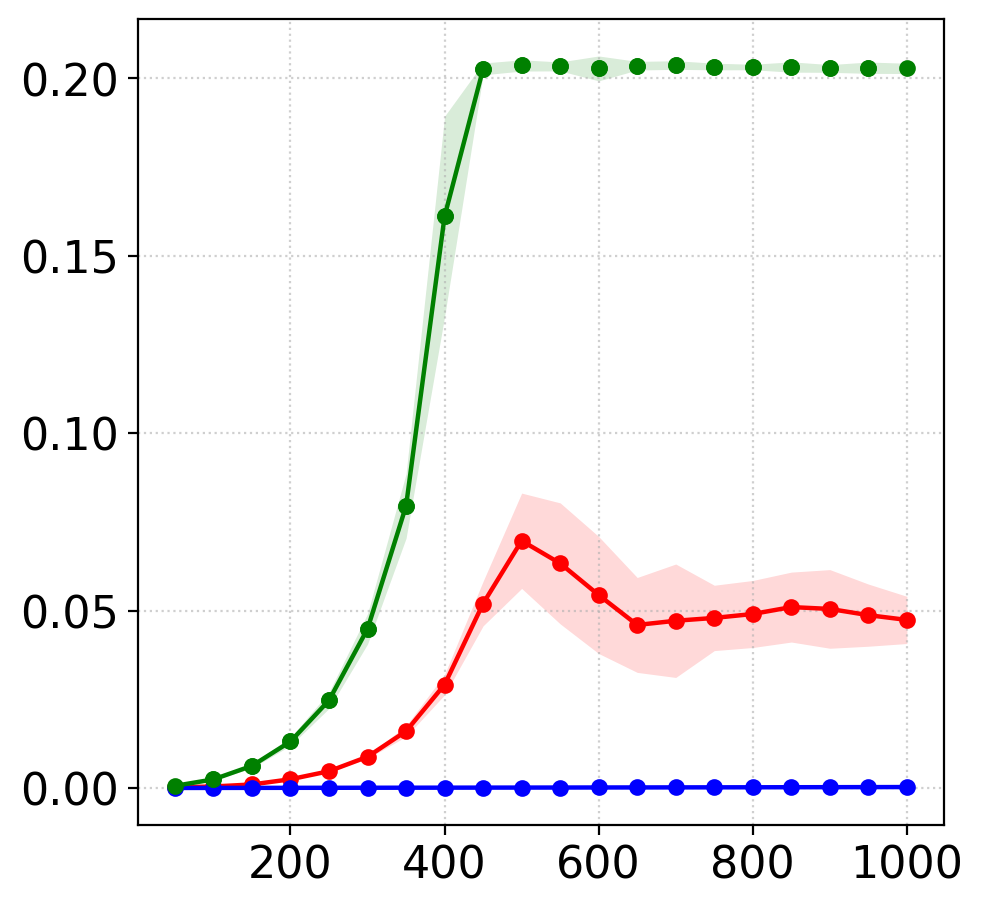} &
    \plotentry{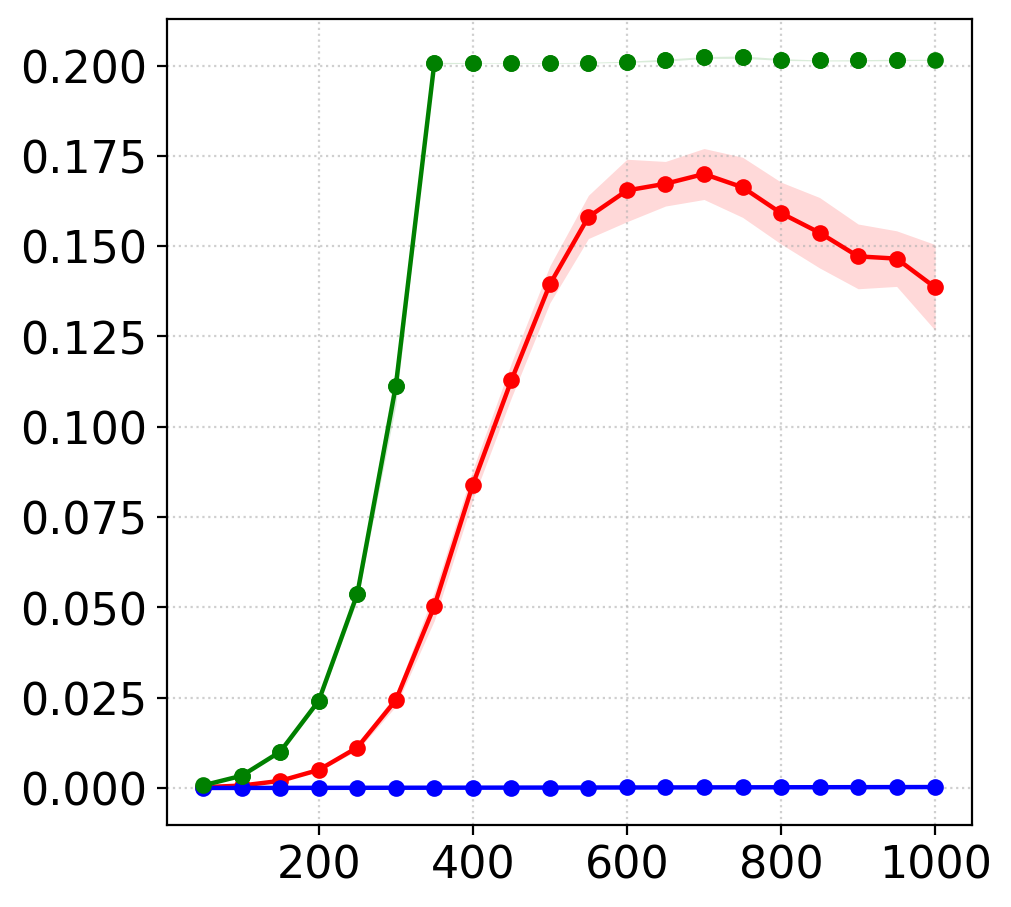}&
    \plotentry{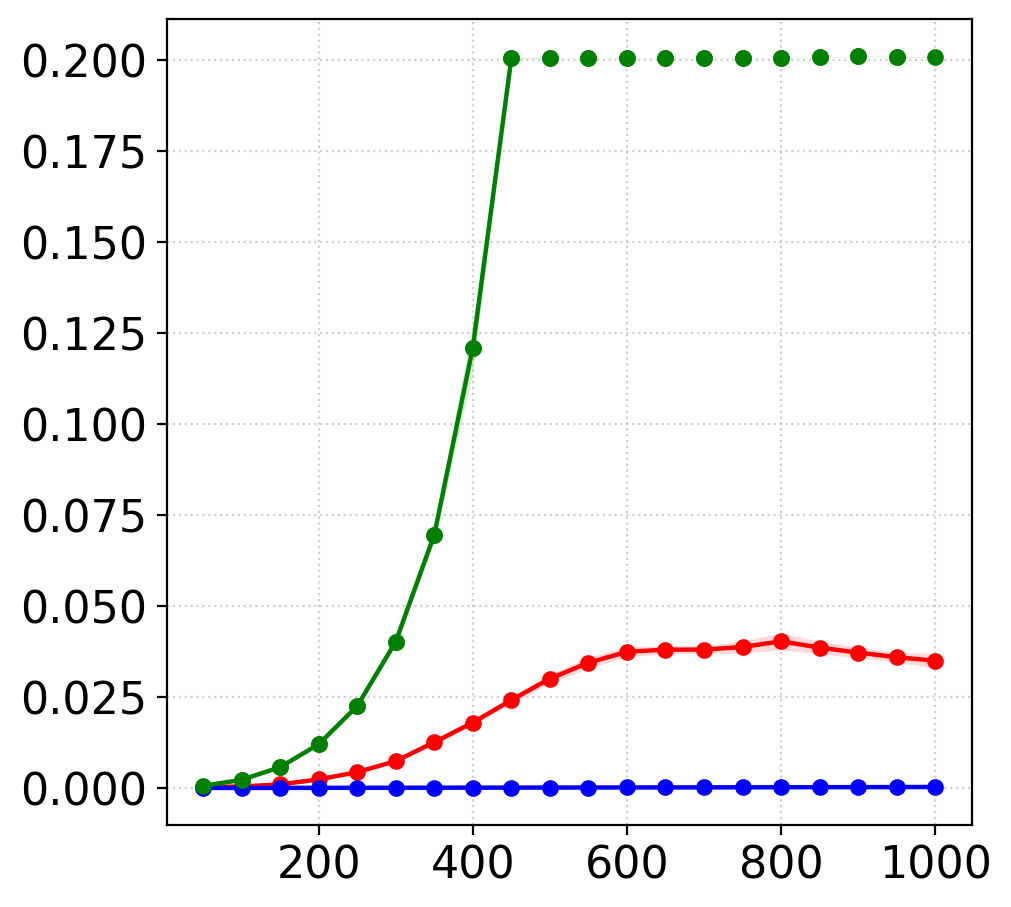}&
    \plotentry{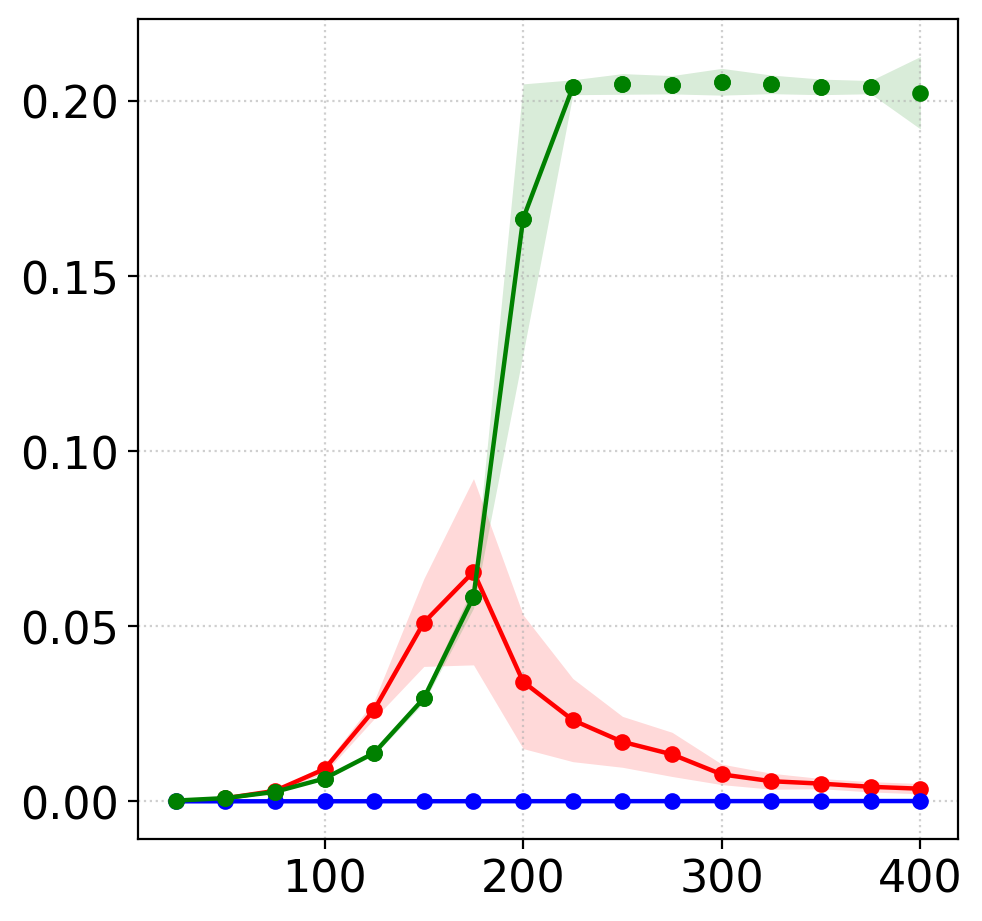} &
    \plotentry{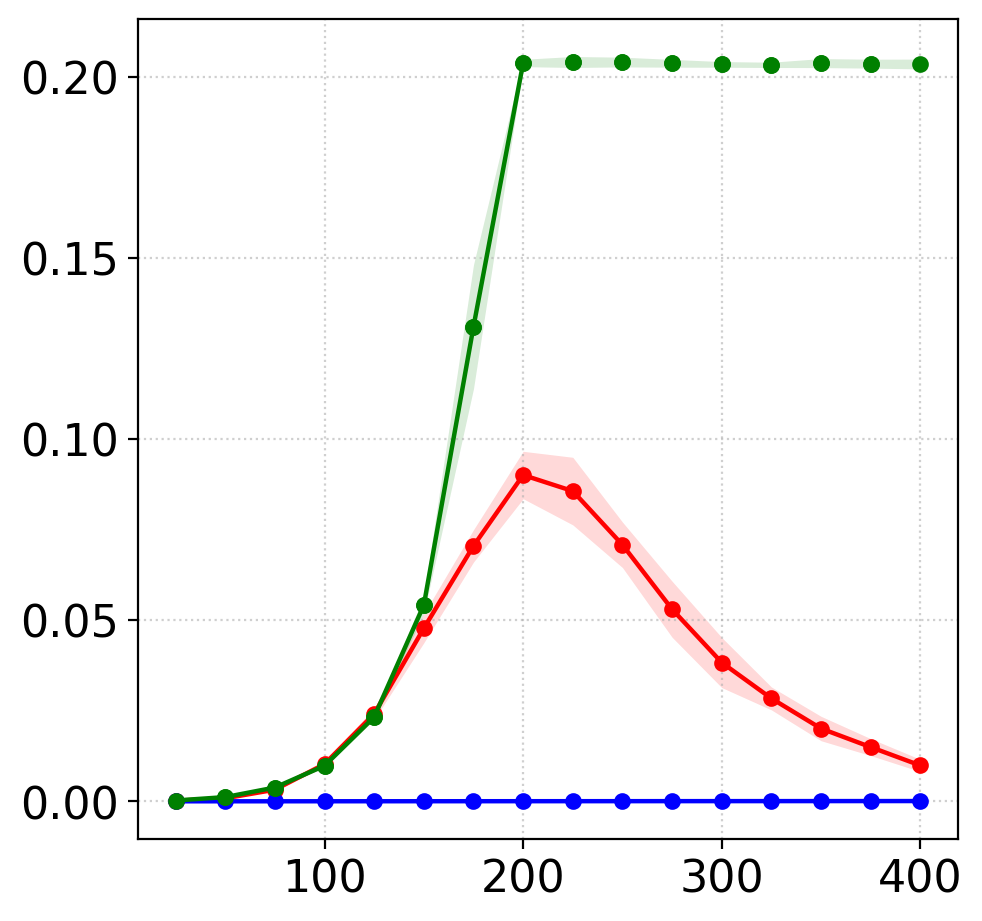} &
    \plotentry{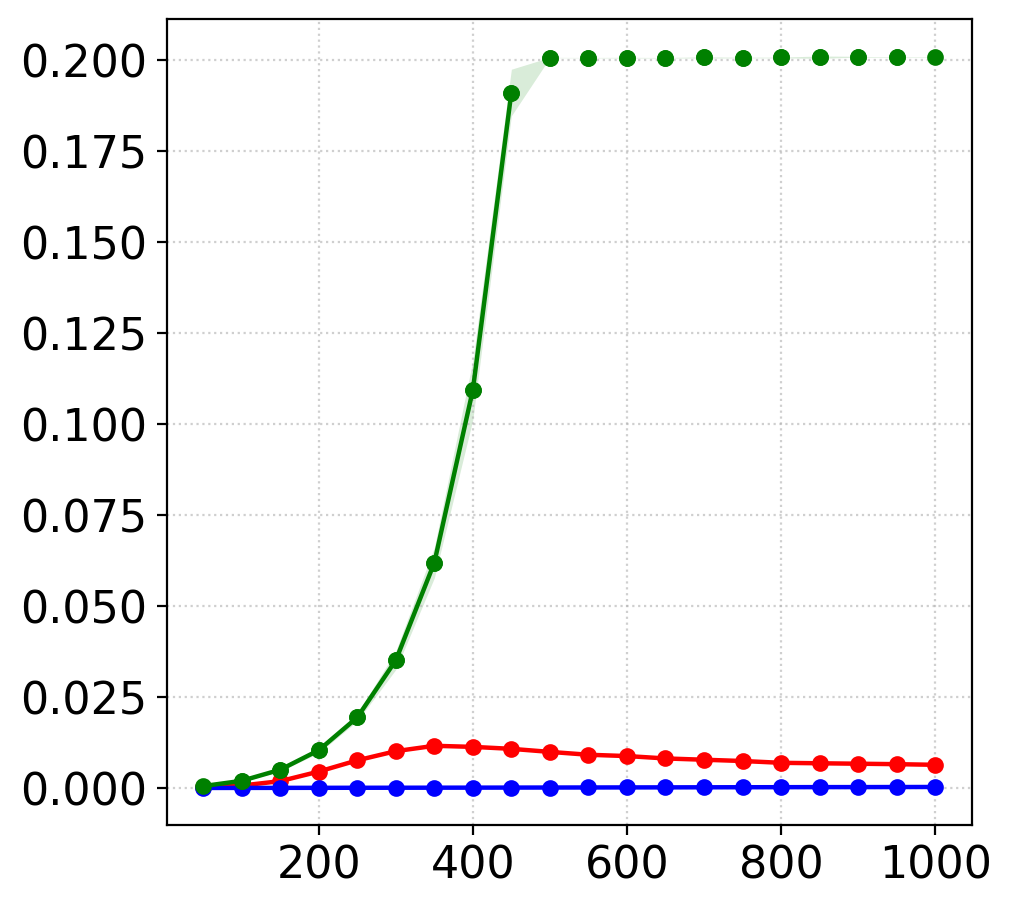} &
    \\
    \\[0.2em]
    \ylabelwt
    \plotentry{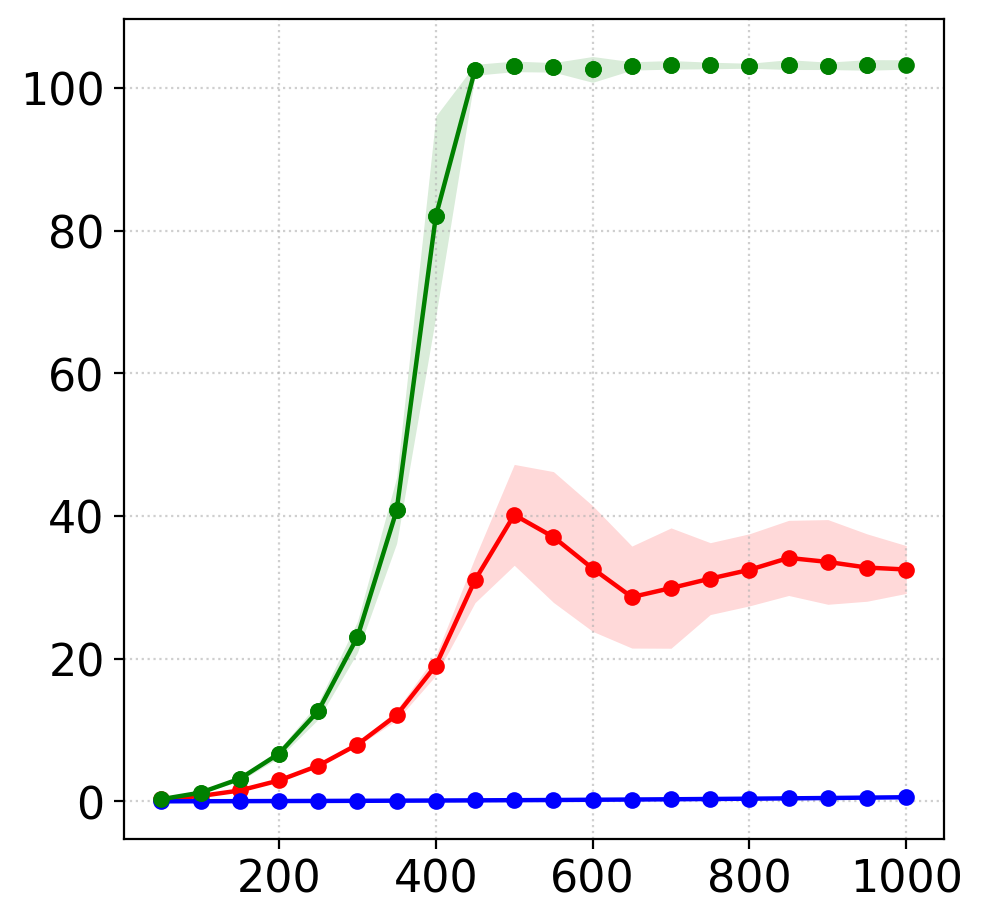} &
    \plotentry{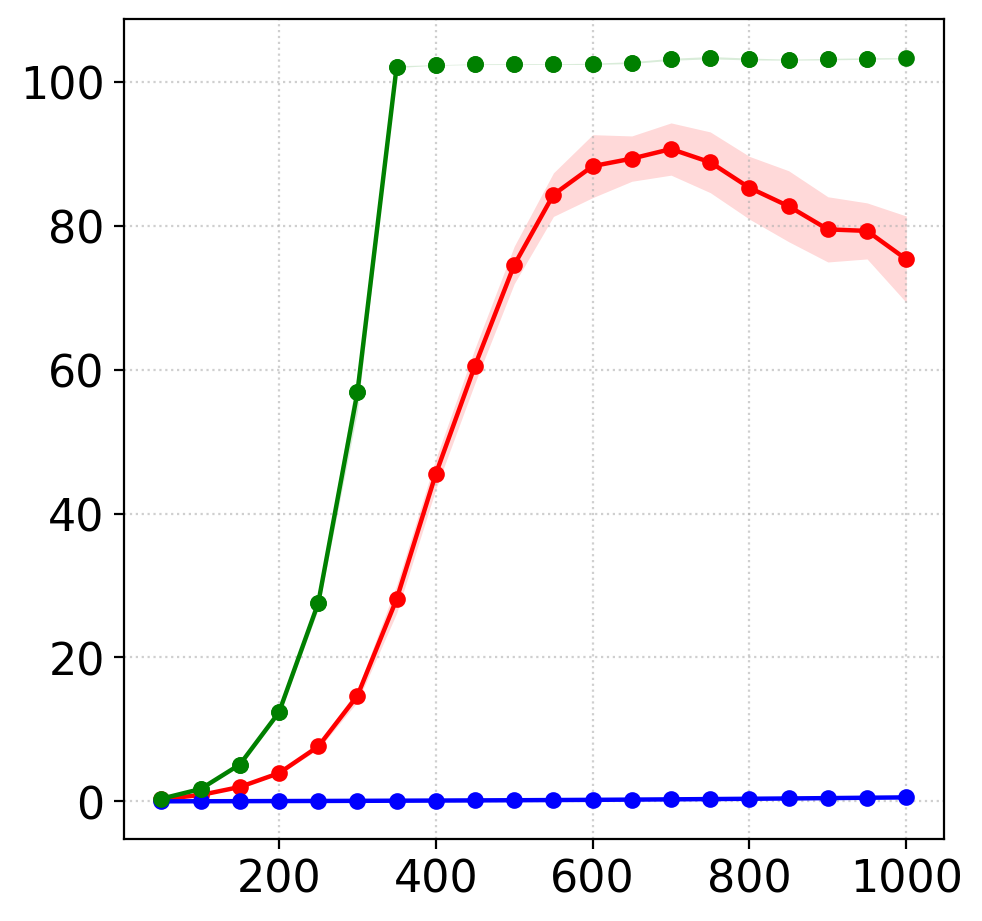}&
    \plotentry{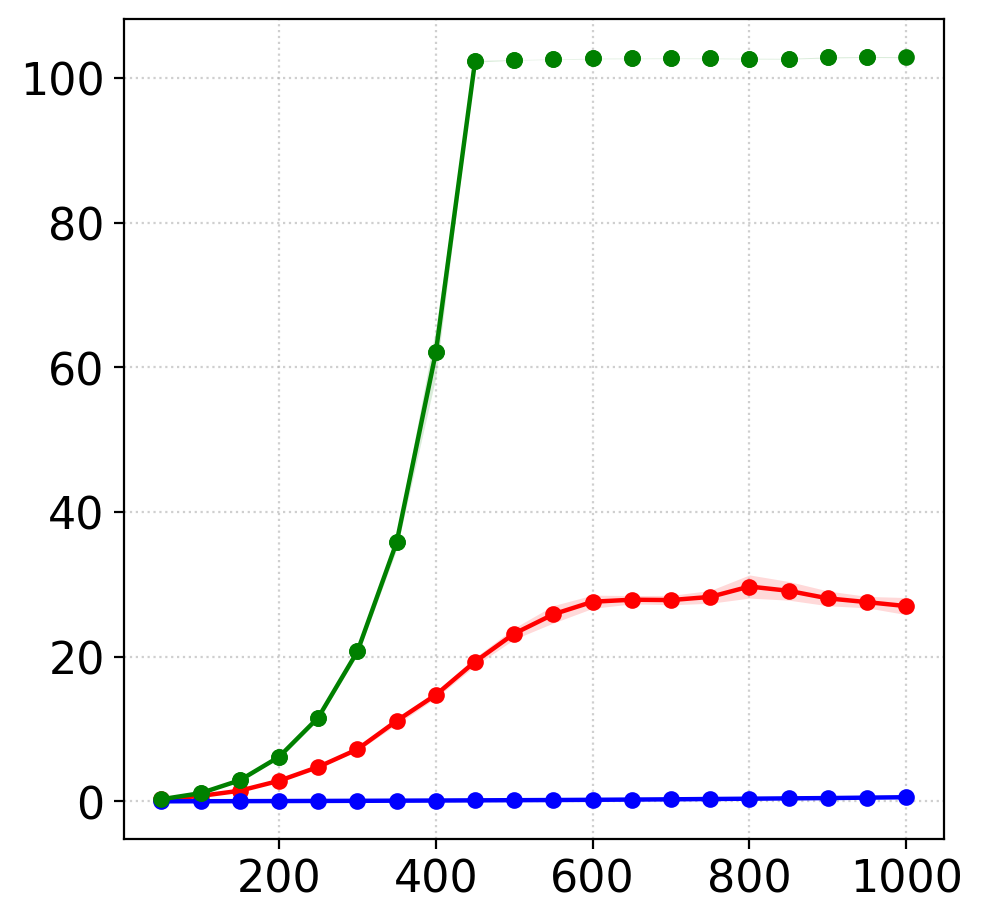}&
    \plotentry{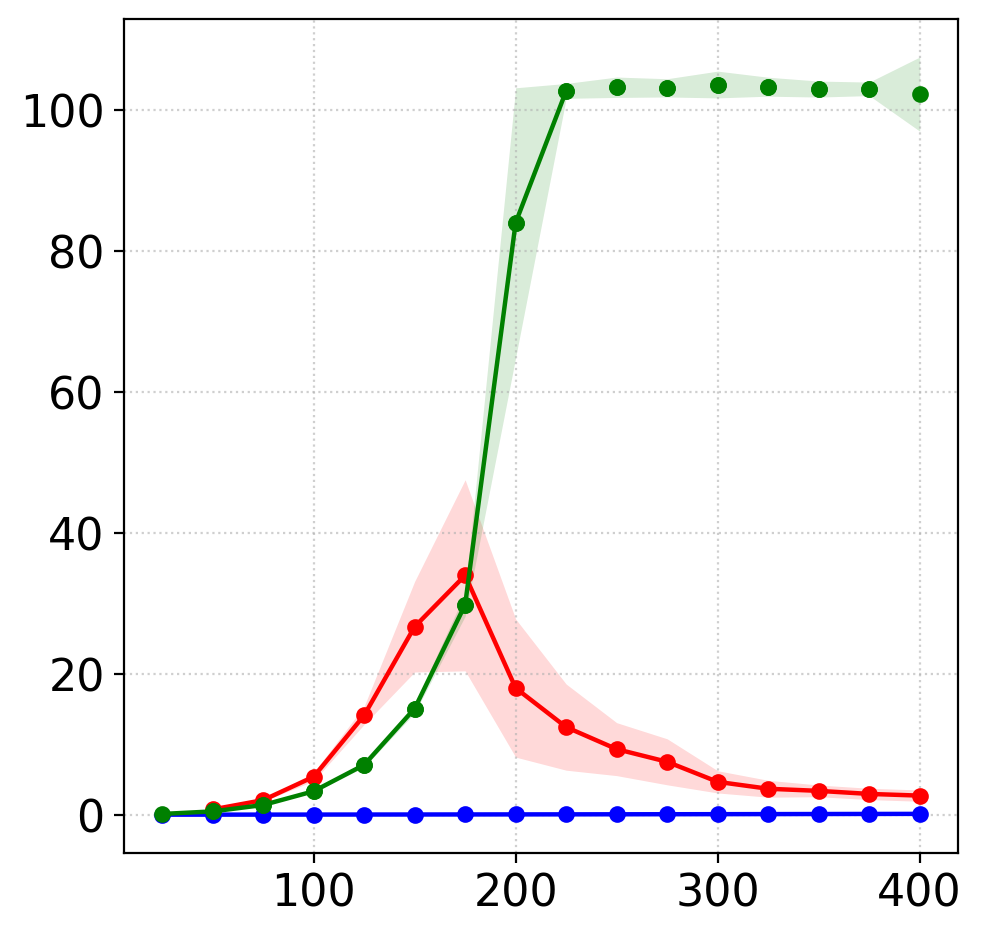} &
    \plotentry{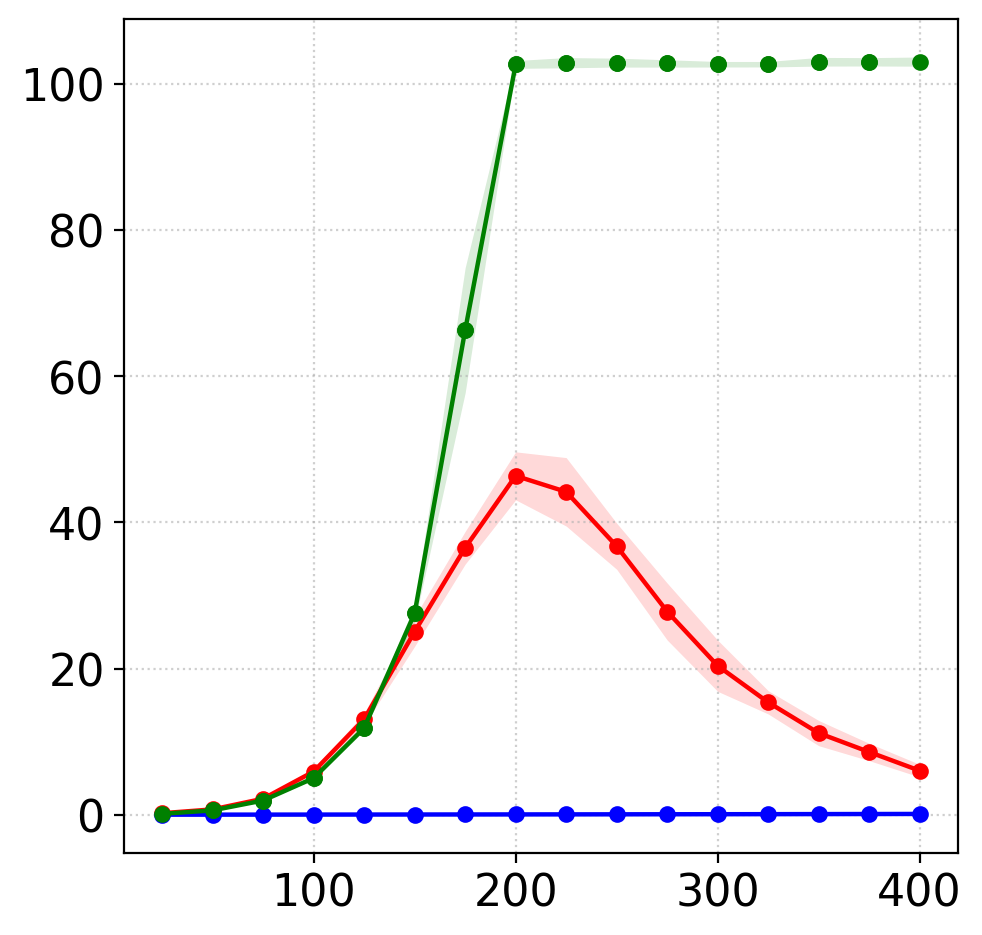} &
    \plotentry{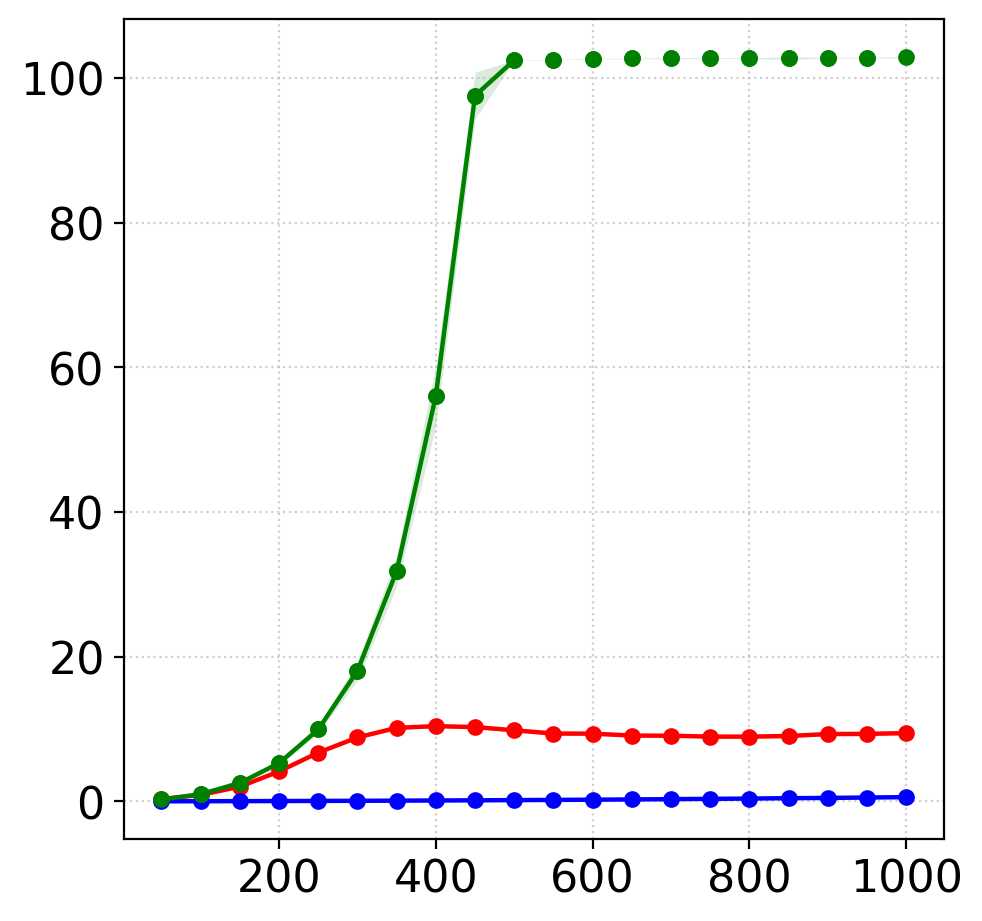} &
    \\
    & \multicolumn{4}{c}{agents}
  \end{tabular}
  \caption{Wall times for random navigation. The color coding is \textcolor{red}{GD-RHCR}, \textcolor{green!50!black}{RHCR}, and \textcolor{blue}{PIBT}}
  \label{random_wall}
\end{figure*}

\newcommand{\ylabelnumgroups}{\rotatebox{90}{\hspace{3ex}\small mean num groups}}
\newcommand{\ylabelmeangroupsize}{\rotatebox{90}{\hspace{3ex} \small mean group size}}
\newcommand{\ylabelmaxgroupsize}{\rotatebox{90}{\hspace{1ex}\small mean max group size}}

\begin{figure*}[!t]
  \centering
  \setlength{\tabcolsep}{1pt}
\setlength{\plotheight}{3.1cm}
  \setlength{\plotwidth}{0.20\linewidth}
  \begin{tabular}{@{}c@{}cccc@{}}
    &
    \setlength{\mapwidth}{0.21\linewidth}
    \setlength{\mapheight}{2cm}
    \mapentry{maps/warehouse-10-20-10-2-1_placements2.png}{warehouse-10-20-10-2-1} &
    \setlength{\mapwidth}{0.21\linewidth}
    \setlength{\mapheight}{2cm}
    \mapentry{maps/warehouse-10-20-10-2-2_placements2.png}{warehouse-10-20-10-2-2} &
    \setlength{\mapheight}{2cm}
    \setlength{\mapwidth}{0.19\linewidth}
    \mapentry{maps/sortation_w50_h22_gx1_gy1_bx3_by3_skip8_lead2_placements.png}{sortation-1}&
    \setlength{\mapwidth}{0.21\linewidth}
    \setlength{\mapheight}{2cm}
    \mapentry{maps/sortation_w25_h11_gx2_gy2_bx3_by3_skip9_lead2_placements.png}{sortation-2}
    \\[0.8em]
    \\[0.3em]

    \ylabelnumgroups &
    \plotentry{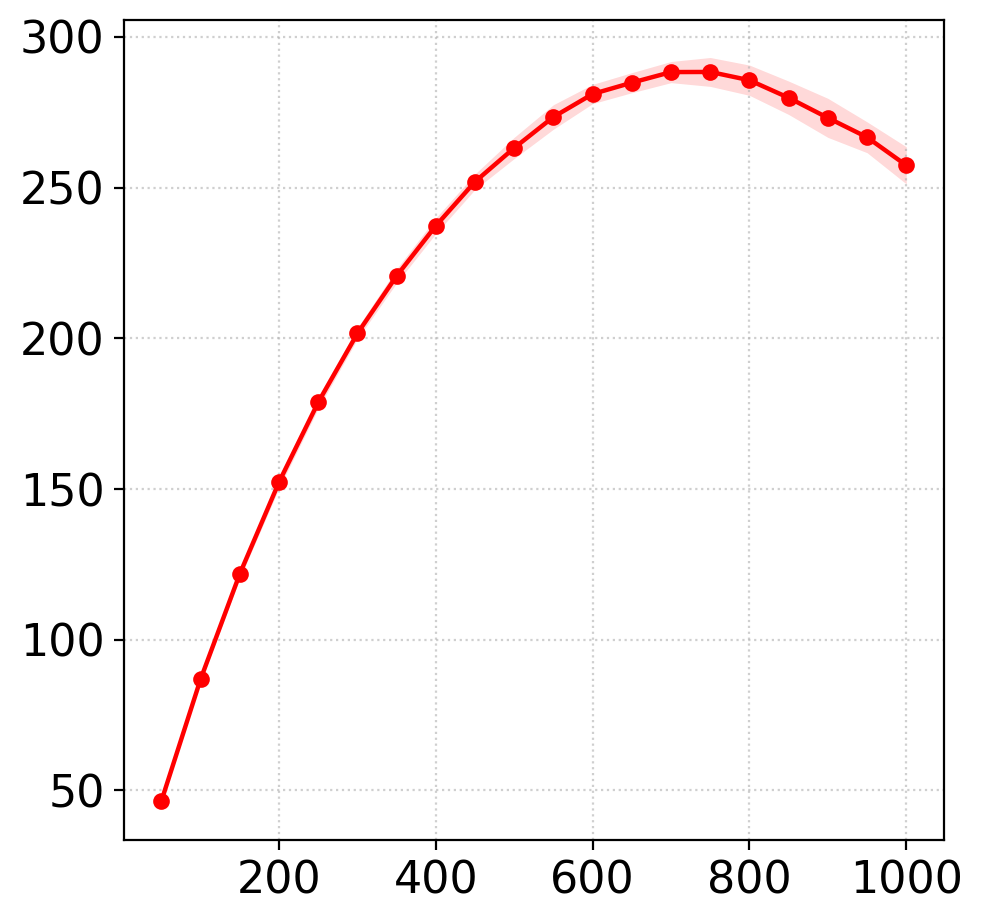} &
    \plotentry{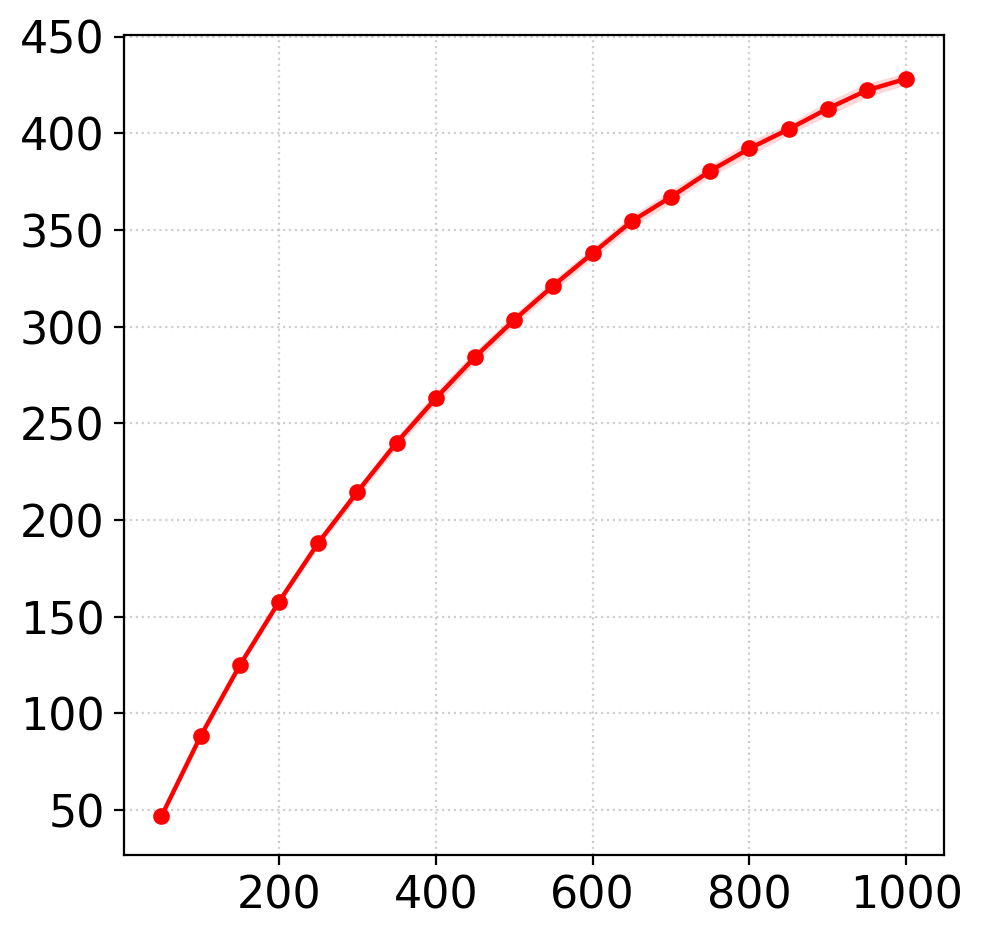} &
    \plotentry{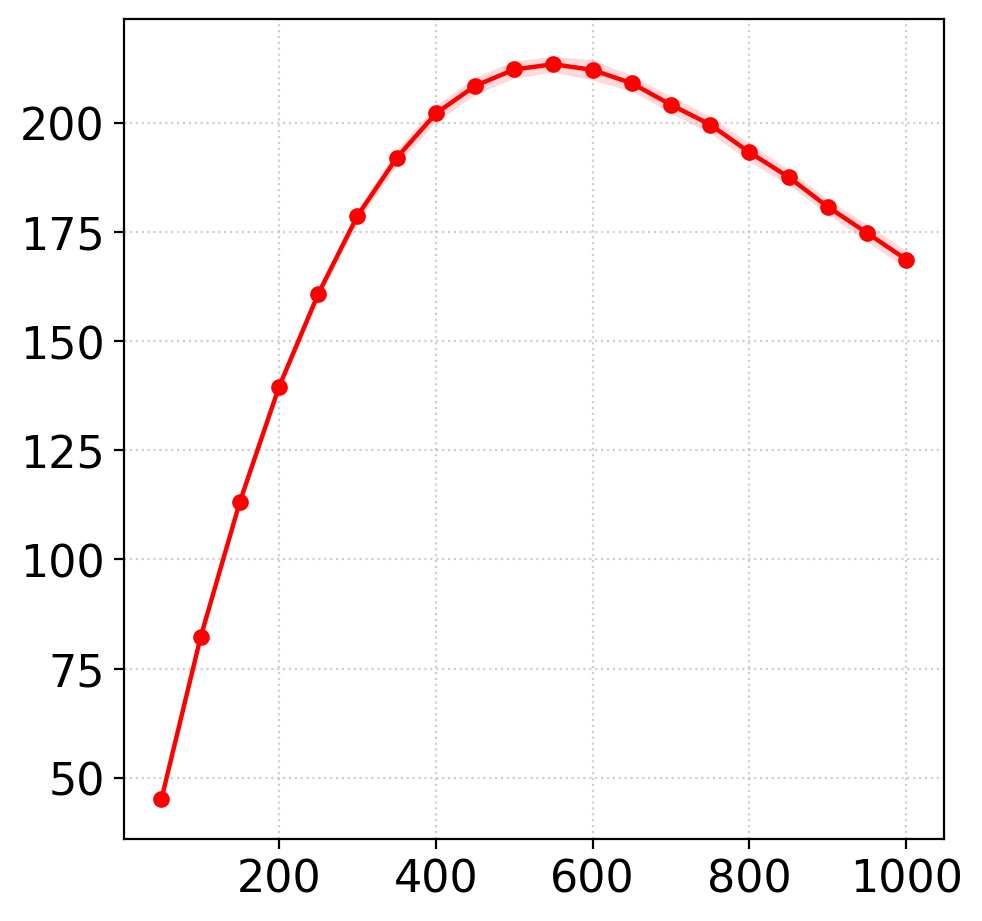} &
    \plotentry{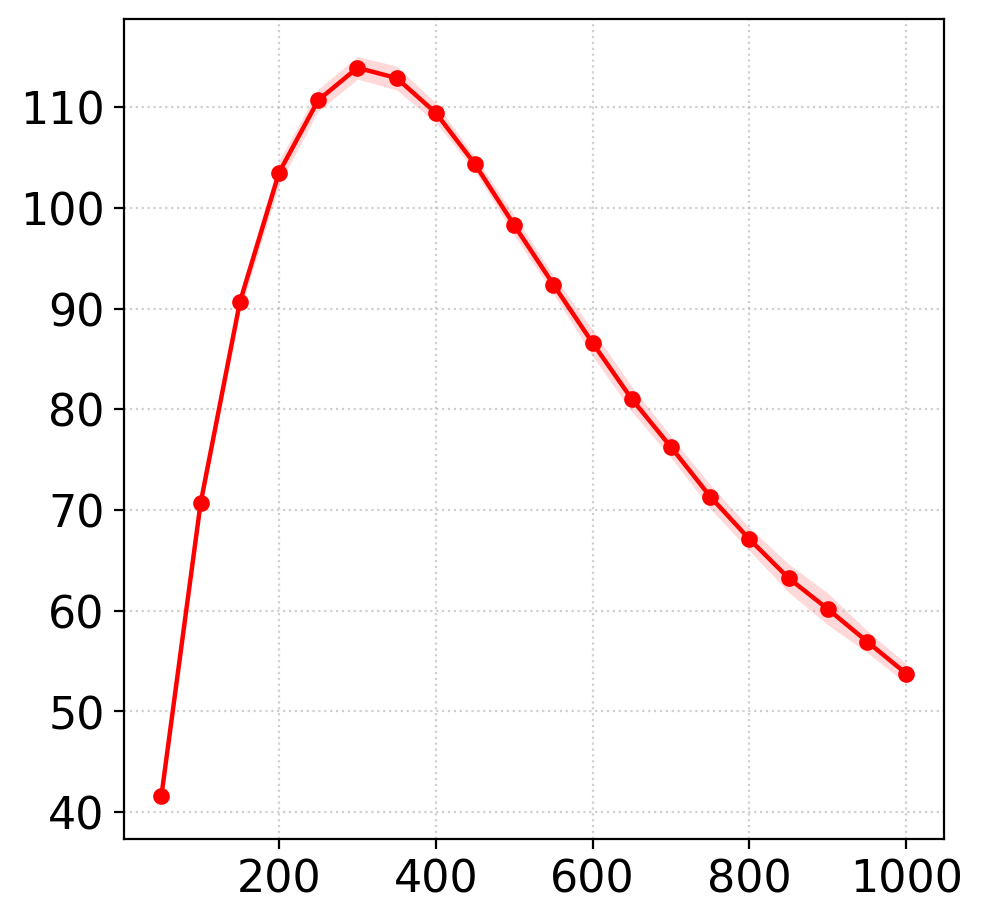}
    \\
    \ylabelmeangroupsize &
    \plotentry{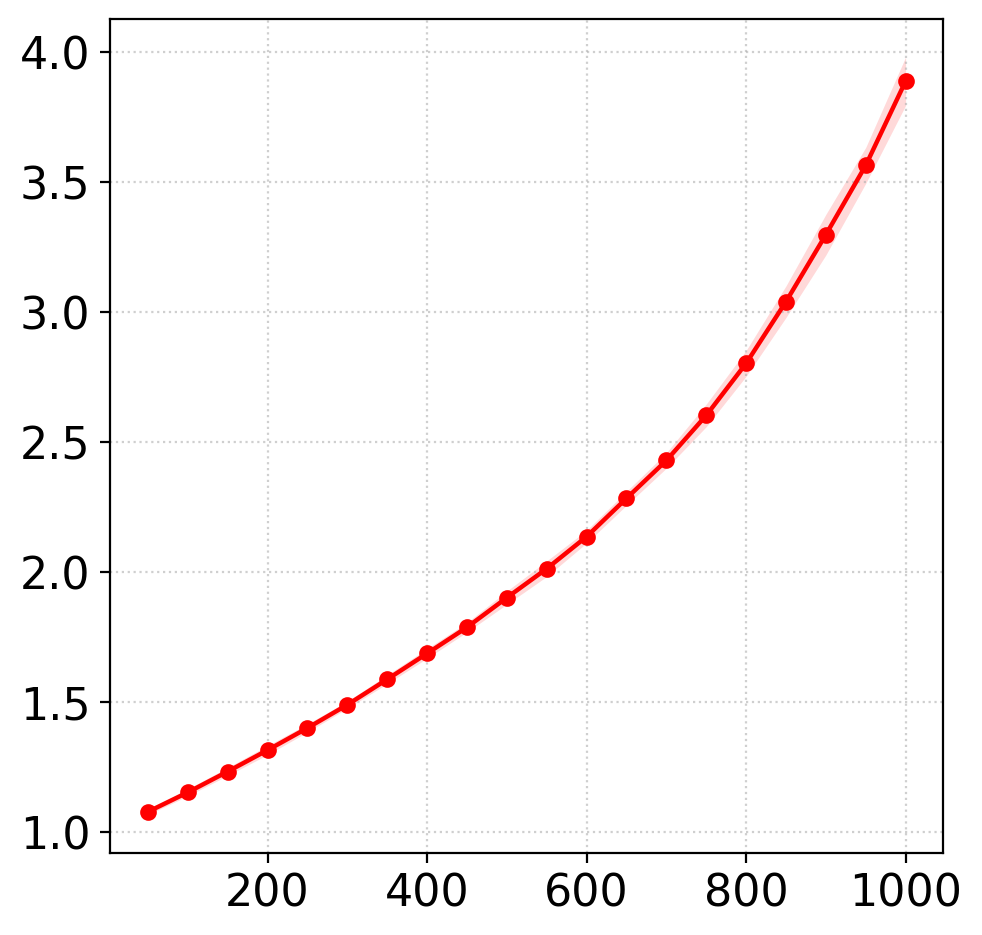} &
    \plotentry{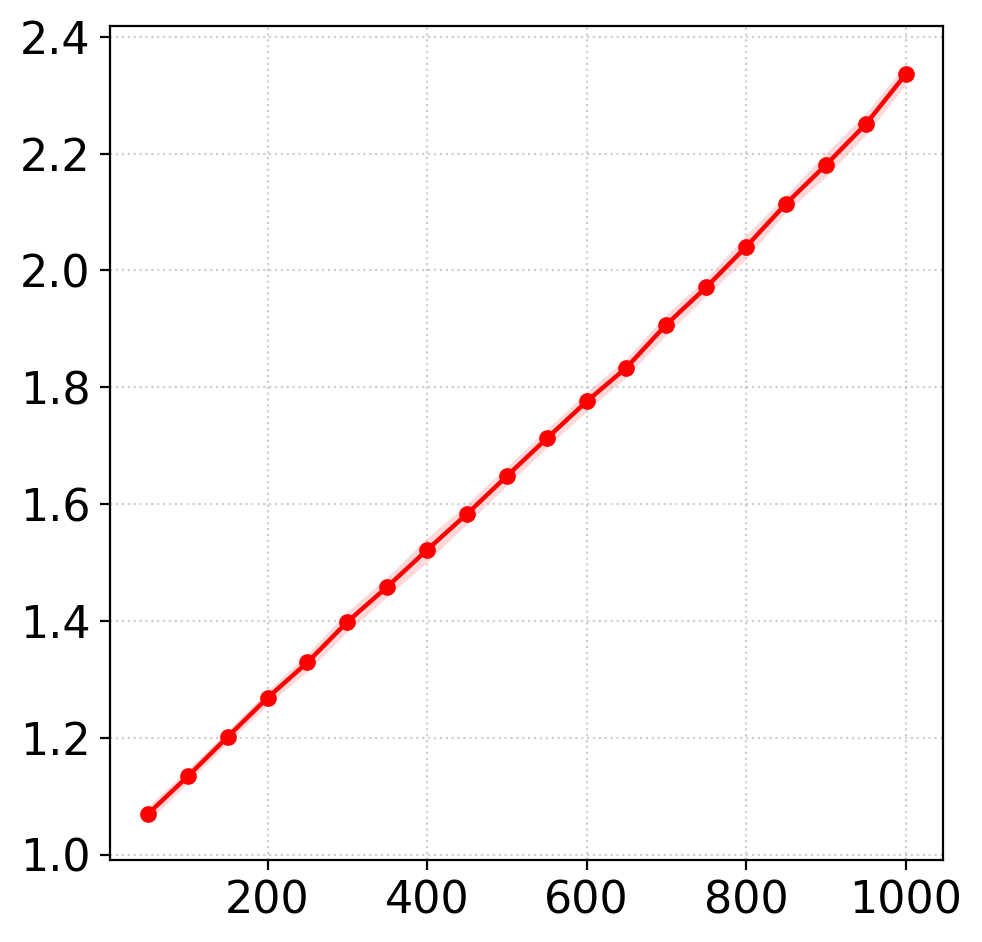} &
    \plotentry{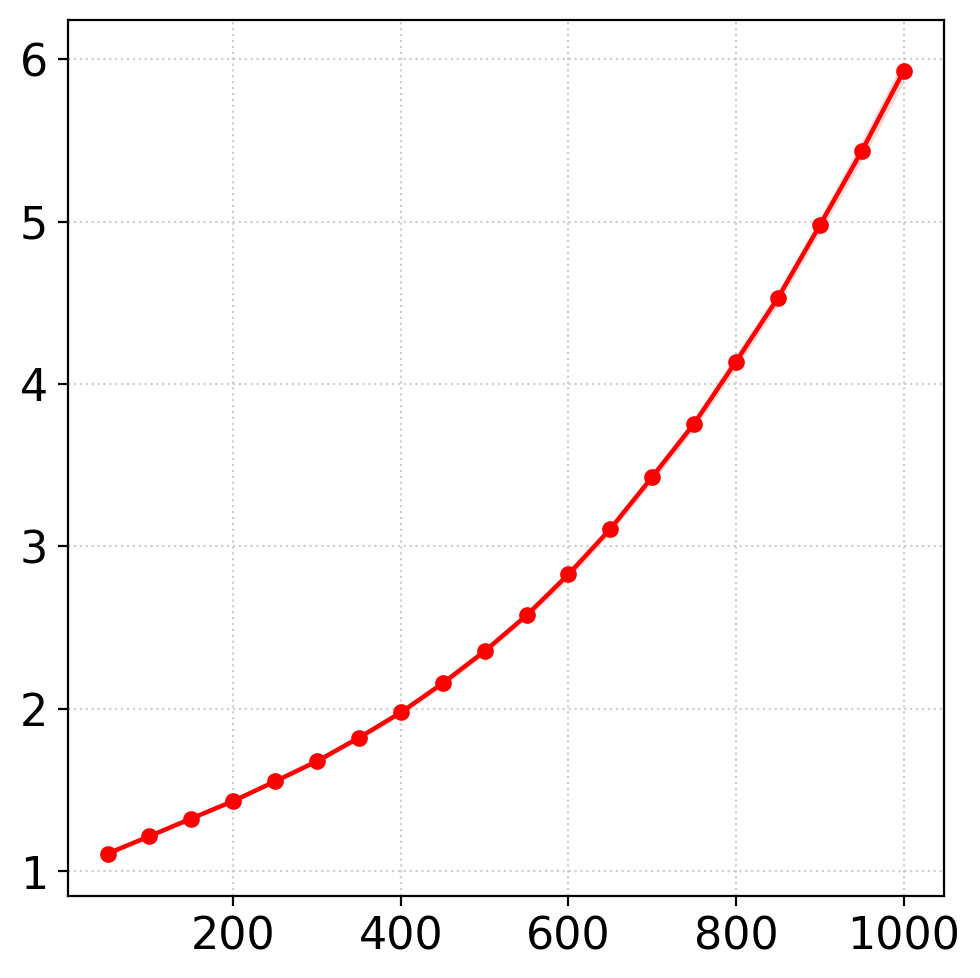} &
    \plotentry{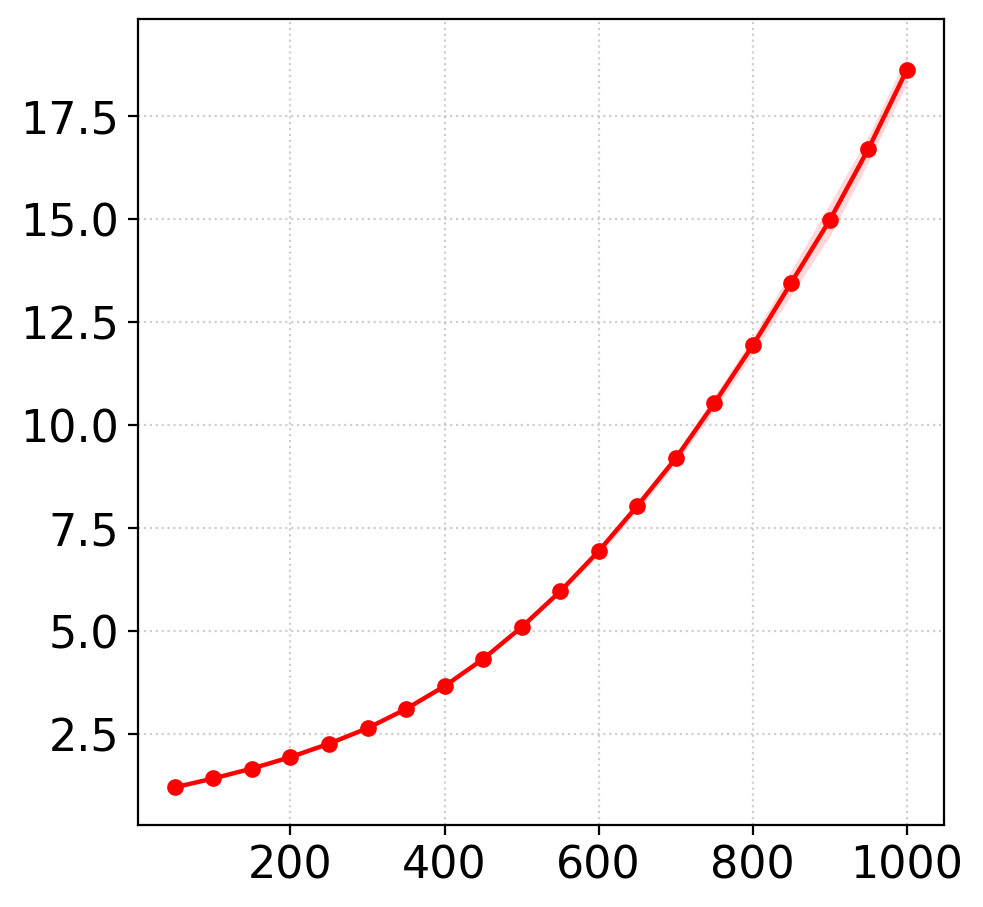}
    \\
    \ylabelmaxgroupsize &
    \plotentry{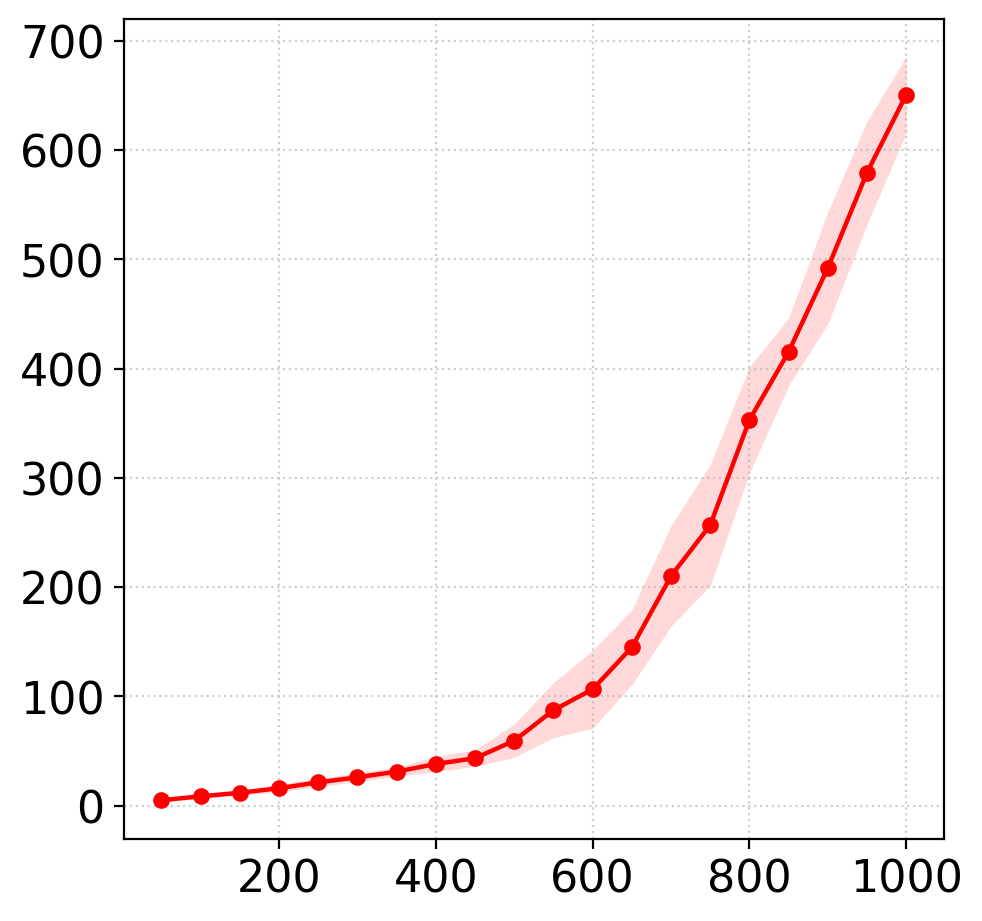} &
    \plotentry{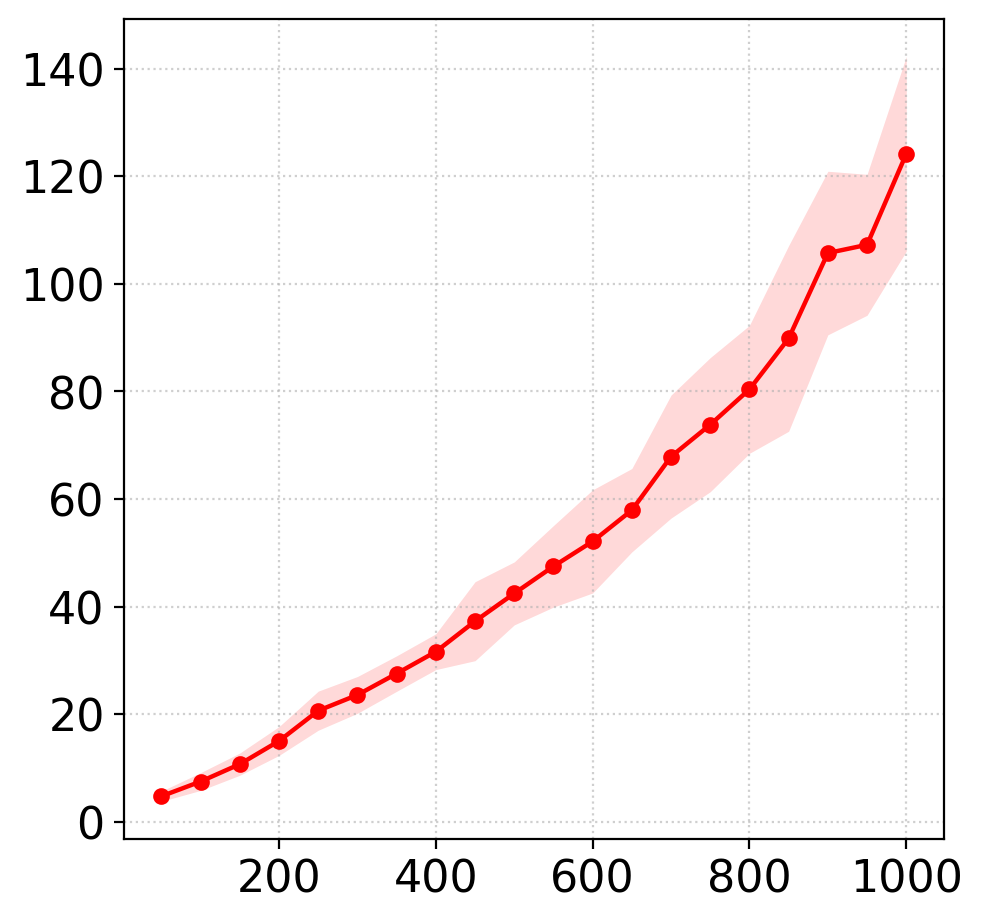} &
    \plotentry{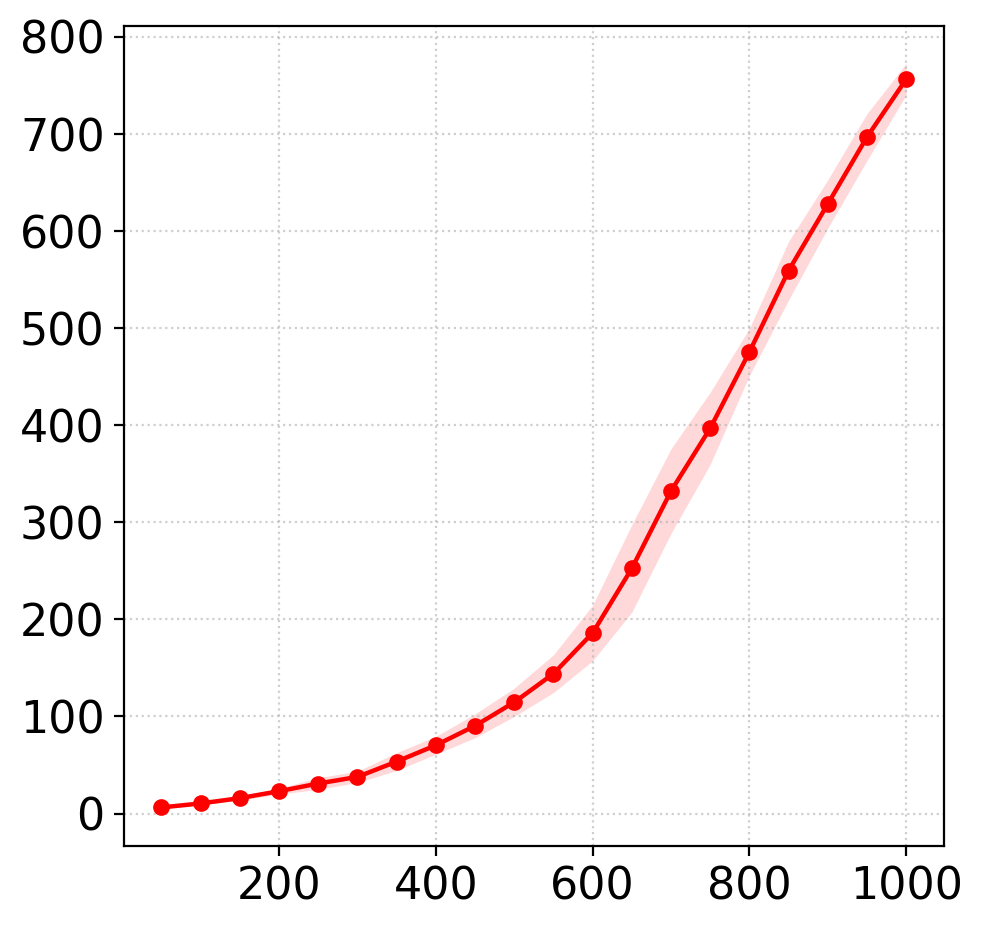} &
    \plotentry{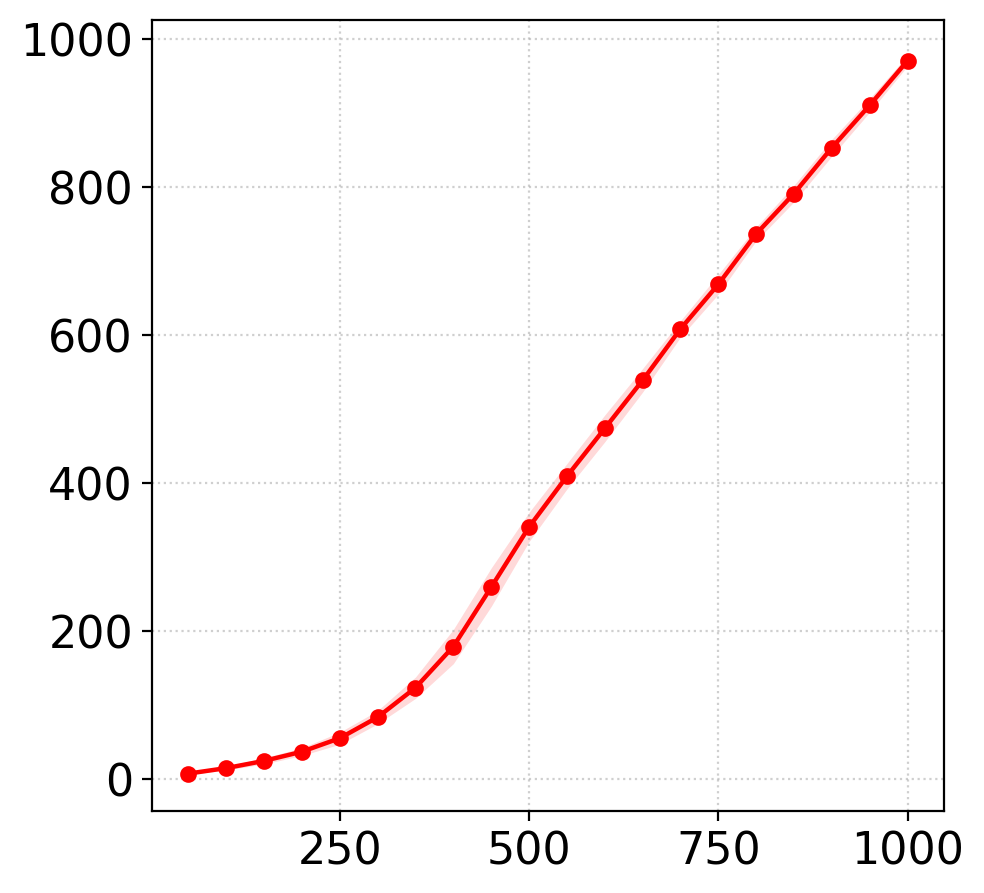}
    \\
    & \multicolumn{4}{c}{\small agents}
    \\[0.2em]
  \end{tabular}
  \caption{MAPD group statistics for GD-RHCR. The color coding is \textcolor{red}{GD-RHCR}}
  \label{mapd_group}
\end{figure*}

\begin{figure*}[!ht]
  \centering
  \setlength{\tabcolsep}{1pt}
  \setlength{\plotwidth}{0.16\linewidth}
  \setlength{\plotheight}{2.8cm}
  \begin{tabular}{@{}c@{}cccccc@{}}
    \mapentry{maps/random-64-64-20_hotspots_full.png}{random-64-64-20} &
    \mapentry{maps/room-64-64-uniform_1_hotspots_full.png}{room-64-64-var1} &
    \mapentry{maps/room-64-64-uniform_2_hotspots_full.png}{room-64-64-var2}&
    \mapentry{maps/random-32-32-20_hotspots_full}{random-32-32-20} &
    \mapentry{maps/room-32-32-uniform_1_hotspots_full.png}{room-32-32-var1} &
    \mapentry{maps/empty-48-48_hotspots_full.png}{empty-48-48} &
    \\[0.8em]
    \ylabelnumgroups
    \plotentry{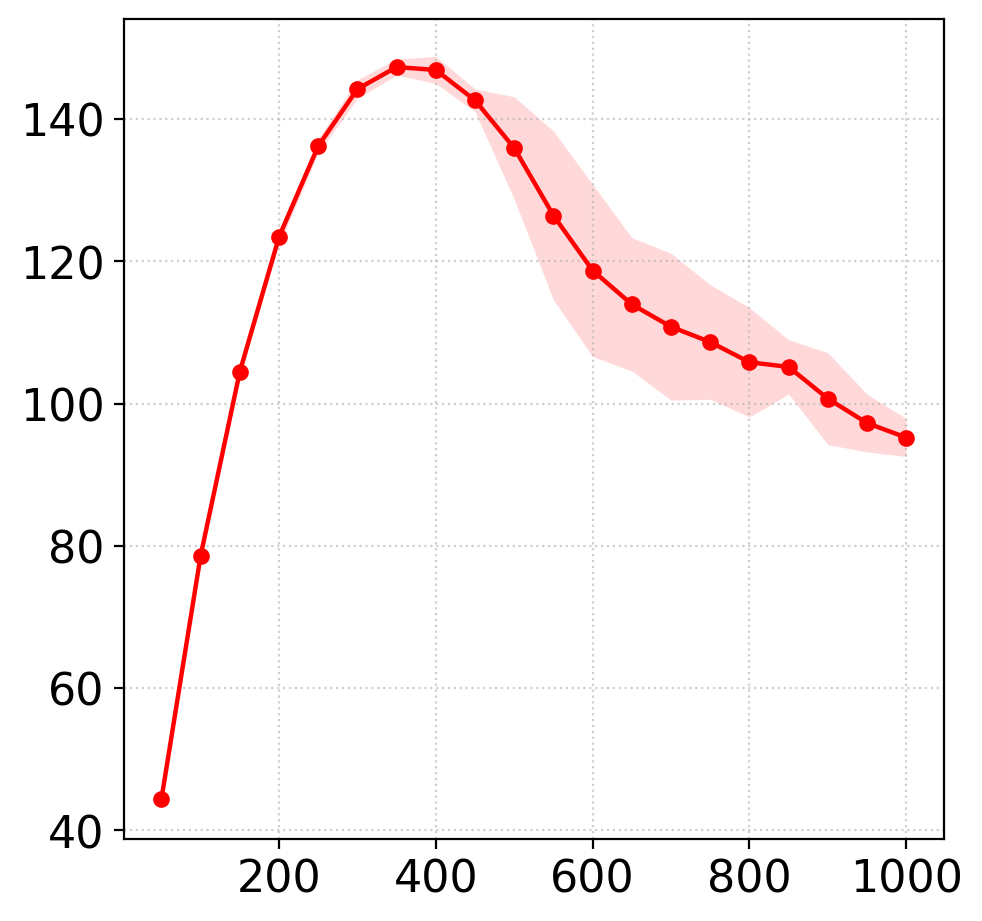} &
    \plotentry{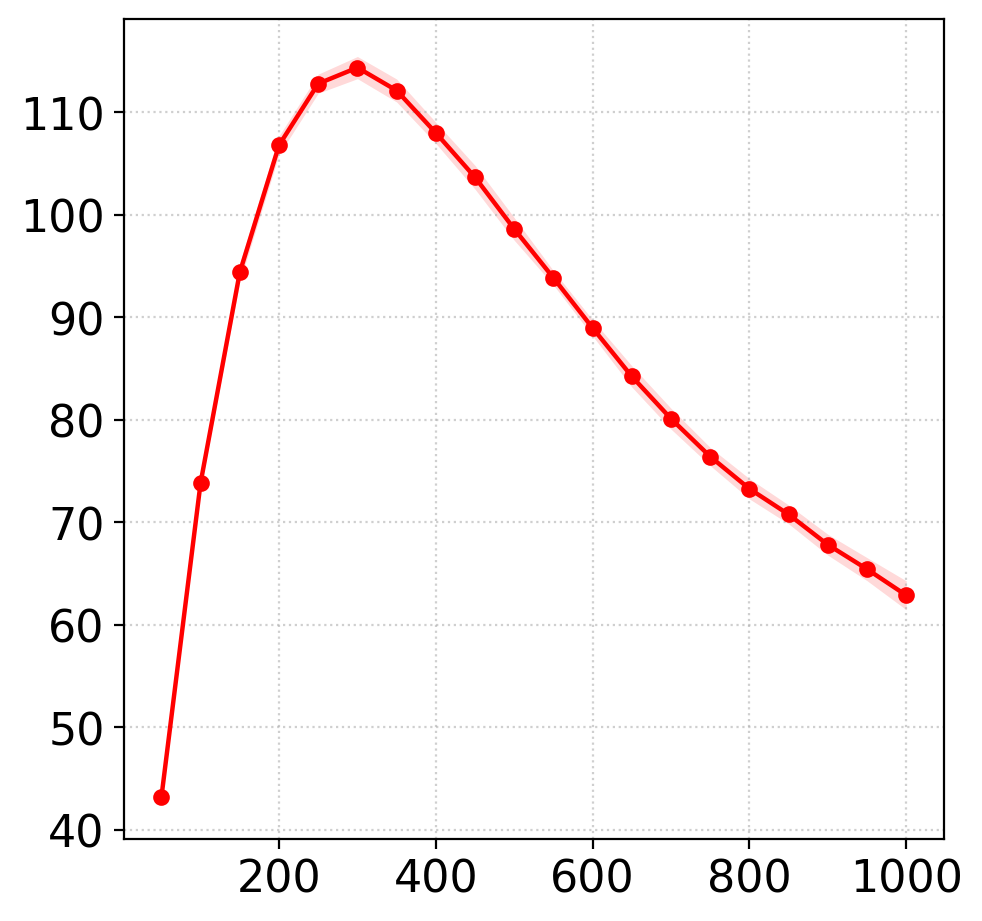}&
    \plotentry{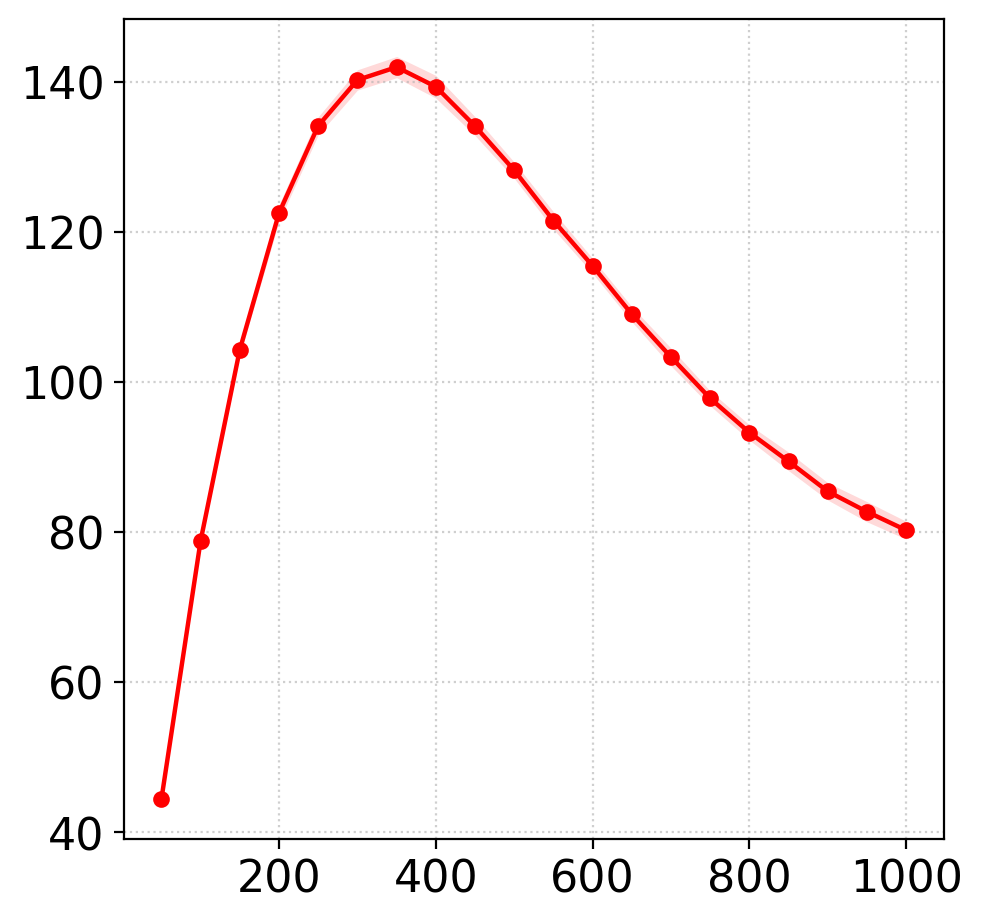}&
    \plotentry{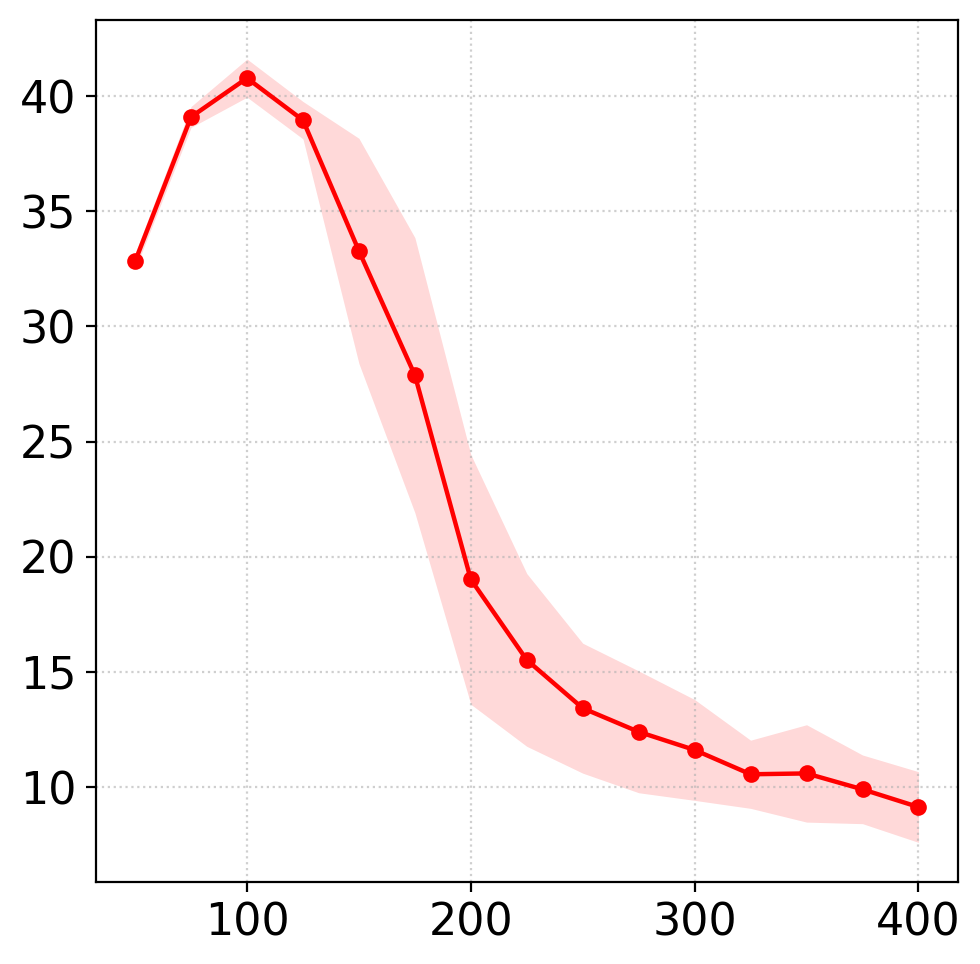} &
    \plotentry{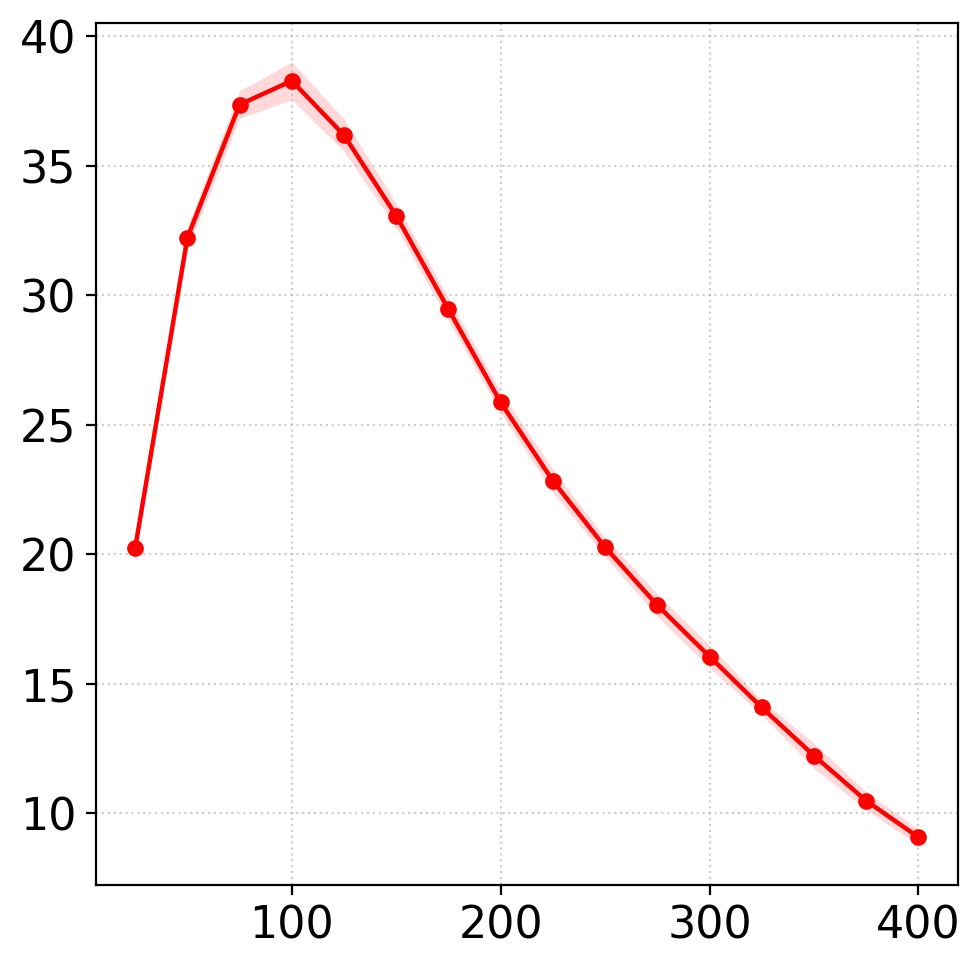} &
    \plotentry{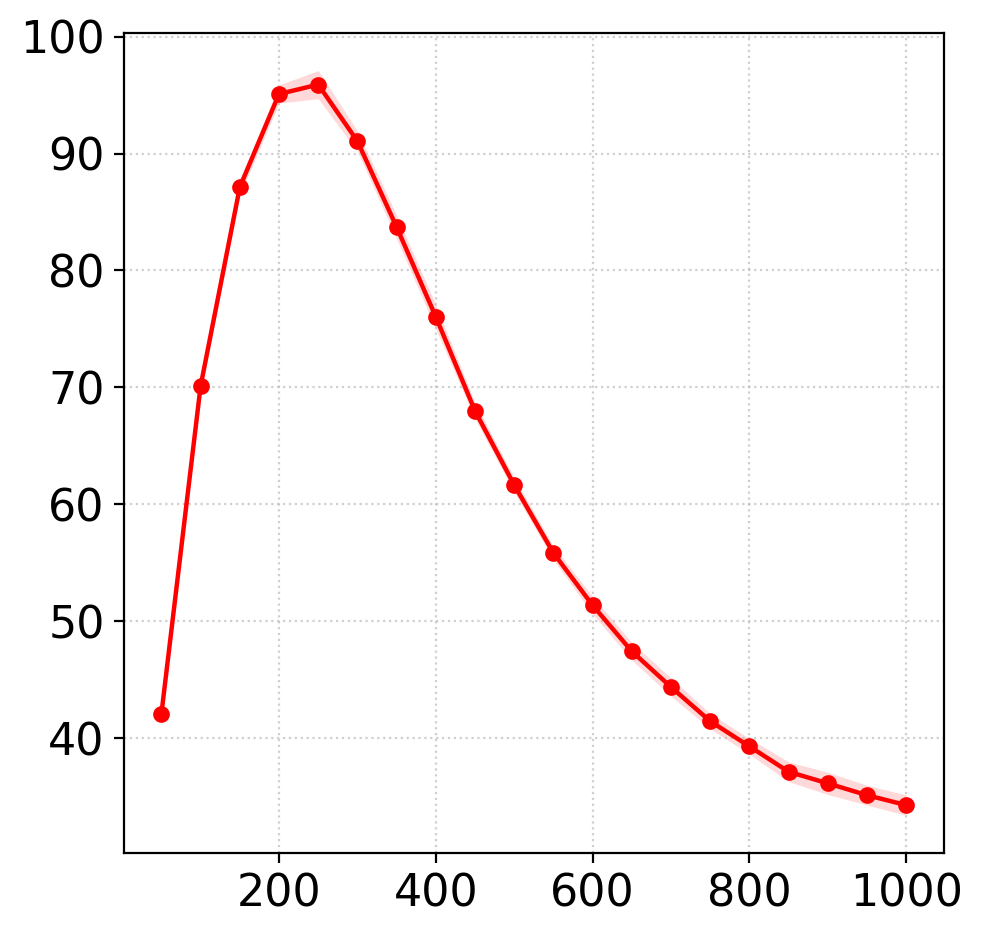} &
    \\
    \ylabelmeangroupsize
    \plotentry{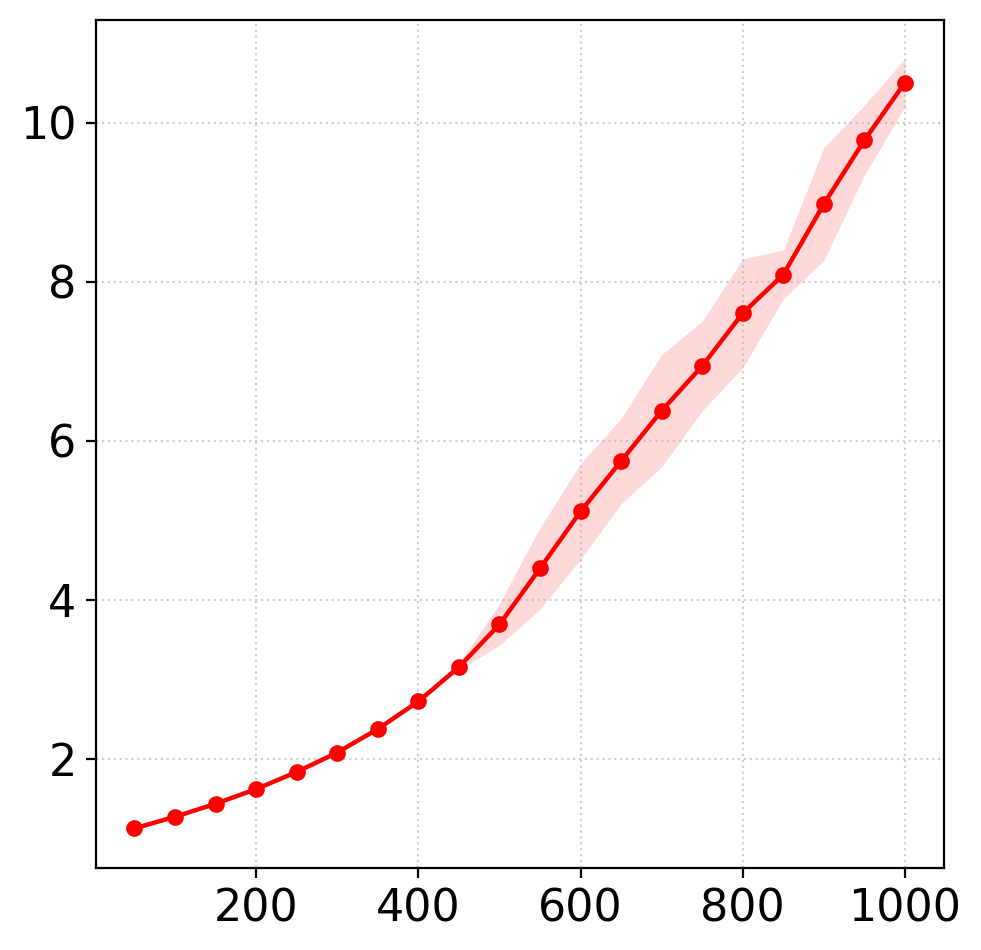} &
    \plotentry{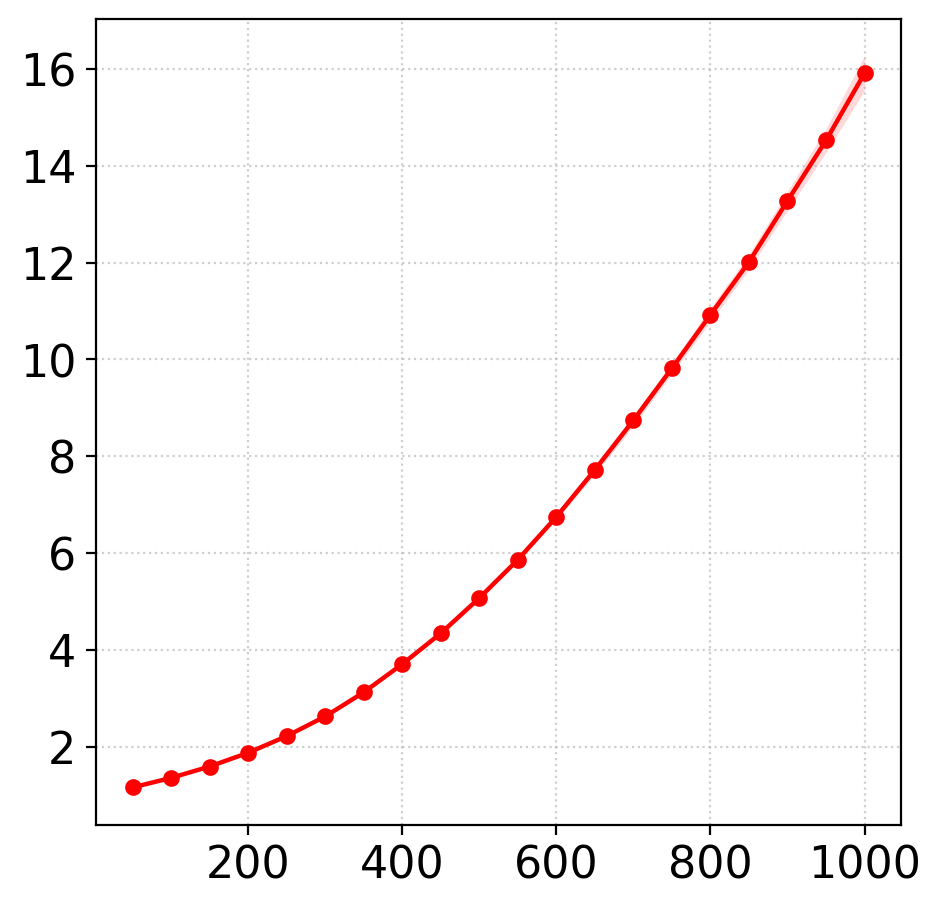}&
    \plotentry{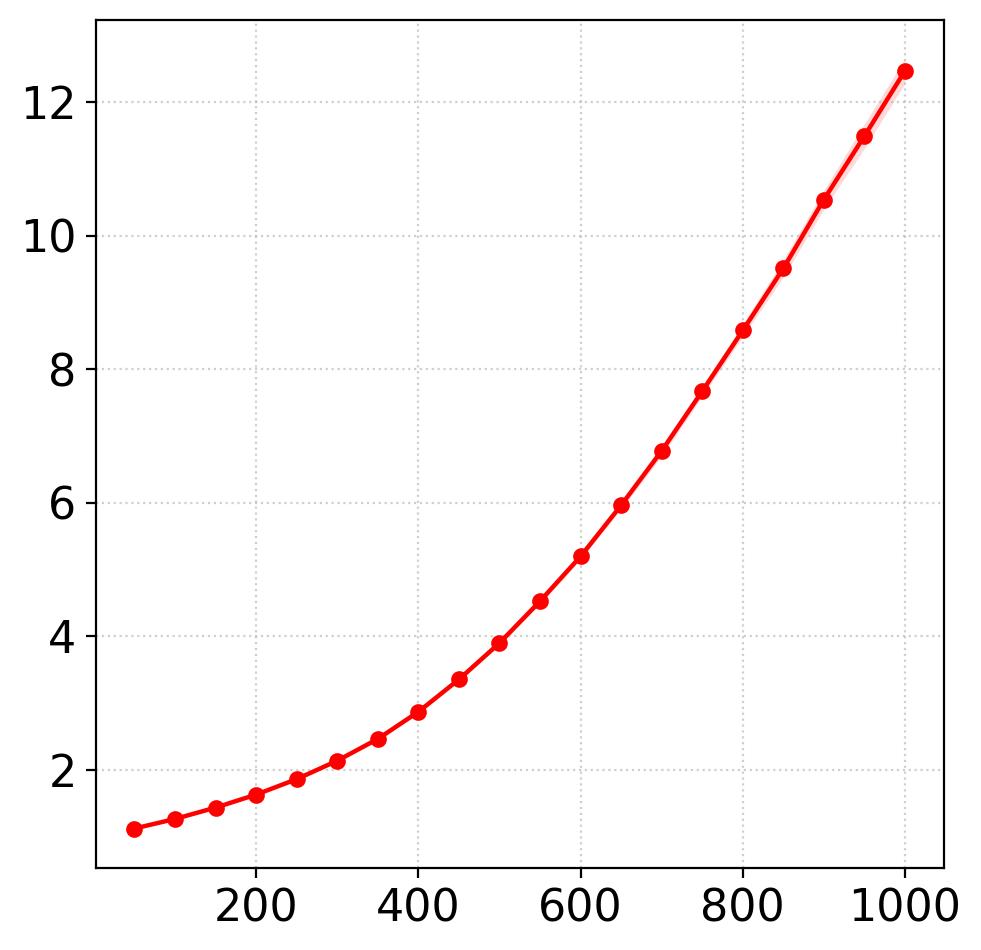}&
    \plotentry{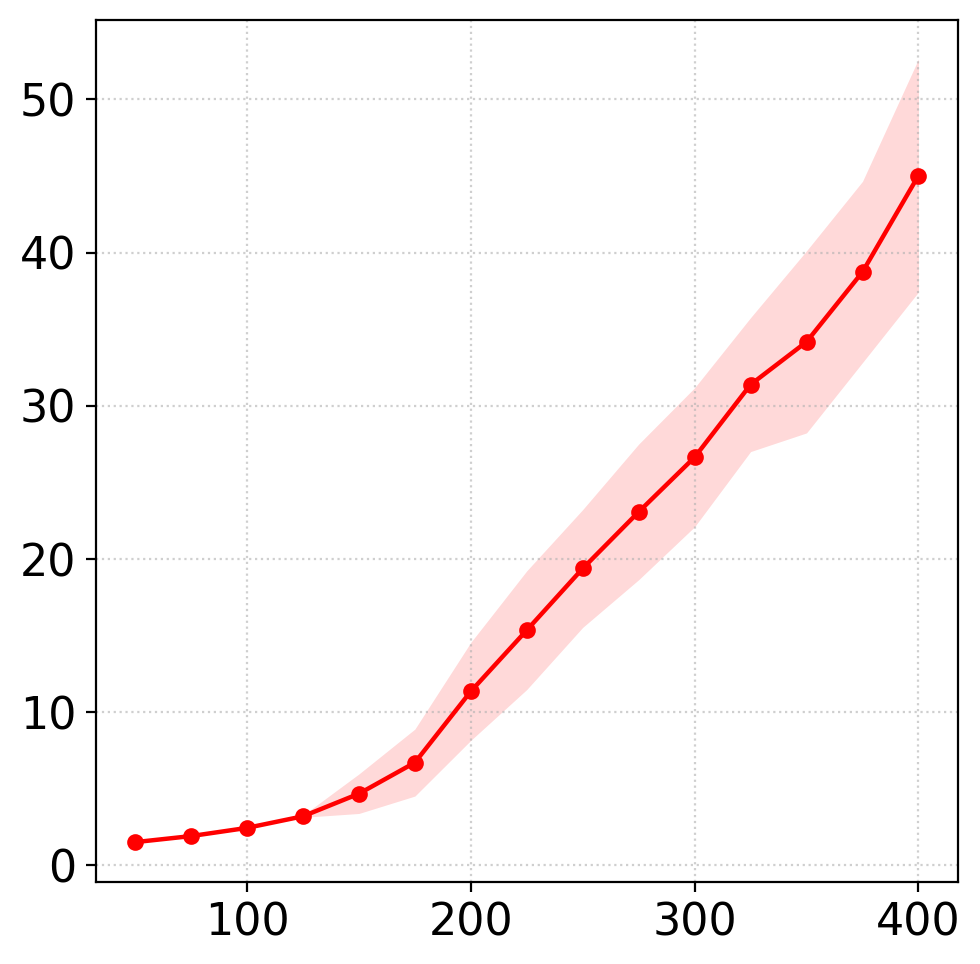} &
    \plotentry{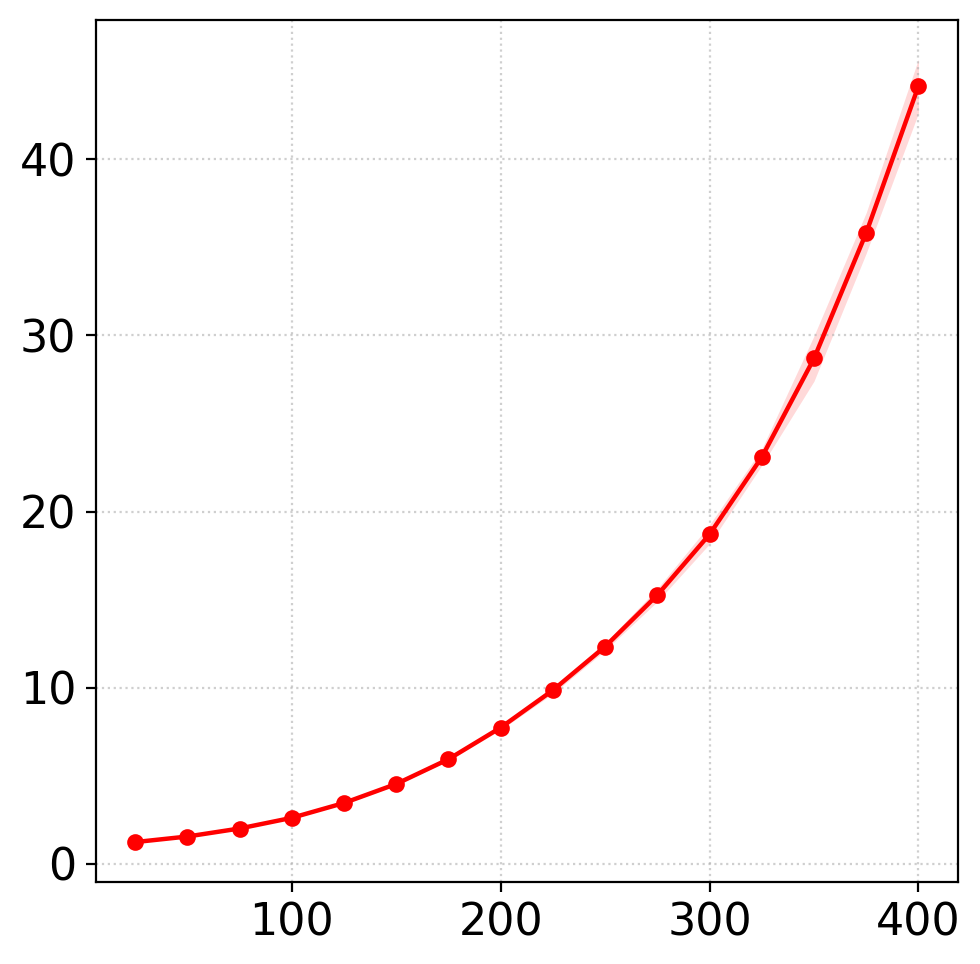} &
    \plotentry{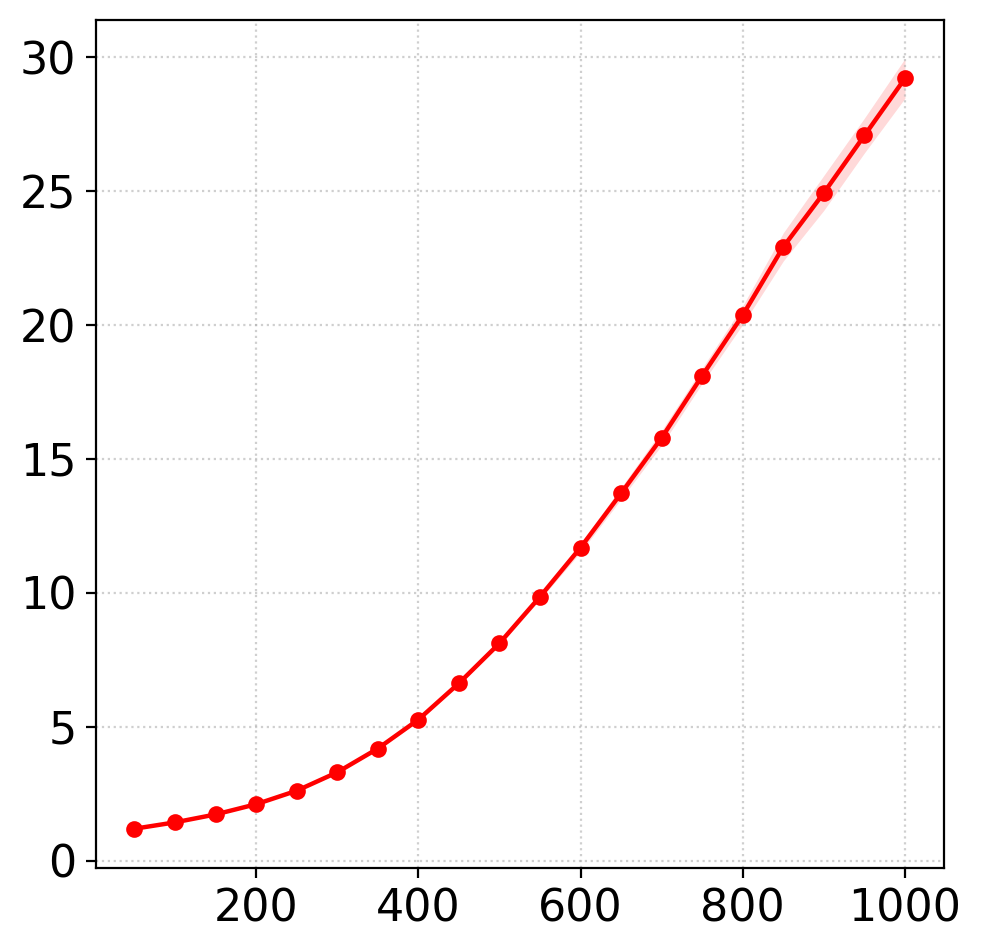} &
    \\
    \ylabelmaxgroupsize
    \plotentry{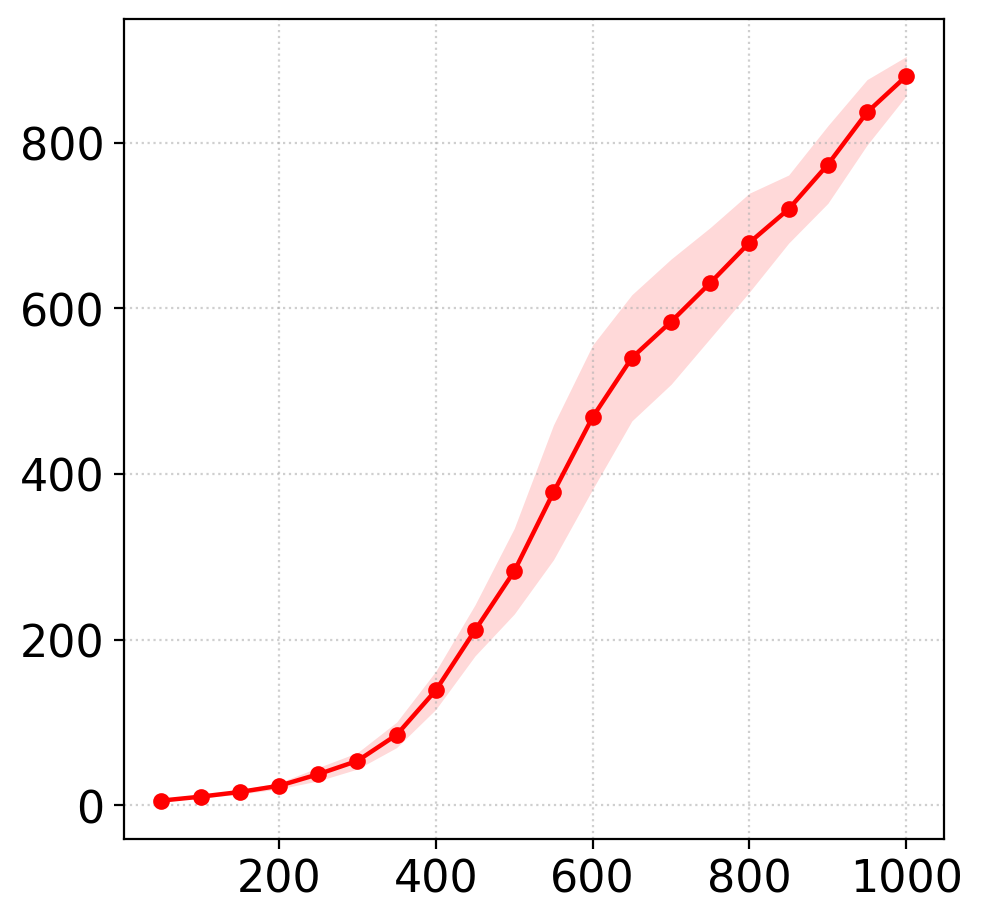} &
    \plotentry{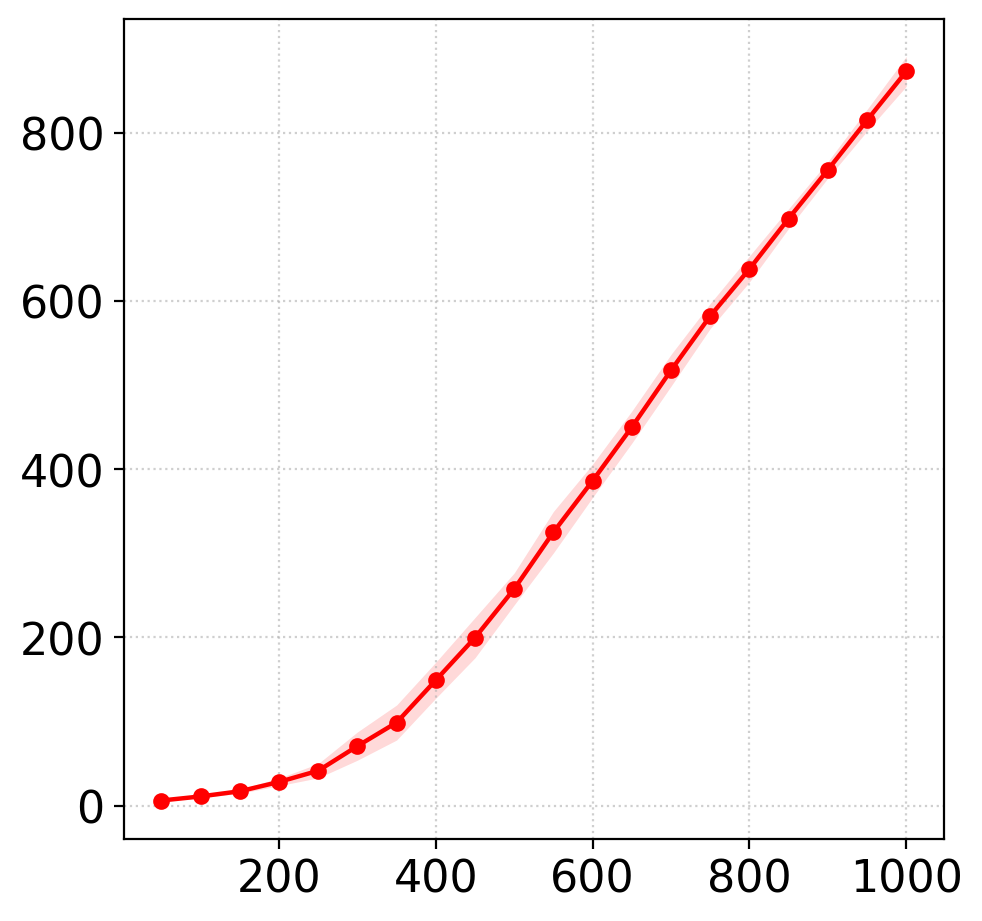}&
    \plotentry{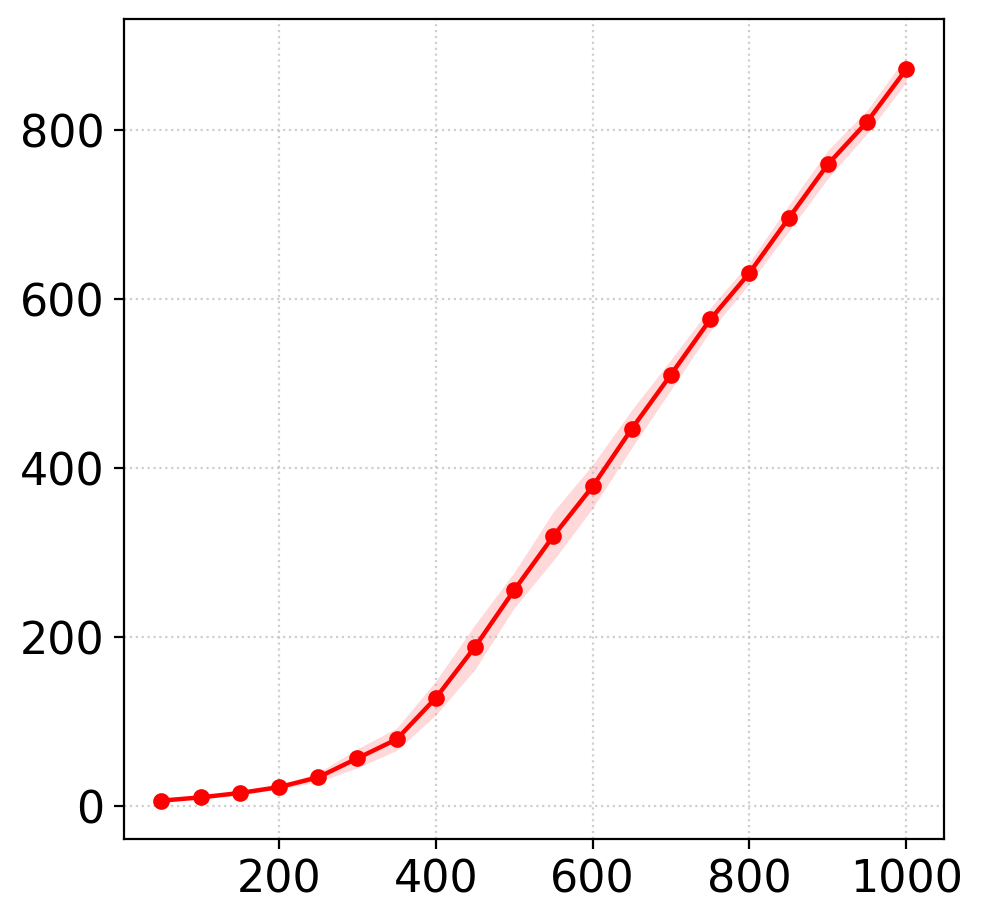}&
    \plotentry{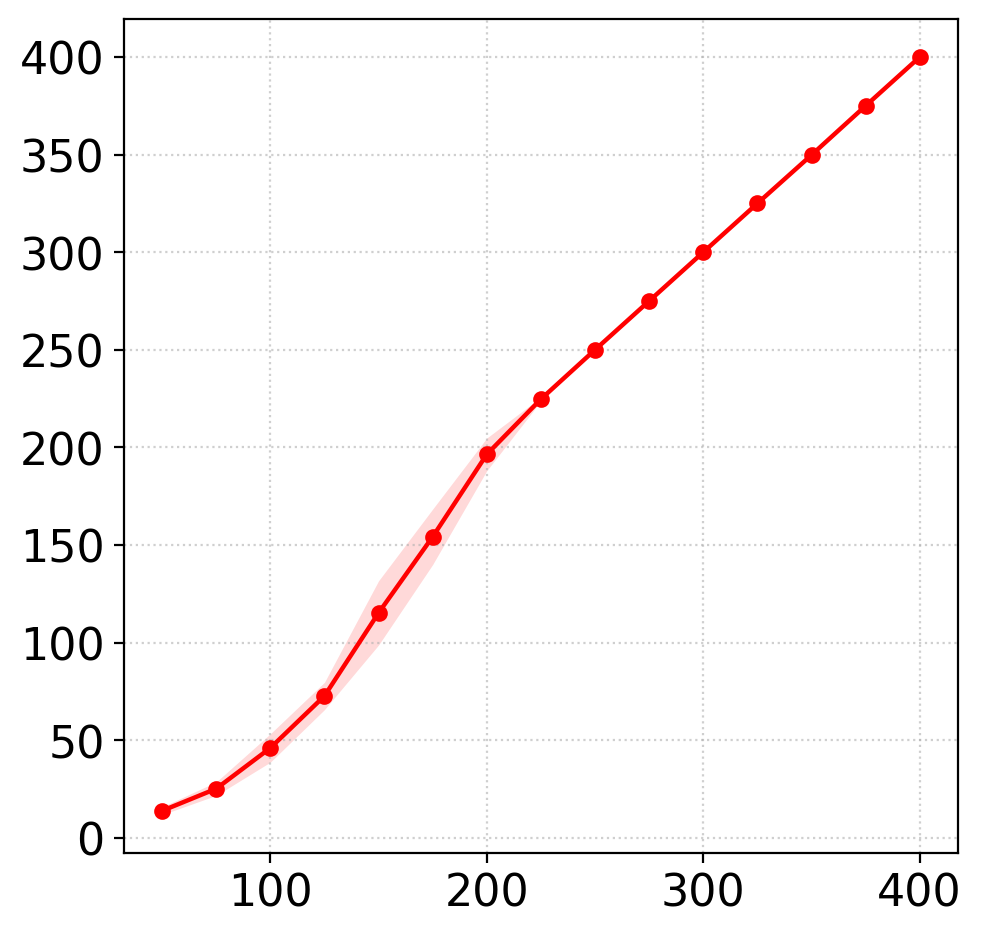} &
    \plotentry{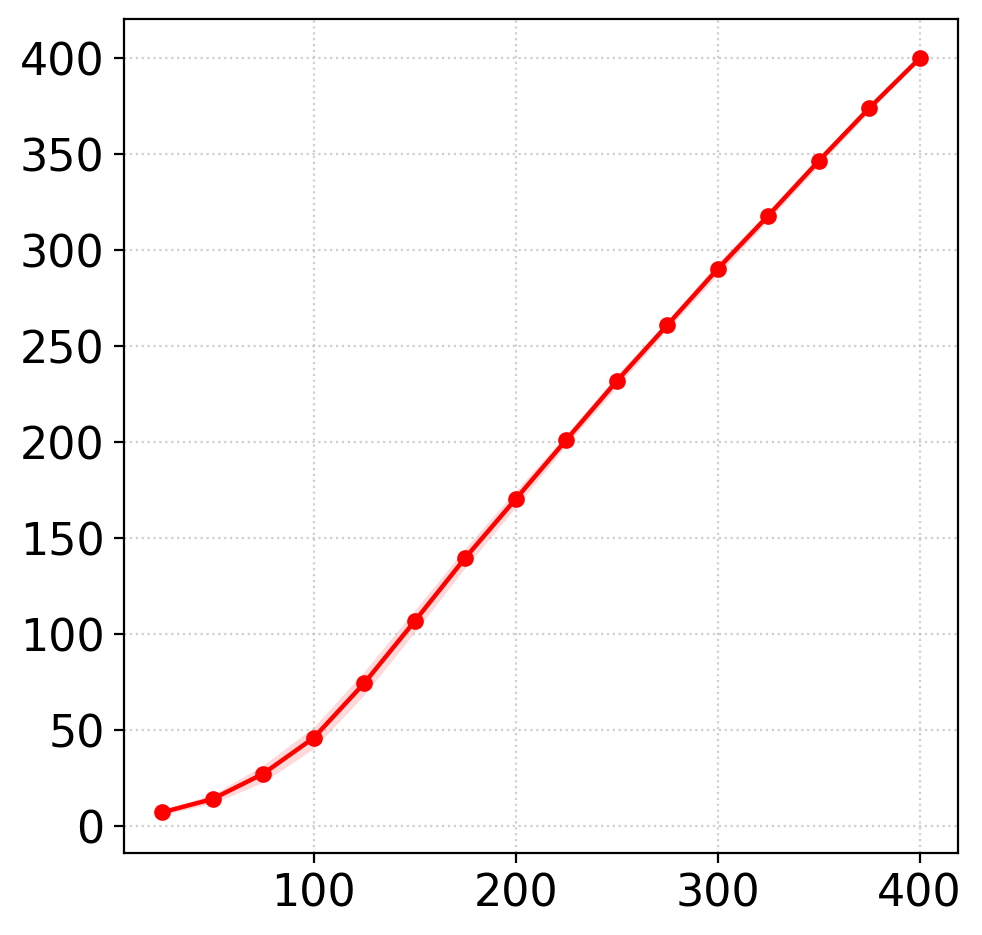} &
    \plotentry{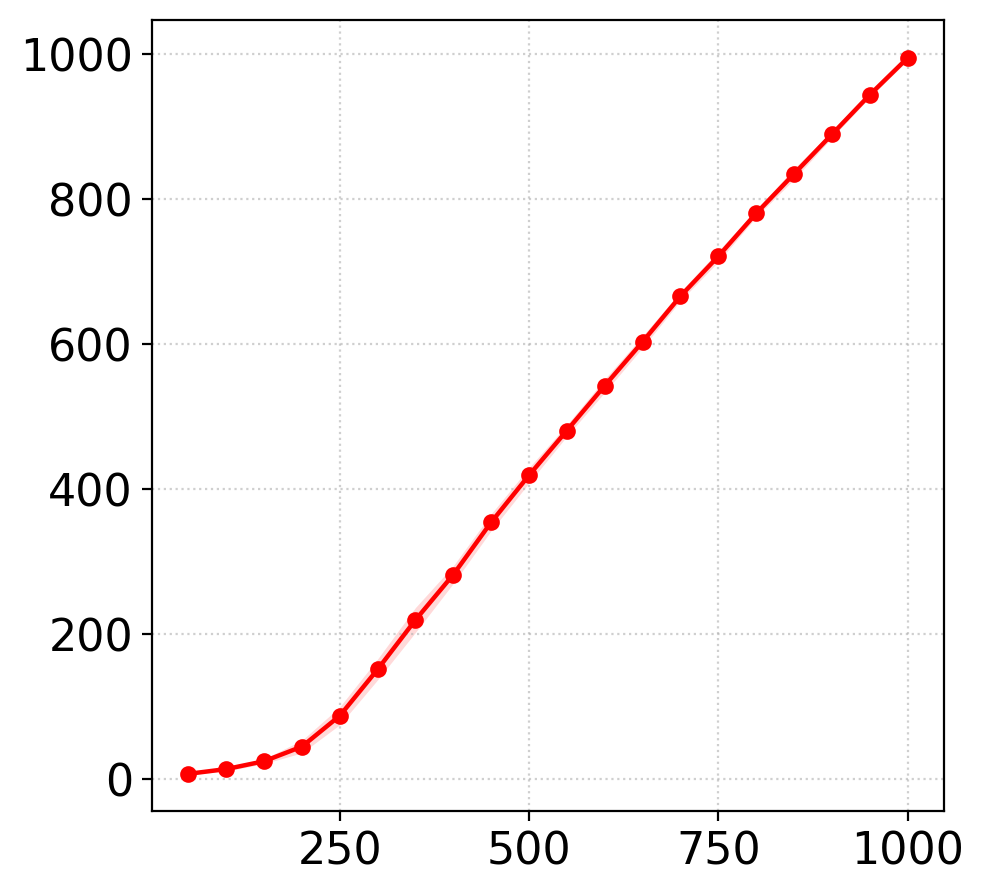} &
    \\
    & \multicolumn{4}{c}{agents}
  \end{tabular}
  \caption{Random navigation group statistics for GD-RHCR. The color coding is \textcolor{red}{GD-RHCR}}
  \label{random_groups}
\end{figure*}

\end{document}